\documentclass[11pt]{article}
\usepackage{mathptmx}
\usepackage{amsmath,amssymb,amsthm}
\usepackage{authblk}
\usepackage[final]{graphicx}
\usepackage{subfig}
\usepackage{color}
\usepackage[margin=0.75in]{geometry}
\usepackage{chngcntr}
\usepackage{indentfirst}
\usepackage{enumitem}
\usepackage{hyperref}
\usepackage[abspath]{currfile}
\usepackage{float}
\usepackage{bbm}
\usepackage{natbib}
\usepackage{accents}
\usepackage{siunitx}
\usepackage{mathrsfs}
\usepackage{xcolor}
\usepackage[linesnumbered,ruled,vlined]{algorithm2e}

\SetCommentSty{mycommfont}

\SetKwInput{KwInput}{Input}                
\SetKwInput{KwOutput}{Output}              

\usepackage{optidef}
\usepackage{cleveref}
\usepackage{enumitem}

\usepackage{apxproof}

\DeclareMathOperator*{\argmin}{arg\,min}

\usepackage{tikz}
\usetikzlibrary{matrix,positioning}
\tikzset{bullet/.style={circle,fill,inner sep=2pt}}
\usepackage{adjustbox}
\usepackage{pgfplots}

\title{Robust Arbitrage Conditions for Financial Markets}
\author{Derek Singh, \, Shuzhong Zhang}
\affil{Department of Industrial and Systems Engineering, University of Minnesota \\ singh644@umn.edu, \, zhangs@umn.edu}
\date{\vspace{-5ex}}

\newcommand{\sgn}{\operatorname{sgn}}
\newcommand{\quotes}[1]{``#1''}

\theoremstyle{case}
\newtheorem{case}{Case}

\newtheorem{remark}{Remark}
\newtheoremrep{prop}{Proposition}[section]
\newtheoremrep{theorem}{Theorem}[section]
\newtheoremrep{lemma}[theorem]{Lemma}
\newtheoremrep{conj}[theorem]{Conjecture}
\newtheoremrep{corollary}{Corollary}[theorem]
\newtheorem{corollaryP}{Corollary}[prop]
\newtheorem*{theorem*}{Theorem}

\providecommand{\keywords}[1]
{
  \small	
  \textbf{\textit{Keywords---}} #1
}

\usepackage[section]{placeins}

\DeclareMathAlphabet{\mathcal}{OMS}{cmsy}{m}{n}

\begin{document}

\maketitle

\begin{abstract}
This paper investigates arbitrage properties of financial markets under distributional uncertainty using Wasserstein distance as the ambiguity measure. The weak and strong forms of the classical arbitrage conditions are considered. A relaxation is introduced for which we coin the term statistical arbitrage. The simpler dual formulations of the robust arbitrage conditions are derived. A number of interesting questions arise in this context. One question is: can we compute a critical Wasserstein radius beyond which an arbitrage opportunity exists? What is the shape of the curve mapping the degree of ambiguity to statistical arbitrage levels? Other questions arise regarding the structure of best (worst) case distributions and optimal portfolios. Towards answering these questions, some theory is developed and computational experiments are conducted for specific problem instances. Finally some open questions and suggestions for future research are discussed.
\end{abstract}


\section{Introduction and Overview}
\subsection{The Characterization of Arbitrage in Financial Markets}


\indent Financial arbitrage with respect to securities pricing is a fundamental concept regarding the behavior of financial markets developed by Ross in the 1970s. A couple of his seminal papers include \textit{Return, Risk, and Arbitrage} \citep{ross1973return} and \textit{The Arbitrage Theory of Capital Asset Pricing} \citep{ross1976arbitrage}. In the author's own words the arbitrage model or arbitrage pricing theory (APT) was developed as an alternate approach to the (mean variance) Capital Asset Pricing Model (CAPM) \citep{sharpe1964capital} which was itself an extension of the foundational work on Modern Portfolio Theory by Harry Markowitz \citep{markowitz1952portfolio}. Ross argued that APT imposed less restrictions on the capital markets as did CAPM such as its requirement that the market be in equilibrium and its consideration of (only) a single market risk factor as measured by variance of asset returns. Recall that CAPM uses the security market line to relate the expected return on an asset to its beta or sensitivity to systematic (market) risk. APT, on the other hand, is a multi-factor cross sectional model that explains the expected return on an asset in linear terms of betas to multiple market risk factors that capture systematic risk \citep{ross1973return},  \citep{ross1976arbitrage}. \par 
The motivating idea behind APT is the no-arbitrage principle as characterized by the no-arbitrage conditions. This principle asserts that in a securities market it should not be possible to construct a zero cost portfolio that guarantees per scenario either a riskless profit or no chance of losses, across all possible market scenarios. If this were the case, one would be able to make money from nothing, so to speak. Ross formulates the no-arbitrage conditions and via duality theory of linear programming shows the equivalent existence of a state price vector to recover market prices \citep{ross1973return}. Existing results in the literature \citep{delbaen2006mathematics} have shown the equivalence between the single period and multi period no-arbitrage properties (on a finite probability space). To simplify the analysis, we focus on the discrete single period setting. \par
As a further refinement, the notions of \textit{weak} and \textit{strong} arbitrage were developed. A portfolio $w \in \mathbb{R}^n$ of $n$ market securities is designated a weak arbitrage opportunity if $w \cdot S_0 \leq 0$ but $\Pr(w \cdot S_1 \geq 0) = 1$ \textit{and} $\Pr(w \cdot S_1 > 0) > 0$ for initial asset price vector $S_0$ and time 1 asset price vector $S_1$. Similarly, a portfolio $w \in \mathbb{R}^n$ is designated a strong arbitrage opportunity if $w \cdot S_0 < 0$ but $\Pr(w \cdot S_1 \geq 0) = 1$. In a discrete setting with $s$ market states, given security price vector $p \in \mathbb{R}^n$ and payoff matrix $X \in \mathbb{R}^{n \times s}$, a weak arbitrage opportunity is a portfolio $w \in \mathbb{R}^n$ that satistifes $X^\top w \gneq 0$ and $p^\top w \leq 0$. Similarly, a strong arbitrage opportunity is a portfolio $w \in \mathbb{R}^n$ that satistifes $X^\top w \geq 0$ and $p^\top w < 0$. Note there are cases of weak arbitrage portfolios which are not strong arbitrage portfolios \citep[cf. e.g.][]{leroy2014principles}. \par 

In a discrete setting, the well known Farkas Lemma can be used to characterize the property of (weak) strong arbitrage. The Farkas Lemma characterization says that security price vectors $p$ exclude (weak) strong arbitrage iff given payoff matrix $X$ (across all market scenarios) there exists a (strictly) positive solution $q$ to $p = X q$. The normalized state price vectors $q^*_s = q_s / \sum_s q_s$ \textit{become} the set of discrete risk neutral probabilities that defines the measure $Q$ \citep[cf. e.g.][]{leroy2014principles}.  The fundamental theorem of asset pricing (also: of arbitrage, of finance) equates the non-existence of arbitrage opportunities in a financial market to the existence of a risk neutral (or martingale) probability measure $Q$ which can be used to compute the fair market value of all assets. A financial market is said to be complete if such a measure $Q$ is unique \citep[cf. e.g.][]{follmer2011stochastic}. The unique measure $Q$ is frequently used in mathematical finance and the pricing of derivative securities in particular, in both discrete time \citep{shreve2005stochastic} and continuous time settings \citep{shreve2004stochastic}. \par

In the context of distributional uncertainty, a natural question arises as to how to characterize the notion of arbitrage. One would presumably seek a balance of generality and practicality in developing a framework to study the arbitrage properties. Some structure is needed to develop intuition and understanding. On the other hand, too much structure could be restrictive and limit useful degrees of freedom. The approach taken in this line of research is to start from the fundamental (weak and strong) no-arbitrage conditions and investigate how the market model transitions from one of no-arbitrage to arbitrage or vice versa. Distributional uncertainty is characterized via the Wasserstein metric for a couple reasons. The Wasserstein metric is a (reasonably) well understood metric and a natural, intuitive way to compare two probability distributions using ideas of transport cost. It is also a flexible approach that encompasses parametric and non-parametric distributions of either discrete or continuous form. Furthermore, recent duality results and structural results on the worst case distributions could help us to understand and/or quantify the market model transitions as well as measure (in a relative sense) the degree of arbitrage or no-arbitrage inherent to a given market model. \par

Logical reasoning dictates that it should be possible to distort a no-arbitrage measure into an arbitrage admissible measure. For a simple discrete example, consider a one-period binomial tree of stock prices where $0 < S_d < 1+r < S_u$, $p_u + p_d = 1$, $p_u > 0 \implies p_d > 0$ are the conditions that characterize an arbitrage-free market \citep{shreve2005stochastic}. If we now distort the above $Q$ measure into a $P$ measure such that $p_d = 0$, it is clear to see that a zero cost portfolio that is long the stock and short a riskless bond will make profit with probability 1. So then, how ``far" is this distorted measure $P$ from the original no-arbitrage measure $Q$? Can we safeguard ourselves within a ball of (only) arbitrage-free probability measures $Q'$ of distance at most $\delta$ from the reference measure $Q$?  What is the structure of the worst case distributions and optimal portfolios within this ball? Is there a critical radius $\delta^*$ for this ball of arbitrage-free measures \textit{beyond} which an arbitrage admissible measure is sure to exist? Alternatively, suppose the reference measure $Q$ admitted arbitrage. What is the nearest arbitrage-free measure to this measure? Is that minimal distance, call it $\delta^*_g$, computable? These questions are the motivation for the line of research conducted in this paper. As mentioned above, this research uses the Wasserstein distance metric \citep[cf. e.g.][]{Villani08}. To the best of our knowledge, this paper is the first to investigate these notions under the Wasserstein metric and develop a mixture of theoretical and computational answers to these questions. \par
The contributions of this paper are as follows. Primal problem formulations for the classical and statistical arbitrage conditions (under distributional uncertainty using Wasserstein ambiguity) are developed. Using recent duality results \citep{Gao16}, \citep{blanchetFirst}, simpler dual formulations that only involve the reference arbitrage-free probability measure are constructed and solved. The $\max$-$\min$ and $\max$-$\max$ dual problems are formulated as nonlinear programming problems (NLPs). The structure of the best (worst) case distributions is analyzed. A formal proof for the NP hardness of the dual no-arbitrage problem is also given. Using this theoretical machinery, the critical radii $\delta^*$, the best (worst) case distributions, and/or optimal portfolios are computed for a few specific problem instances involving real world financial market data. The complementary problem to compute the minimal distance $\delta^*_g$ to an arbitrage-free measure for a reference measure that admits arbitrage is formulated and solved. We make use of the fundamental theorem of asset pricing to do this \citep[cf. e.g.][]{leroy2014principles, follmer2011stochastic}. \par

An outline of this paper is as follows. Section 1 gives an overview of the financial concepts of arbitrage and statistical arbitrage as well as a literature review. Section 2 develops the main theoretical results to characterize arbitrage under distributional uncertainty using Wasserstein distance. Section 3 extends this machinery to cover the notion of statistical arbitrage. Section 4 presents applications of the theory developed in Sections 2 and 3. Section 5 gives formal proofs for the NP hardness of the no-arbitrage problem. Section 6 is a computational study of the arbitrage properties for a few specific problem instances and computes numerical solutions. Section 7 discusses conclusions and suggestions for further research. \par

\subsection{The Characterization of Statistical Arbitrage in Financial Markets}
Statistical arbitrage denotes a class of data driven quantitative trading and algorithmic investment strategies, for a set of securities, to exploit deviations in relative market prices from their \quotes{true} distributions. Classical notions of statistical arbitrage opportunites involve estimation and use of statistical time series models (such as cointegration or kalman filter) to describe structural properties of asset prices such as mean reversion, volatility, etc.\ and help identify temporal deviations in market prices that present trading and/or investment opportunites before the market ``reverts" to its equilibrium behavior \citep{focardi2016new}. One particular sub-class of such strategies that is prevalent in both the literature and industry practice is known as pairs trading. The canonical example here is the coke vs.\ pepsi trade where one identifies a price dislocation and then simultaneously shorts the over-priced asset and buys the under-priced asset and waits for the relative prices to restore to equilibrium, and closes out the position, thus realizing a profit for the arbitrageur \citep{krauss2017statistical}. \par
Practitioners, such as investment banks and hedge funds, employ a wide array of professionals to work in multiple aspects of this: such as trading systems design and technology support, data collection, model development, trade execution, risk management, reporting, business development, and so on. The actual practice of statistical arbitrage typically involves a mixture of art and science. The science component is reflected through the estimation and use of statistical time series models and incorporation of emerging trends in the academic literature and technology (for the practical aspects of trade execution and risk management). The art component is reflected through incorporation of investment professionals knowledge, experience, and beliefs about financial markets' current state and future outlook \citep{lazzarino2018statistical}. \par
Classical notions of statistical arbitrage ``already" have an intrinsic notion of variability, hence their name. The motivation for the line of research in this paper is to extend this notion to incorporate distributional uncertainty within the framework of Wasserstein distance and the corresponding duality results. In this sense, the objectives are analagous, with the topic of focus shifted from classical arbitrage to statistical arbitrage. The first steps are to define notions of statistical arbitrage and robust statistical arbitrage and characterize their meaning. A survey of the literature reveals that no universal definition of statistical arbitrage currently exists \citep{lazzarino2018statistical}. With that in hand, next steps are to quantify the best case $(\alpha^{bc})$ and worst case $(\alpha^{wc})$ levels of statistical arbitrage as a function of the degree of distributional uncertainty, as represented by the radius $\delta$ of the Wasserstein ball. A related, complementary, problem is how to find the nearest probability measure (to the original, reference measure) that guards against statistical arbitrage of level $\alpha$ close to 1. \par

\subsection{Literature Review}
In conducting the literature review for this research, not many references were found that have investigated the topic of arbitrage under distributional uncertainty. From Section 1.1 above, one can see that considerable research has been done in academic circles regarding the classical notions of arbitrage in financial markets. Indeed, several academic papers and financial textbooks have been written that cover these topics from their origin in the 1970s until today. It was surprising to us, at least, to find only a few papers that address and/or extend the classical notions of arbitrage under the presence of some form of distributional uncertainty. This subsection gives an overview of what we found in the academic literature. \par

An earlier paper by \citet{jeyakumar2011robust} took a Farkas Lemma approach to describe linear systems subject to data uncertainty in the form of bounded uncertainty sets. The authors develop a notion of a robust Farkas Lemma in terms of the closure of a convex cone they call the robust characteristic cone. As an application of the lemma, they characterize robust solutions of conic linear programs with data contained in closed convex uncertainty sets. Recently \citet{dinh2017robust} applied the robust Farkas Lemma approach to characterize weakly minimal elements of multi-objective optimization problems with uncertain constraints. Note that weakly minimal elements correspond to the notion of optimal solution in the scalar (singleton vector) case. The authors remark that their results are consistent with existing literature in the scalar case. \par

One seminal paper of note by Ostrovskii used the total variation (TV) metric to characterize a radius $\delta_{TV}$ such that all probability measures $Q'$ within this distance from a weak arbitrage-free reference measure $Q$ are also weak arbitrage-free. The author remarks that $\delta_{TV}$ can be interpreted as the minimal probability of success that a zero cost initial portfolio $w \in \mathbb{R}^n$ achieves positive value $w \cdot S_1$ at time 1. The additional constraint on the selected portfolio $w$ is that it must have a strictly positive probability of profit under the reference measure $P$. This allows $\delta_{TV} > 0$ to hold. This lemma is proven using tools from probability theory and real analysis. The main result relating $\delta_{TV}$ to the minimal probability of success is established via proof by contradiction \citep{ostrovski2013stability}. The bound appears to be tight although this result is not proven in the paper. \par
The author remarks that the probability measures $Q$ and $Q'$ could have different support and/or generate different probability spaces. Furthermore, Ostrovski describes the no-arbitrage conditions and computes the critical radius $\delta_{TV}$ for a one-period binomial and trinomial tree respectively. The conditions for the one-period binomial tree are given in Section 1.1 above. The corresponding radius $\delta_{TV}$ is $\min(p_u,p_d)$. For the one-period trinomial tree, different configurations are possible. Let $q_d, q_m, q_u$ denote the one-period transition probabilities to the down, middle, and up nodes respectively. For the case $S_d < S_m < 1 + r < S_u$ the trinomial tree would allow arbitrage iff $q_d = q_m = 0$ or $q_u = 0$. In the first case, the TV distance between the binomial and trinomial trees would be $\max(1 - p_u, p_d) = \max(p_d,p_d) = p_d$. In the second case it would be $\max(p_u, q_m, |p_d - q_d|) \geq p_u$. Thus the trinomial model would be arbitrage-free if the TV distance to the binomial model were less than $\min(p_u,p_d)$. The other cases $S_d < S_m = 1+ r < S_u$ and $S_d < 1+ r < S_m < S_u$ can be handled similarly \citep{ostrovski2013stability}. While these results are tractable it was not clear (to us) how to apply these results to develop a dual formulation to study the market model transitions from no-arbitrage to arbitrage or vice versa. Furthermore, total variation distance has been described as a strong notion of distance in the academic literature. Given our motivation to avoid (strong) restrictions in our characterization of robust no-arbitrage markets, it would seem that a different notion of distance between probability measures might be more appropriate. \par

A recent paper by \citet{Bartl17} \textit{explicitly} incorporates a no-arbitrage constraint directly into the worst case European call option pricing problem under Wasserstein ambiguity. We consider this problem from a different perspective in this paper, namely we restrict the Wasserstein ball of probability measures to \textit{implicitly} consider only those measures which are arbitrage-free without the need to enforce an explicit constraint. In Section 2, the theoretical machinery to compute a critical radius $\delta^*_{w(s)}$ is developed to pursue this approach. Simpler worst case option pricing formulas (that omit the explicit no-arbitrage constraint) are derived as well. \par 

Finally, another recent paper by the same author \citep{bartl2019exponential} investigates the robust exponential utility maximization problem in a discrete time setting. The worst case expected utility is maximized under a family of probabilistic models of endowment that satisfy no-arbitrage conditions \textit{by assumption}. The authors show that an optimal trading strategy exists and they provide a dual representation for the primal optimization problem. Furthermore, the optimal value is shown to converge to the robust superhedging price as the risk aversion parameter increases.

\subsection{Arbitrage Framework}
This section lays out the foundations for our framework to investigate the arbitrage properties under distributional uncertainty. Recall the approach taken here is to start from the classical no-arbitrage conditions and introduce a notion of distributional uncertainty via the Wasserstein distance metric. As such, we include definitions for these terms as well as commentary on some important results: 
\begin{enumerate}[label=(\roman*)]
\item definitions for no-arbitrage and statistical arbitrage conditions;
\item Lagrangian duality to formulate the dual problem for robust arbitrage in financial markets;
\item existence and structure of worst case distributions;
\item computation of Wasserstein distance between distributions.
\end{enumerate}
\subsubsection{Weak and Strong No-Arbitrage (NA) Conditions}
The set of admissible portfolio weights for the  weak no-arbitrage conditions is 
\begin{equation*}\label{weakW}
 \Gamma_w(S_0) := \{ w \in \mathbb{R}^n : w \cdot S_0 = 0; \:  w \neq 0 \}. \tag{WW}
\end{equation*}
The set of admissible portfolio weights for the strong no-arbitrage conditions is  
\begin{equation*}\label{strongW}
 \Gamma_s(S_0) := \{ w \in \mathbb{R}^n : w \cdot S_0 < 0 \}. \tag{SW}
\end{equation*}
The no-arbitrage condition to be evaluated under probability measure $Q$ in both cases is
$\Pr( w \cdot S_1 \geq 0 ) = \mathbb{E}^Q [ \, \mathbbm{1}_{\{ w \cdot S_1 \geq 0 \}} \, ] < 1$.
Note that portfolio weight vectors $w$ satisfy the positive homogeneity property (of degree zero) since $\Pr( w \cdot S_1 \geq 0 ) = \Pr( \tilde{w} \cdot S_1 \geq 0 )$ for $\tilde{w} = c w \; \text{and} \; c > 0$. It is the proportions of the holdings in the assets that distinguish $w$ vectors, not their absolute sizes.
Weak arbitrage requires two conditions to hold: $\Pr(w \cdot S_1 \geq 0) = 1$ \textit{and} $\Pr(w \cdot S_1 > 0) > 0$. The second condition is not easily incorporated into the duality framework of this paper and hence it is omitted. Consequently the critical radius $\delta^*_w$ that is developed in Section 2 may not be tight. Strong arbitrage requires just one condition hence the bound $\delta^*_s$ will be tight. \par
For a given measure $Q$, no weak arbitrage means that 
$\sup_{w \in \Gamma_w} \mathbb{E}^Q [ \, \mathbbm{1}_{\{ w \cdot S_1 \geq 0 \}} \, ] < 1$.
Similarly, for a given measure $Q$, no strong arbitrage means that 
$\sup_{w \in \Gamma_s} \mathbb{E}^Q [ \, \mathbbm{1}_{\{ w \cdot S_1 \geq 0 \}} \, ] < 1$.
The empirical measure, $Q_N$, is defined as
$Q_N(dz) = \frac{1}{N} \sum_{i=1}^{N} \mathbbm{1}_{s_{(1,i)}} (dz)$.
To simplify the notation, the leading subscript on $s_{(1,i)}$ is suppressed and going forward we refer to the realization of time 1 asset price vector $s_{(1,i)}$ as just $s_i$. In the context of this work, the uncertainty set for probability measures is
$\mathcal{U}_{\delta}(Q_N) = \{Q: D_c(Q,Q_N) \leq \delta\}$
where $D_c$ is the optimal transport cost or Wasserstein discrepancy for cost function $c$ \citep {blanchetMV}.
The definition for $D_c$ is
\[ D_c(Q,Q') = \inf \{ \mathbb{E}^\pi[c(X,Y)]: \pi \in \mathcal{P}(\mathbb{R}^n \times \mathbb{R}^n), \pi_X = Q, \pi_Y = Q' \} \]
where $\mathcal{P}$ denotes the space of Borel probability measures and $\pi_X$ and $\pi_Y$ denote the distributions of $X$ and $Y$. 
Here $X$ denotes asset prices $S_X \in \mathbb{R}^n$ and $Y$ denotes asset prices $S_Y \in \mathbb{R}^n$ respectively.
This work uses the cost function $c$ where
$c(u,v) = \| u-v \|_2$.

\subsubsection{Note on Equivalence of Single and Multi Period NA}
For clarity we cite the following result from the literature \citep{delbaen2006mathematics}. Let $S = (S_t)_{t=0}^T$ be a discrete price process (with unit increments and $T \in \mathbb{N}$) on a finite probability space $(\Omega,\mathcal{F},\mathbb{P})$. Then the following are equivalent:
\begin{enumerate}[label=(\roman*)]
	\item $S$ satisfies the no-arbitrage property;
	\item For each $0 \leq t < T$, we have that the one-period market $(S_t,S_{t+1})$ with respect to the filtration $( \mathcal{F}_t, \mathcal{F}_{t+1})$ satisfies the no-arbitrage property.
\end{enumerate} 

Further detail on the equivalence of single and multi period no-arbitrage can be found in e.g.\ \citet{leroy2014principles}. As our focus in this paper is on the discrete single period setting, the above relationship suffices. One direction for further research would be to consider the robust no-arbitrage properties in a multi period \textit{continuous} time setting for a suitable class of admissible trading strategies. A more general version of the fundamental theorem of asset pricing applies there. See \citet{delbaen2006mathematics} for additional detail on this topic.

\subsubsection{Weak and Strong Statistical Arbitrage (SA) Conditions}
To characterize the situation where a profitable trading opportunity is highly likely yet not necessarily certain, we introduce a notion of statistical arbitrage. Recall that no universal definition of statistical arbitrage currently exists \citep{lazzarino2018statistical}. Towards that end, we propose using a relaxation of the classical arbitrage conditions to define a notion of statistical arbitrage. In particular, let us write the best case (bc) statistical arbitrage (of level $\alpha^{bc} \in (0,1)$) condition under probability measure $Q$ as $\Pr( w \cdot S_1 \geq 0 ) = \mathbb{E}^Q [ \, \mathbbm{1}_{\{ w \cdot S_1 \geq 0 \}} \, ] \leq \alpha^{bc}$. The set of admissible portfolio weights for the weak (strong) condition is $w \in \Gamma_{w(s)}$ as before (see Section 1.4.1). Intuitively, the best case statistical arbitrage condition says that it should not be possible to construct a zero (or negative) cost portfolio that returns either a profit or no chance of losses with probability $\alpha^{bc}$ close to 1. In the limit $\alpha^{bc} \rightarrow 1$ one recovers the classical arbitrage condition. Similarly, the worst case (wc) condition (of level $\alpha^{wc} \in (0,1)$) is $\Pr( w \cdot S_1 \geq 0 ) = \mathbb{E}^Q [ \, \mathbbm{1}_{\{ w \cdot S_1 \geq 0 \}} \, ] \geq \alpha^{wc}$. Probability $\alpha^{wc}$ close to 0 describes a no-win situation.

\subsubsection{Restatement of Lagrangian Duality Result}
In Section 2 we formulate the primal stochastic optimization problem for distributionally robust arbitrage-free markets. As in our earlier work \citep{Singh19xva1} a key step in the approach is to use recent Lagrangian duality results to formulate the equivalent dual problem.  The dual problem is much more tractable than the primal problem since it only involves the reference probability measure as opposed to a Wasserstein ball of probability measures (of some finite radius). This allows us to solve a \textit{maximin} optimization problem under the original empirical measure defined by the selected data set. A brief restatement of this duality result follows next. \par
For real valued upper semicontinuous objective function $f \in L^1$ and non-negative lower semicontinuous cost function $c$ such that $\{ (u,v) : c(u,v) < \infty \}$ is Borel measurable and non-empty, it holds that \citep{BlanchetML}
\[ \sup_{ Q \in \mathcal{U}_{\delta}(Q_N) } \mathbb{E}^Q[ f(X) ] = \inf_{ \lambda \geq 0 } \: [ \lambda \delta + \frac{1}{N} \sum_{i=1}^n \Psi_\lambda(x_i) ] \]
where
\[ \Psi_\lambda(x_i) := \sup_{ u \in \text{dom} (f) } [ f(u) - \lambda c(u,x_i) ]. \]
The primal problem (LHS above) is concerned with the worst case expected loss for some objective function $f$ with respect to a Wasserstein ball of probability measures of finite radius $\delta$. The Wasserstein ball is used to reflect some (real world) uncertainty about the true underlying distribution for random variable (or vector) $X$. Note that the primal problem is an infinite dimensional stochastic optimization problem and thus difficult to solve directly. The simplicity and tractability of the dual problem (RHS above) make it quite attractive as an analytical and/or computational tool in our toolkit. \par
Further details, including proofs and concrete examples, can be found in the papers by \citet{blanchetFirst}, \citet{Gao16}, and \citet{ Esfahani17}. These authors independently derived these results around the same time although \citet{blanchetFirst} did so in a more general setting. The duality result has been applied by the above authors and others in several papers on topics in data driven distributionally robust stochastic optimization such as robust machine learning, portfolio selection, and risk management. For these types of robust optimization problems, the incorporation of distributional uncertainty can be viewed as adding a penalty term (similar to penalized regression) to the optimal solution \citep{blanchetMV}. This gives us a nice intuitive way to think about the \textit{cost} of robustness.

\subsubsection{Characterization of Worst Case Distributions}
Simply put, the set of worst case (wc) distributions (when non-empty) can be defined as $WC(f,\delta) := \{ Q^* : \mathbb{E}^{Q^*} [f(X)] = \sup_{ Q \in \mathcal{U}_{\delta}(Q_N) } \mathbb{E}^Q[ f(X) ] \}$. Another recent set of results from the literature describes the existence and structure of the worst case distribution(s) when they exist \citep{blanchetFirst}, \citep{Gao16}, \citep{ Esfahani17}.  The boundedness conditions for existence are tied to the growth rate $\kappa := \limsup\limits_{d(X,X_0) \rightarrow \infty} \frac{f(X) - f(X_0)}{d(X,X_0)}$ for fixed $X_0$ and the value of the dual minimizer $\lambda^*$. For empirical reference distributions, supported on $N$ points, such that $WC(f,\delta)$ is non-empty, there exists \textit{a} worst case distribution that is \textit{another} empirical distribution supported on at most $N+1$ points. This worst case distribution can be constructed via a greedy approach. For up to $N$ points, they can be identified as solving $x_i^* \in \argmin_{\tilde{x} \in dom(f)} [ \lambda^* c(\tilde{x},x_i) - f(\tilde{x}) ]$. At most one point has its probability mass split into two pieces (according to budget constraint $\delta$) that solve $x_{i_0}^{*},x_{i_0}^{**}  \in \argmin_{\tilde{x} \in dom(f)} [ \lambda^* c(\tilde{x},x_{i_0}) - f(\tilde{x}) ]$. Details can be found in \citet{Gao16}. For our problem setting, the growth rate conditions are satisfied and hence we proceed to formulate and then apply a greedy algorithm (see Section 2.2) to compute the worst case distribution for a concrete example in Section 5. A similar example from the literature, which uses a greedy algorithm to compute the minimal (worst case) membership to a given set $C$, is covered in \citep{Gao16}. Note that other worst case distributions can be constructed with different support sets and/or probability mass functions (PMFs). It can be insightful to examine how the reference distribution can be perturbed for a given objective $f$ as $\delta$ varies. See Section 2.2 for specific commentary on the structure and construction of the worst case distribution(s) for the robust NA problem.

\subsubsection{On Computing Wasserstein Distance}
This section introduces some standard and recent results on computing Wasserstein distance between distributions. The recent results are focused on discrete distributions since our problems of interest are data driven. The standard results (below) are taken from the online document by \citet{LWonWD}. Wasserstein distance has simple expressions for univariate distributions. The Wasserstein distance of order $p$ is defined over the set of joint distributions $\mathcal{P}$ with marginals $Q$ and $Q'$ as 
\[
W_p(Q,Q') = \left(  \inf_{\pi \in \mathcal{P}(X,Y)} \int \| x - y \|^p \: d \pi(x,y) \right)^{1/p}.
\]
Note that in this work we consider Wasserstein distance of order $p=1$. When $d=1$ there is the formula 
\[
W_p(Q,Q') = \left( \int_{0}^{1} | F^{-1}(z) - G^{-1}(z) |^p \: dz \right)^{1/p}.
\]
For empirical distributions with $N$ points, there is the formula using order statistics on $(X,Y)$
\[
W_p(Q,Q') = \left( \sum_{i=1}^{N} \| X_{(i)} - Y_{(i)} \|^p \right)^{1/p}.
\]
Additional closed forms are known for: (i) normal distributions, (ii) mappings that relate Wasserstein distance to multiresolution $L_1$ distance. See \citet{LWonWD} for details. This concludes the brief survey of standard (closed form) results.

For discrete distributions, at least a couple of methods have been recently developed to compute approximate and/or (in the limit) exact Wasserstein distance. The commentary on these methods is taken from \citet{xie2018fast}. For distributions with finite support, and cost matrix $C$, one can compute $W(Q,Q') := \min_{\pi} \langle C, \pi \rangle$ with probability simplex constraints using linear programming (LP) methods of $O(N^3)$ complexity. An entropy regularized version of this, using regularizer $h(\pi) := \sum \pi_{i,j} \log \pi_{i,j}$ gives rise to the Sinkhorn distance
\[
W_{\epsilon}(Q,Q') := \min_{\pi} \langle C, \pi \rangle + \epsilon h(\pi)
\]
which can be solved using iterative Bregman projections via the Sinkhorn algorithm. However, the authors comment that certain problems (such as generative model learning and barycenter computation) experience performance degradation for a moderately sized $\epsilon$ but opting for a small size can be computationally expensive. To address these shortcomings, they develop their own approach called inexact proximal point method for optimal transport (IPOT). The proximal point iteration takes the form
\[
\pi^{(t+1)} = \argmin_{\pi} \langle C, \pi \rangle + \beta^{(t)} D_h( \pi, \pi^{(t)} )
\]
where $\beta$ denotes a parameter of the method and $D_h$ denotes the Bregman divergence based on the entropy function. Substitution for Bregman divergence gives the form
\[
\pi^{(t+1)} = \argmin_{\pi} \langle C - \beta^{(t)} \log \pi^{(t)}, \pi \rangle + \beta^{(t)} h( \pi ).
\]
It turns out that this iteration can also be solved via the Sinkhorn algorithm. However the authors propose an inexact method that improves efficiency while maintaining convergence. See \citet{xie2018fast} for details.

\section{Theory: Robust Arbitrage Conditions for Financial Markets}
This section develops the theory for robust arbitrage in financial markets. In Section 2.1, the primal problem is formulated using classical notions of arbitrage as discussed in Section 1.4.1. The dual problem is formulated using the Lagrangian duality result from Section 1.4.4. Note that the dual problem is a \textit{maximin} stochastic optimization problem. The inner optimization problem (evaluating $\Psi_{\lambda,w}$) can be solved analytically using the Projection Theorem \citep{calafiore2014optimization}. The middle optimization problem (evaluating the dual objective function over $\inf_{\lambda \geq 0}$) can be solved via execution of a simple linear search algorithm over a finite set of points. The outer optimization problem (evaluating over $\sup_{w \in \Gamma_{w(s)}}$) can be formulated as an NLP. Finally, the middle and outer problems can be solved jointly via a \textit{maximin} NLP approach. \par

Section 2.2 gives details on the worst case distributions and Sections 2.3 and 2.4 show how to incorporate portfolio restrictions (such as short sales) in a straightforward manner. Section 2.5 introduces the complementary problem of how to find the nearest arbitrage-free measure to the arbitrage admissible reference measure. This machinery gives us a practical approach to explore applications of our framework for robust arbitrage.

\subsection{Robust Weak and Strong No-Arbitrage (NA) Conditions}
The robust weak no-arbitrage conditions can be expressed as
\begin{equation*}\label{weakP}
\sup_{ w \in \Gamma_w } \: { \sup_{ Q \in \mathcal{U}_{\delta}(Q_N) } \mathbb{E}^Q [ \, \mathbbm{1}_{\{ w \cdot S_1 \geq 0 \}} \, ] } < 1 \tag{WP}
\end{equation*}
where $\Gamma_w$ is defined in \ref{weakW}.
Note the indicator function $\mathbbm{1}_{\{ w \cdot S_1 \geq 0 \}}$ on closed set ${\{ w \cdot S_1 \geq 0 \}}$ is upper semicontinuous  hence we can apply the duality theorem (see Section 1.4.4) to obtain the dual formulation
\begin{equation*}\label{weakD}
\sup_{ w \in \Gamma_w } \: { \inf_{ \lambda \geq 0 } \: [ \: \lambda \delta + \frac{1}{N} \sum_{i=1}^N \Psi_{\lambda,w} (s_i) \: ] } < 1 \tag{WD}
\end{equation*}
where $\Psi_{\lambda,w}$ is defined, in terms of cost function $c$, as
$\Psi_{\lambda,w} = \sup_{\tilde{s} \in \mathbb{R}^n} [ \, \mathbbm{1}_{\{ w \cdot \tilde{s} \, \geq \, 0 \}} - \lambda c( \tilde{s}, s_i) \, ]$.
Similarly, for the robust strong no-arbitrage conditions 
\begin{equation*}\label{strongP}
\sup_{ w \in \Gamma_s } \: { \sup_{ Q \in \mathcal{U}_{\delta}(Q_N) } \mathbb{E}^Q [ \, \mathbbm{1}_{\{ w \cdot S_1 \geq 0 \}} \, ] } < 1 \tag{SP}
\end{equation*}
where $\Gamma_s$ is defined in \ref{strongW}, the dual formulation is
\begin{equation*}\label{strongD}
\sup_{ w \in \Gamma_s } \: { \inf_{ \lambda \geq 0 } \: [ \: \lambda \delta + \frac{1}{N} \sum_{i=1}^N \Psi_{\lambda,w} (s_i) \: ] } < 1. \tag{SD}
\end{equation*}

\subsubsection{Inner Optimization Problem}
The objective here is to evaluate $\Psi_{ \lambda,w }$ in closed form. There are two cases to consider.
\begin{case}
\[ \mathbbm{1}_{\{ w \cdot s_i \geq 0 \}} = 1 \implies \Psi_{ \lambda,w } (s_i) = 1 - \lambda \cdot 0 = 1 \quad \text{which is optimal.} \]
\end{case}

\begin{case}
\[ \mathbbm{1}_{\{ w \cdot s_i \geq 0 \}} = 0 \implies \Psi_{ \lambda,w } (s_i) = [ 1 - \lambda c( s_i^*, s_i) ]^+ \quad \text{where} \quad s_i^* = \argmin \| \tilde{s} - s_i\|_2 \quad \text{is optimal.} \]
\end{case}
By the Projection Theorem \citep{calafiore2014optimization},  $\| s_i^* - s_i\|_2 = \frac{|w^\top s_i|}{\|w\|_2} \implies \Psi_{ \lambda,w } (s_i) = [ 1 - \lambda c_i ]^+ \text{\; for \;} \\ c_i = \frac{| w^\top s_i | }{ \| w \|_2 } \in \mathbb{R}^n_+$.

\begin{prop}
\[ \frac{1}{N} \sum_{i=1}^N \Psi_{\lambda,w}(s_i) = K_0(w) + K_1(\lambda,w) \]
where \: $K_0(w) = \frac{1}{N} \sum_{i=1}^N \mathbbm{1}_{\{ w \cdot s_i \geq 0 \}}$ \: and \: $K_1(\lambda,w) = \frac{1}{N} \sum_{i=1}^N \mathbbm{1}_{\{ w \cdot s_i < 0 \}} [ 1 - \lambda c_i ]^+ $ \: for \: $c_i = \frac{| w^\top s_i |}{ \| w \|_2 } \in \mathbb{R}^n_+$.
\end{prop}
\begin{proof}
This follows by a straightforward application of the two cases above.
\end{proof}

\subsubsection{Middle Optimization Problem}
\begin{remark}
In this subsubsection, the dependency of $\lambda^*$ on $(w,\delta)$ is suppressed to ease the notation.
\end{remark}
Now the objective is to evaluate $\inf_{ \lambda \geq 0} H(\lambda) := [ \: \lambda \delta + K_0(w) + K_1(\lambda,w) \: ]$.
Since $H(\lambda)$ is a convex function of $\lambda$, the first order optimality condition suffices to determine $\lambda^* = \argmin_{\lambda \geq 0} H(\lambda)$. Note that $H(\lambda)$ may have kinks so we look for $\lambda^*$ such that $0 \in \partial H(\lambda^*)$. Following the approach in our earlier work \citep{Singh19xva1}, we arrive at the following result.
\begin{prop}
$
Let \: \lambda^{*} = \sup_{\lambda \geq 0} \{ \lambda :  \delta - \frac{1}{N} [ \sum_{i \in J^+_1(\lambda)} \mathbbm{1}_{\{ w \cdot s_i < 0 \}} c_i ] \leq 0  \} = \inf_{\lambda \geq 0} \{  \lambda :  \delta - \frac{1}{N} [ \sum_{i \in J_1(\lambda)} \mathbbm{1}_{\{ w \cdot s_i < 0 \}} c_i ] \geq 0 \}, where\\
$
$J^+_1(\lambda) \: = \{ i \in \{1,\dots,N\} : 1 - c_i \lambda > 0 \}, \:$
$J_1(\lambda) \: = \{ i \in \{1,\dots,N\} : 1 - c_i \lambda \geq 0 \}$. 
In the degenerate case, where $\sup_{\lambda \geq 0}$ is taken over an empty set, select $\lambda^* = 0 \implies H(\lambda^*) = 0$.
\end{prop}
\begin{proofsketch}
This result follows from writing down the first order conditions for left and right derivatives for convex objective function $H(\lambda)$. For each additional index $i \in J^+_1 (J_1)$ such that at least one indicator function is true, we pick up an additional $c_i$ term in the left (right) derivative. Search on $\lambda$ (from the left or the right) until we find $\lambda^*$ such that $0 \in \partial H(\lambda^*)$.
\end{proofsketch}
\begin{proof}
The first order optimality condition says
\[ \delta - \frac{1}{N} \sum_{i \in J^+_1(\lambda)} \mathbbm{1}_{\{ w \cdot s_i < 0 \}} c_i \leq 0 \leq \delta - \frac{1}{N} \sum_{i \in J_1(\lambda)} \mathbbm{1}_{\{ w \cdot s_i < 0 \}} c_i. \]
Note the LHS is an increasing function in $\lambda$. Hence one can write
\[ \lambda^* = \sup_{\lambda \geq 0} \{ \lambda : \delta - \frac{1}{N} \sum_{i \in J^+_1(\lambda)} \mathbbm{1}_{\{ w \cdot s_i < 0 \}} c_i \leq 0 \}. \]
Similarly the RHS is also an increasing function in $\lambda$. Equivalently, one can write
\[ \lambda^* = \inf_{\lambda \geq 0} \{ \lambda : \delta - \frac{1}{N} \sum_{i \in J_1(\lambda)} \mathbbm{1}_{\{ w \cdot s_i < 0 \}} c_i \geq 0 \}. \]
\end{proof}

\begin{prop}
Equivalently, $\lambda^*$ can be computed via a linear seach over $\{ \frac{1}{c_i} \}$ as in Algorithm 1 (listed below). 
\end{prop}
\begin{proof}
The break points for $J_1 (J^+_1)$ are $\{ c_i : i \in \{1,\dots,N\} \}$. Observe that the only possible candidates for $\lambda^*$, as given in  Proposition 2.2, are $\{ \frac{1}{c_i }: i \in \{1,\dots,N\} \}$ or 0. One can sort and relabel the $c_i$ to be in increasing order. Note that $(1 - c_j \lambda) > 0 \implies (1 - c_i \lambda) > 0 \:\: \forall \:\: c_i \leq c_j$. Thus $m \in J_1(J^+_1) \implies \{1,\dots,m \} \in J_1(J^+_1)$. Search backwards to find the smallest index $k^* \in \{1,\dots,N\}$ such that $\sum_{i=1}^{k^*} \mathbbm{1}_{\{ w \cdot s_i < 0 \}} c_i \geq N \delta$. If no such index $k^*$ is found, return $\lambda^* = 0$ else return $\lambda^* = \frac{1}{c_{k^*}}$.
\end{proof}

\begin{algorithm}[H]
\DontPrintSemicolon
	\KwInput{$\{\frac{1}{c_i}\},\: w \:,\: \{s_i\}\:,\: N\:,\: \delta$}
  	\KwOutput{$\lambda^* = \sup_{\lambda \geq 0} \{ \lambda :  \delta - \frac{1}{N} [ \sum_{i \in J^+_1(\lambda)} \mathbbm{1}_{\{ w \cdot s_i < 0 \}} c_i ] \leq 0  \} = \inf_{\lambda \geq 0} \{  \lambda :  \delta - \frac{1}{N} [ \sum_{i \in J_1(\lambda)} \mathbbm{1}_{\{ w \cdot s_i < 0 \}} c_i ] \geq 0 \}$ }
	Set $Q^* = Q_N$ \;
	Sort $\{ c_i \}$ Increasing \;
	Compute $\{ V_k \}$ where $V_k := \sum_{i=1}^k \mathbbm{1}_{\{ w \cdot s_i < 0 \}} c_i$ \;
    	$k = N$ \;
	\If{ $V_k < N \delta$ }
    {
		return $\lambda^* := 0$ \;
    }
    \Else
    {
		\While{ $k \geq 1$ and $V_k \geq N \delta$ }
		{
			$k = k - 1$ \;
		}
		$k^* = k + 1$ \;
		return $\lambda^* := \frac{1}{c_{k^*}}$ \;
    }
\caption{Linear Search over $\{ \frac{1}{c_i} \}$ to compute $\lambda^*$}
\end{algorithm}

\subsubsection{Outer Optimization Problem}
The weak no-arbitrage conditions can \textit{now} be expressed as
\begin{equation*}\label{weakD2}
v_w(\delta) := \sup_{ w \in \Gamma_w } \: \{ \lambda^*(w,\delta) \delta + K_0(w) + K_1(\lambda^*(w,\delta),w) \} < 1. \tag{WD2}
\end{equation*}
Similarly, for the strong no-arbitrage conditions  
\begin{equation*}\label{strongD2}
v_s(\delta) := \sup_{ w \in \Gamma_s } \: \{ \lambda^*(w,\delta) \delta + K_0(w) + K_1(\lambda^*(w,\delta),w) \} < 1. \tag{SD2}
\end{equation*}
The authors are not aware of any such pairing of mixed integer nonlinear program (MINLP) formulation and solver that can return the (global) optimal values $v_{w(s)}(\delta)$ for arbitrary problem instances. Our attempts at such an MINLP formulation to be solved using Neos / Baron MINLP solvers \citep{ts:05} and/or Neos / Knitro solvers \citep{byrdintegrated} were successful on small but not large problem instances. Difficulties were encountered in finding feasible solutions and/or returning optimal solutions.
Given the findings above, our original solution strategy was revised to focus on solving an equivalent NLP \textit{maximin} problem formulation to \textit{local} optimality using the Matlab \textit{fminimax} solver and the identity $\max_{x} \min_{k} F_k(x) = - \min_{x} \max_{k} (-F_k(x)) $. The equivalent formulation is constructed from the observation that $\lambda^* \in \{ \frac{1}{c_k }: k \in \{1,\dots,N\} \} \cup \{ \lambda_0 := 0 \}$. Developing a \textit{global} solution strategy would be an interesting area for further research.

\begin{theorem}
$v_w(\delta)$ is approximated by the (global) solution to nonlinear program (NLP) N\_WNA (listed below).
\end{theorem}
\noindent The constraints on variables below, with index $i$, apply for $i \in \{1,\dots,N\}$, although this is suppressed. Also recall that weight vectors $w$ satisfy homogeneity, hence the use of ``big M" to express $w \in \Gamma_{w(s)}$ is appropriate.
\begin{maxi}|l|<b>
{\substack{w \in \mathbb{R}^n} }
{\min_{\lambda_k \: : \: k \in \{0,1,...,N\}} \quad F_k(w) := \lambda_k \delta + \frac{1}{N} \bigg[ \sum_{i=1}^N \mathbbm{1}_{\{ w \cdot s_i \geq 0 \}} + \sum_{i=1}^N z^+_i \mathbbm{1}_{\{ w \cdot s_i < 0 \}} \bigg] }{\label{M_WNA}}
{v_w(\delta)=}
\addConstraint{c_i}{= \frac{| w^\top s_i |}{\| w \|_2}}
\addConstraint{\lambda_k}{=\frac{1}{c_k} \quad \forall k \in \{1,\dots,N\}}
\addConstraint{\lambda_0}{=0}
\addConstraint{|w_i|}{\leq M }
\addConstraint{w \cdot S_0}{= 0}
\addConstraint{\sum_{j=1}^n|w_j|}{\geq \epsilon}
\addConstraint{z_i}{= [1-\lambda_k c_i]}
\end{maxi}
\begin{proof}
The NLP formulation follows from equation \ref{weakD2} and the fact that $\lambda^* \in \{ \frac{1}{c_k }: k \in \{1,\dots,N\} \} \cup \{ \lambda_0 := 0 \}$.
\end{proof}

\begin{corollary}
$v_s(\delta)$ is approximated by the solution to NLP N\_SNA (described next). N\_SNA is very similar to N\_WNA. One just needs to omit the $\sum_{j=1}^n |w_j| \geq \epsilon$ constraint and replace the initial cost constraint $w \cdot S_0 = 0$ with $-M \leq w \cdot S_0 \leq - \epsilon$, or equivalently with $w \cdot S_0 = \kappa < 0$, ($\kappa$ arbitrary), using the homogeneity property of $w$.
\end{corollary}
\begin{proof}
There is a slight variation on the constraints to express $w \in \Gamma_s$. No other changes are needed.
\end{proof}

\begin{theorem}
The critical radius $\delta^*_{w(s)}$ can be expressed as $\inf \{\delta \geq 0 :v_{w(s)}(\delta) = 1\}$. Furthermore, $\delta^*_{w(s)}$ can be explicitly computed via binary search. Let $\delta_{w(s)} < \delta^*_{w(s)}$. For $Q_{w(s)} \in \mathcal{U}_{\delta_{w(s)}}(Q_N)$, it follows that $Q_{w(s)}$ is weak (strong) arbitrage-free. For $Q_{w(s)} \notin \mathcal{U}_{\delta^*_{w(s)}}(Q_N)$, it follows that $Q_{w(s)}$ may admit weak (strong) arbitrage.

\end{theorem}
\begin{proof}
This characterization of the critical radius $\delta^*_{w(s)}$ follows from the condition \ref{weakD2} (\ref{strongD2}) as well as the definition of weak (strong) no-arbitrage. The asymptotic properties of $v_{w(s)}$ are such that $v_{w(s)}(0) \leq 1$ and $\lim_{\delta \to \infty} v_{w(s)}(\delta) = 1$. Furthermore, since $v_{w(s)}(\delta)$ is a non-decreasing function of $\delta$, it follows that $\delta^*_{w(s)}$ can be computed via binary search.
\end{proof}
\noindent One can view the critical radius $\delta^*_{w(s)}$ as a relative measure of the \textit{degree} of weak (strong) arbitrage in the reference measure $Q_N$. Those $Q_N$ which are ``close" to allowing arbitrage will have a relatively smaller value of $\delta^*_{w(s)}$.

\subsection{Best Case Distribution for Arbitrage Condition}
This subsection expands on the commentary in Section 1.4.5 and works through the details for how this notion applies to the robust no-arbitrage problem. First recall from Section 1.4.5 the definition of the set of worst case distributions as $WC(f,\delta) := \{ Q^* : \mathbb{E}^{Q^*} [f(X)] = \sup_{ Q \in \mathcal{U}_{\delta}(Q_N) } \mathbb{E}^Q[ f(X) ] \}$ and $x_i^* \in \argmin_{\tilde{x} \in dom(f)} [ \lambda^* c(\tilde{x},x_i) - f(\tilde{x}) ]$. For the NA problem, $c_i$ represents $c(s^*_i,s_i)$ and the objective function is $f(S_1) := \mathbbm{1}_{\{w \cdot S_1 \geq 0\}}$ hence growth rate $\kappa = 0 \implies WC$ non-empty (growth rate condition satisfied). From an arbitrageur's perspective, $Q^*$ represents a best case distribution, hence let us relabel the set $WC$ as $BC$. We use the notation $BC(w,\delta)$ to emphasize the parametrization on $w$. In Section 6 the greedy algorithm (to be described below) is used to compute \textit{a} best case distribution $Q^*_w \in BC(w,\delta^*)$. Please note that although this $Q^*_w$ satisfies $\mathbb{E}^{Q^*} [f(S_1)] = 1$ it does not \textit{necessarily} allow arbitrage. Intuitively, an arbitrage distribution would use up budget $\delta \geq \delta^*$ to allow arbitrage whereas the greedy worst case distribution may not do so. An arbitrage distribution must satisfy
\[
\sup_{ w \in \Gamma_{w(s)} } \: { \sup_{ Q \in \mathcal{U}_{\delta^*}(Q_N) } \mathbb{E}^Q [ \, \mathbbm{1}_{\{ w \cdot S_1 \geq 0 \}} \, ] } = 1. 
\]
whereas a (greedy) worst case distribution with budget $\delta \geq \delta^*$ only needs to satisfy the condition that the inner $\sup$ evaluates to 1. However, selecting portfolio weights $w^*$ that satisfy the outer $\sup$ condition above, one \textit{can} recover $Q^*_{w^*}$ that allows arbitrage. 
\begin{algorithm}[!htb]
\DontPrintSemicolon
	\KwInput{$f\:,\: w \:,\: \{s_i\}\:,\: \{c_i\}\:, N\:,\: \delta$}
  	\KwOutput{$Q^*_w : \mathbb{E}^{Q^*_w} [f(X)] = \sup_{ Q \in \mathcal{U}_{\delta}(Q_N) } \mathbb{E}^Q[ f(X) ] $}

	Define $Q^*_w := \{ Q^*_v, Q^*_p \}$ where $Q^*_{v}$ denotes the support and $Q^*_{p}$ denotes probabilities \;
	Set $Q^*_w = Q_N$ so that those scenarios $\{ i \in \{1,\dots,N\} : \mathbbm{1}_{\{ w \cdot s_i \geq 0 \}} \}$ do not move \;
	Sort $\{ c_i \}$ Increasing \;
	Set $V_0 := 0$ and Compute $\{ V_k \}$ where $V_k := \sum_{i=1}^k \mathbbm{1}_{\{ w \cdot s_i < 0 \}} c_i$ \;
    	$k = 1$ \;
	\While{ $k \leq N$ and $V_k \leq N \delta$ }
	{
		\If{ $\mathbbm{1}_{\{ w \cdot s_k < 0 \}}$ and $(  1 - \lambda^* c_k ) \geq 0$ }
		{
			$Q^*_v(k) = s_k - \sgn{ (w \cdot s_k) } c_k \frac{w}{\|w\|}$ \;
		}
		$k = k + 1$ \;
	}
	\If{ $k \leq N$ and $V_k > N \delta$ and $\mathbbm{1}_{\{ w \cdot s_k < 0 \}}$ }
    {
		$p_0 = ( N \delta - V_{k-1} ) / V_k$ \;
		$Q^*_p(N+1) = \frac{p_0}{N}$ \;
		$Q^*_v(N+1) = s_k - \sgn{ (w \cdot s_k) } c_k \frac{w}{\|w\|}$ \;
		$Q^*_p(k) = \frac{1 - p_0}{N}$ \;
    }
\caption{Greedy Algorithm to compute $Q^*_w \in BC(w,\delta)$ for NA}
\end{algorithm}

\subsection{Portfolio Restrictions}
This subsection discusses refinements to the no-arbitrage conditions (see Section 2.1) to characterize portfolio restrictions such as short sales restrictions, min and max position constraints, and cardinality constraints \citep{cornuejols2018optimization}. For efficiency of presentation, we refer the reader to the N\_WNA and N\_SNA NLP problems discussed in Section 2.1.3 and do no restate those formulations here. An advantage of the computational machinery developed in this paper is that such portfolio restrictions can be readily incorporated into the existing framework. Note that these additional constraints may cause the restricted NLP problem to violate the homogeneity property of $w$ so one should exercise caution in formulating the new problem correctly. For example, for restricted N\_SNA one should use the $-M \leq w \cdot S_0 \leq -\epsilon$ constraint instead of $w \cdot S_0 = \kappa < 0$, ($\kappa$ arbitrary). Table 1 below describes the various portfolio restrictions (discussed here) and associated constraints. Others are possible as well. Note that the index set is $j \in \{1,\dots,n\}$ which is suppressed for brevity.

\renewcommand{\arraystretch}{1.5}
\begin{table}[h]
\normalsize
\begin{center}
\caption{Portfolio Restrictions}
\begin{tabular}{ |c|c|l| }
 \hline
\textit{Restriction} & \textit{MINLP Constraint} & \textit{No Restriction} \\
 \hline
Short Sales & $w_j \geq ss_j$ \:\: where $ss_j \in \mathbb{R}_{-}$ \:\: is short sales limit & $ss_j = -M$ \\ 
\hline
Min Positions & $|w_j| \geq \underline{w}$ \:\: where $\underline{w} \in \mathbb{R}_{+}$ \:\: denotes min position & $\underline{w} = 0$ \\ 
 \hline
Max Positions & $|w_j| \leq \overline{w}$ \:\: where $\overline{w} \in \mathbb{R}_{+}$ \:\: denotes max position & $\overline{w} = M$ \\ 
\hline
Cardinality & $\sum_{j=1}^n \mathbbm{1}_{\{ |w_j| \geq \epsilon \}} \leq m$ \:\: where $m \in \{1,\dots,n\}$ \:\: is cardinality constraint & $m=n$ \\ 
\hline
Allocations & $| \sum_{j \in A_k} w_j S_{0j} |  \leq \overline{A_k}$ \:\: where $\overline{A_k} \in \mathbb{R}_{+}$ \:\: is asset class $k$ allocation constraint & $\overline{A_k}=M n$ \\ 
\hline
\end{tabular} 
\end{center} 
\end{table}
\renewcommand{\arraystretch}{1}

\subsection{NA Conditions Under No Short Sales}
This subsection gives a brief summary (using the author's notation) of the work by \citet{oleaga2012arbitrage} to formulate equivalent (weak) no-arbitrage conditions, in terms of existence of risk neutral probability measures, under no short sales. A similar exercise could be conducted for strong no-arbitrage conditions although the author focuses on the weak conditions. From the previous subsection, no short sales conditions can be \textit{directly} imposed by setting $ss_j = 0$ for $j \in J$ for some index set $J \subseteq \{1,\dots,n\}$. Oleaga begins his paper with a remark that the Fundamental Theorem of Finance establishes the equivalence between the no-arbitrage conditions and the existence of a risk neutral probability measure (see Section 1.1 of this paper for details) under the assumption that short selling of risky securities is allowed. He remarks that when short sales are not allowed, the academic literature is scarce regarding equivalent conditions on probability measures. As motivation for his main result (which implies that existence of a risk neutral measure is not guaranteed under no short sales) the author develops two examples: one using a simple one-period binomial model with one risky asset, and another involving wagers in a stylized market where the assets are \text{Arrow-Debreu} securities. Using standard techniques in linear algebra, convex analysis, and the separating hyperplane theorem the author proves his main result which is stated below for convenience.
\begin{theorem*} (Arbitrage Theorem for No Short Sales). The market model $\mathcal{M}$ with $m$ scenarios for $n$ assets $X_j : j \in \{1,\dots,n\}$ has no-arbitrage opportunities iff there exists a probability measure $\pi$ such that the initial prices $x_j$ are greater than or equal to the discounted value of the expected future prices under $\pi$. Written in symbols we have:
\[
x_j \geq \frac{1}{1+r_0} \sum_{i=1}^m \pi_i X_{ij} \quad where \quad j \in \{1,\dots,n\}.
\]
Moreover, for those assets $X_j : j \in \{1,\dots,n\}$ where short selling is allowed, equality is achieved in the above relation.
In particular, the bank account or cash bond (used to execute the borrowing to purchase the portfolio at time 0) is treated as a special asset $X_0$ excluded from the above relation. It would hold with equality if included.
\end{theorem*}

\noindent In an independent work, \citet{leroy2014principles} develop essentially the same results for both weak and strong no-arbitrage conditions. They show that for the weak conditions, the probability measure $\pi$ is such that $\pi > 0$ whereas for the strong conditions $\pi \geq 0$.

\subsection{Nearest NA Problem}
Recall that the motivating question here is how to find the nearest arbitrage-free measure to the arbitrage admissible reference measure. 
\subsubsection{Short Sales Allowed}
This subsection looks at the problem of computing the minimal distance $\delta^*_g$ to an arbitrage-free measure for a reference measure $Q_N$ that admits arbitrage. In a discrete setting, the nearest (strong) no-arbitrage problem can be formulated as
\begin{equation*}\label{NstrongP}
\delta^*_{ns} = \min_{\tilde{X}} \| X - \tilde{X} \|_F \:\: \text{such that} \:\:\: \exists \; q \geq 0 \: : \: p = \tilde{X} q \tag{NSP}
\end{equation*}
where $\|X\|_F$ denotes the Frobenius norm of matrix $X$. A penalty relaxation can be formulated as
\begin{equation*}\label{NstrongPR}
\delta^*_{nsr}(\beta) = \min_{\tilde{X}, q \geq 0} \| X - \tilde{X} \|_F + \beta \| p - \tilde{X} q \|_F^2  \tag{NSPR}
\end{equation*}
A tight lower bound $\delta^*_{nst} \leq \delta^*_{ns}$ to the relaxation problem \ref{NstrongPR} is given by 
\begin{equation*}\label{NstrongPRT}
\delta^*_{nst} = \sup_{\beta \geq 0} \: \delta^*_{nsr}(\beta) \tag{NSPRT}
\end{equation*} 
For a complete market with non-redundant securities, note that $X$ (and hence $\tilde{X}$) is a full rank, invertible square $n \times n$ matrix. 

\subsubsection{No Short Sales}
This subsection mimics the approach of the previous subsection, however we make use of the equivalent probability measure condition discussed in Section 2.4 \citep{oleaga2012arbitrage}, \citep{leroy2014principles}. In a discrete setting, the nearest (weak) no-arbitrage problem, under no short sales, can be formulated as
\begin{equation*}\label{NNweakP}
\delta^*_{nns} = \min_{\tilde{X}} \| X - \tilde{X} \|_F \:\: \text{such that} \:\:\: \exists \:\; \text{probability measure} \:\; q > 0 \: : \: p \geq \frac{\tilde{X}}{1+r_0} q. \tag{NNWP}
\end{equation*}
A penalty relaxation can be specified as
\begin{equation*}\label{NNweakPR}
\delta^*_{nnsr}(\beta) = \min_{\tilde{X}, q > 0} \| X - \tilde{X} \|_F + \beta \| (\tilde{X} q - (1+r_0)p)^+ \|_F^2  \tag{NNWPR}
\end{equation*}
A tight lower bound $\delta^*_{nnst} \leq \delta^*_{nns}$ to the relaxation problem \ref{NNweakPR} is given by 
\begin{equation*}\label{NNweakPRT}
\delta^*_{nnst} = \sup_{\beta \geq 0} \: \delta^*_{nnsr}(\beta) \tag{NNWPRT}
\end{equation*} 
Recall the bank account or cash bond (used to borrow) is excluded from the above relation. For a complete market with non-redundant securities, note that $X$ (and hence $\tilde{X}$) is a full rank, invertible square $n \times n$ matrix. 

\subsection{Alternate Robust NA Conditions}
For completeness, we comment on an alternate formulation of the robust NA conditions (from Section 2.1) that \textit{exchanges} the order of $\sup$ operators. Such conditions can be expressed as
\begin{equation*}\label{strongRNAP}
\sup_{ Q \in \mathcal{U}_{\delta}(Q_N) } \: \sup_{ w \in \Gamma_s } \: \mathbb{E}^Q [ \, \mathbbm{1}_{\{ w \cdot S_1 \geq 0 \}} \, ] < 1 \tag{RSNAP}
\end{equation*}
where $\Gamma_s$ is defined in \ref{strongW}. The intuitive meaning of this formulation is that the market player \textit{first} chooses a favorable distribution $Q \in \mathcal{U}_{\delta}(Q_N)$ and \textit{then} the portfolio manager chooses an optimal $w \in \Gamma_s$. It is clear that 
\[
\sup_{ Q \in \mathcal{U}_{\delta}(Q_N) } \: \sup_{ w \in \Gamma_s } \: \mathbb{E}^Q [ \, \mathbbm{1}_{\{ w \cdot S_1 \geq 0 \}} \, ] = \:
\sup_{ w \in \Gamma_s } \: { \sup_{ Q \in \mathcal{U}_{\delta}(Q_N) } \mathbb{E}^Q [ \, \mathbbm{1}_{\{ w \cdot S_1 \geq 0 \}} \, ] }.
\]

\section{Theory: Robust Statistical Arbitrage (SA) Conditions for Financial Markets}
This section develops the theory for robust statistical arbitrage in financial markets. We follow the same approach as in Section 2 for robust arbitrage. For simplicity, and to ease the notation, let us focus on the strong conditions. The weak conditions can be handled similarly, replacing $w \in \Gamma_s$ with $w \in \Gamma_w$, as in Section 2. In Section 3.1, the primal problem for the SA best case conditions is formulated using notions of statistical arbitrage as discussed in Section 1.4.3. The dual problem is formulated using the Lagrangian duality result from Section 1.4.4. The dual problem is a \textit{maximin} stochastic optimization problem. Section 3.2 touches on the best case SA distribution. In Section 3.3, the primal problem for the SA worst case conditions is formulated. The dual problem for this is \textit{maximax}. Both dual problems can be solved as in Section 2. Section 3.4 touches on the worst case SA distribution. Section 3.5 addresses portfolio restrictions. Section 3.6 covers the nearest SA problem. Section 3.7 discusses alternate robust SA conditions. Altogether, this machinery gives us a practical approach to explore applications of our framework in Sections 4 and 6.

\subsection{Robust SA Best Case Conditions}
The robust (strong) statistical arbitrage best case conditions (of level $\alpha^{bc} \in (0,1)$) can be expressed as
\begin{equation*}\label{strongSAP}
\sup_{ w \in \Gamma_s } \: { \sup_{ Q \in \mathcal{U}_{\delta}(Q_N) } \mathbb{E}^Q [ \, \mathbbm{1}_{\{ w \cdot S_1 \geq 0 \}} \, ] } \leq \alpha^{bc}, \tag{SSAP}
\end{equation*}
where $\Gamma_s$ is defined in \ref{strongW}.
As before, the indicator function $\mathbbm{1}_{\{ w \cdot S_1 \geq 0 \}}$ on closed set ${\{ w \cdot S_1 \geq 0 \}}$ is upper semicontinuous  hence we can apply the duality theorem (see Section 1.4.4) to obtain the dual formulation
\begin{equation*}\label{strongSAD}
\sup_{ w \in \Gamma_s } \: { \inf_{ \lambda \geq 0 } \: [ \: \lambda \delta + \frac{1}{N} \sum_{i=1}^N \Psi_{\lambda,w} (s_i) \: ] } \leq \alpha^{bc} \tag{SSAD}
\end{equation*}
where $\Psi_{\lambda,w}$ is defined, in terms of cost function $c$, as
$\Psi_{\lambda,w} = \sup_{\tilde{s} \in \mathbb{R}^n} [ \, \mathbbm{1}_{\{ w \cdot \tilde{s} \, \geq \, 0 \}} - \lambda c( \tilde{s}, s_i) \, ]$.

\subsubsection{Inner Optimization Problem}
The goal here is the same as for the robust no-arbitrage conditions in Section 2.1.1, namely to evaluate $\Psi_{\lambda,w}$ in closed form. As such the solution is also the same, therefore one can invoke Proposition 2.1 to compute $\frac{1}{N} \sum_{i=1}^N \Psi_{\lambda,w}(s_i)$.

\subsubsection{Middle Optimization Problem}
As before, in Section 2.1.2, the objective is to evaluate $\inf_{ \lambda \geq 0} H(\lambda) := [ \: \lambda \delta + K_0(w) + K_1(\lambda,w) \: ]$. As such the solution is also the same, therefore one can invoke Propositions 2.2, 2.3 and Algorithm 1 to compute $\lambda^*$ and $H(\lambda^*)$.

\subsubsection{Outer Optimization Problem}
As before, in Section 2.1.3, the objective is to evaluate $v_s(\delta) := \sup_{ w \in \Gamma_{s} } \: \{ \lambda^*(w,\delta) \delta + K_0(w) + K_1(\lambda^*(w,\delta),w) \}$. As such the solution is also the same, therefore one can invoke Theorem 2.1 and Corollary 2.1.1 to evaluate the above expression(s). The analog to Theorem 2.2 is given below.
\begin{theorem}
The critical radius $\delta^{bc}_{\alpha}$ can be expressed as $\inf \{\delta \geq 0 :v_s(\delta) \geq \alpha^{bc}\}$. Furthermore, $\delta^{bc}_{\alpha}$ can be explicitly computed via binary search. Let $\delta < \delta^{bc}_{\alpha}$. For $Q \in \mathcal{U}_{\delta}(Q_N)$, it follows that $Q$ is (strong) statistical arbitrage free, for level $\alpha > v_s(\delta^{bc}_{\alpha})$. For $Q \notin \mathcal{U}_{\delta^{bc}_{\alpha}}(Q_N)$, it follows that $Q$ may admit (strong) statistical arbitrage for level $\alpha > v_s(\delta^{bc}_{\alpha})$.
\end{theorem}
\begin{proof}
This characterization of the critical radius $\delta^{bc}_{\alpha}$ follows from the condition \ref{strongSAD} as well as the definition of (strong) statistical arbitrage. The asymptotic properties of $v_{s}$ are such that $v_{s}(0) \leq 1$ and $\lim_{\delta \to \infty} v_{s}(\delta) = 1$. Furthermore, since $v_{s}(\delta)$ is a non-decreasing function of $\delta$, it follows that $\delta^{bc}_{\alpha}$ can be computed via binary search.
\end{proof}
\noindent One can view critical radius $\delta^{bc}_{\alpha}$ as a relative measure of the \textit{degree} of (strong) statistical arbitrage in reference measure $Q_N$. Those $Q_N$ which are ``close" to admitting statistical arbitrage of level $\alpha^{bc}$ will have a relatively smaller value of $\delta^{bc}_{\alpha}$. 

\subsection{Best Case Distribution for SA Problem}
The characterization of best case distributions for NA problems carries over into the SA context. In particular, one is interested in best case distributions $Q^{\alpha}_w \in BC(w,\delta^{\alpha})$ such that $\mathbb{E}^{Q^{\alpha}_w} [\, \mathbbm{1}_{\{ w \cdot S_1 \geq 0 \}} \,] = \sup_{ Q \in \mathcal{U}_{\delta^\alpha}(Q_N) } \mathbb{E}^Q[\, \mathbbm{1}_{\{ w \cdot S_1 \geq 0 \}} \,]$.
As before, by selecting portfolio weights $w^\alpha$ that satisfy the outer $\sup$ condition
\[
\sup_{ w \in \Gamma_{s} } \: { \sup_{ Q \in \mathcal{U}_{\delta^\alpha}(Q_N) } \mathbb{E}^Q [ \, \mathbbm{1}_{\{ w \cdot S_1 \geq 0 \}} \, ] } \geq \alpha^{bc},
\]
one \textit{can} recover $Q^{\alpha}_{w^{\alpha}}$ that admits statistical arbitrage of level $\alpha^{bc}$. See Section 6.2 for a concrete example.

\subsection{Robust SA Worst Case Conditions}
The robust (strong) statistical arbitrage worst case conditions (of level $\alpha^{wc} \in (0,1)$) can be expressed as
\begin{equation*}\label{strongSAPwc}
\sup_{ w \in \Gamma_s } \: { \inf_{ Q \in \mathcal{U}_{\delta}(Q_N) } \mathbb{E}^Q [ \, \mathbbm{1}_{\{ w \cdot S_1 \geq 0 \}} \, ] } \geq \alpha^{wc}, \tag{SSAP\textsuperscript{wc}}
\end{equation*}
where $\Gamma_s$ is defined in \ref{strongW}.
Relaxing the objective function from $\mathbbm{1}_{\{ w \cdot S_1 \geq 0 \}}$ to $\mathbbm{1}_{\{ w \cdot S_1 > 0 \}}$ and using the relations $\mathbbm{1}_{\{ w \cdot S_1 > 0 \}} = 1 - \mathbbm{1}_{\{ w \cdot S_1 \leq 0 \}}$ and $\inf(S) = -\sup(-S)$ for bounded set $S$, we have the equivalent condition:
\begin{equation*}\label{strongSAP2wc}
\sup_{ w \in \Gamma_s } \:  -\left\{ \sup_{ Q \in \mathcal{U}_{\delta}(Q_N) } \mathbb{E}^Q [ \, \mathbbm{1}_{\{ w \cdot S_1 \leq 0 \}} - 1 \, ] \right\} \geq \alpha^{wc}. \tag{SSAP2\textsuperscript{wc}}
\end{equation*}
As before, the indicator function $\mathbbm{1}_{\{ w \cdot S_1 \leq 0 \}}$ on closed set ${\{ w \cdot S_1 \leq 0 \}}$ is upper semicontinuous  hence we can apply the duality theorem (see Section 1.4.4) to obtain the dual formulation
\begin{equation*}\label{strongSADwc}
\sup_{ w \in \Gamma_s } \: - \left\{ \inf_{ \lambda \geq 0 } \: [ \: \lambda \delta + \frac{1}{N} \sum_{i=1}^N \Psi^{wc}_{\lambda,w} (s_i) \: ] \right\} \geq \alpha^{wc} \tag{SSAD\textsuperscript{wc}}
\end{equation*}
where $\Psi^{wc}_{\lambda,w}$ is defined, in terms of cost function $c$, as
$\Psi^{wc}_{\lambda,w} = \sup_{\tilde{s} \in \mathbb{R}^n} [ \, \mathbbm{1}_{\{ w \cdot \tilde{s} \, \leq \, 0 \}} - \lambda c( \tilde{s}, s_i) - 1\, ]$.

\subsubsection{Inner Optimization Problem}
The goal here is the same as for the robust no-arbitrage conditions in Section 2.1.1, namely to evaluate $\Psi^{wc}_{\lambda,w}$ in closed form. There are two cases to consider.
\setcounter{case}{0}
\begin{case}
\[ \mathbbm{1}_{\{ w \cdot s_i \leq 0 \}} = 1 \implies \Psi^{wc}_{ \lambda,w } (s_i) = 1 - \, \lambda \cdot 0 \, - 1 = 0 \quad \text{which is optimal.} \]
\end{case}

\begin{case}
\[ \mathbbm{1}_{\{ w \cdot s_i \leq 0 \}} = 0 \implies \Psi^{wc}_{ \lambda,w } (s_i) = [ 1 - \lambda c( s_i^*, s_i) ]^+ - 1 \quad \text{where} \quad s_i^* = \argmin \| \tilde{s} - s_i\|_2 \quad \text{is optimal.} \]
\end{case}
By the Projection Theorem \citep{calafiore2014optimization},  $\| s_i^* - s_i\|_2 = \frac{|w^\top s_i|}{\|w\|_2} \implies \Psi^{wc}_{ \lambda,w } (s_i) = [ 1 - \lambda c_i ]^+ - 1 \text{\; for \;} \\ c_i = \frac{| w^\top s_i | }{ \| w \|_2 } \in \mathbb{R}^n_+$.

\begin{prop}
\[ \frac{1}{N} \sum_{i=1}^N \Psi^{wc}_{\lambda,w}(s_i) = K_0^{wc}(w) + K_1^{wc}(\lambda,w) = K_1^{wc}(\lambda,w) \]
where \: $K_0^{wc}(w) = \frac{1}{N} \sum_{i=1}^N \mathbbm{1}_{\{ w \cdot s_i \leq 0 \}} \cdot 0$ \: and \: $K_1^{wc}(\lambda,w) = \frac{1}{N} \sum_{i=1}^N \mathbbm{1}_{\{ w \cdot s_i > 0 \}} ( [ 1 - \lambda c_i ]^+ - 1 )$ \: for \: $c_i = \frac{| w^\top s_i |}{ \| w \|_2 } \in \mathbb{R}^n_+$.
\end{prop}
\begin{proof}
This follows by a straightforward application of the two cases above.
\end{proof}

\subsubsection{Middle Optimization Problem}
As before, in Section 2.1.2, the objective is to evaluate $\inf_{ \lambda \geq 0} H^{wc}(\lambda) := [ \: \lambda \delta + K^{wc}_1(\lambda,w) \: ]$. As such the solution is also the same, with one \textit{exception}: replace $\mathbbm{1}_{\{ w \cdot s_i < 0 \}}$ with $\mathbbm{1}_{\{ w \cdot s_i > 0 \}}$ in those results. Therefore one can apply Propositions 2.2, 2.3 and Algorithm 1 (with the above replacement of indicator functions) to compute $\lambda^*$ and $H^{wc}(\lambda^*)$.

\subsubsection{Outer Optimization Problem}
As before, in Section 2.1.3, the objective is to evaluate $v^{wc}_s(\delta) := \sup_{ w \in \Gamma_{s} } \: -\{ \lambda^*(w,\delta) \delta + K^{wc}_1(\lambda^*(w,\delta),w) \}$. As such the solution is similar, with the following adjustments: replace $F_k(w)$ with $-F_k^{wc}(w)$ where
\[
-F_k^{wc}(w) := \lambda_k \delta + \frac{1}{N} \bigg[ \sum_{i=1}^N (z^+_i - 1) \mathbbm{1}_{\{ w \cdot s_i > 0 \}} \bigg]
\] and place a \textit{minus} sign in front of the $\min$ term in the \textit{maximin} expression for $v_{w(s)}(\delta)$.
Therefore one can apply Theorem 2.1 and Corollary 2.1.1 (with the above adjustments) to evaluate $v^{wc}_s(\delta)$. 
The revised formulation is shown below.
\begin{theorem}
$v^{wc}_s(\delta)$ is approximated by the (global) solution to nonlinear program (NLP) N\_SSA (listed below).
\end{theorem}
\noindent The constraints on variables below, with index $i$, apply for $i \in \{1,\dots,N\}$, although this is suppressed. 
\begin{maxi}|l|<b>
{\substack{w \in \mathbb{R}^n} }
{\max_{\lambda_k \: : \: k \in \{0,1,...,N\}} \quad F^{wc}_k(w) = -\lambda_k \delta + \frac{1}{N} \bigg[ \sum_{i=1}^N (1 - z^+_i) \mathbbm{1}_{\{ w \cdot s_i > 0 \}} \bigg] }{\label{N_SSA}}
{v^{wc}_s(\delta)=}
\addConstraint{c_i}{= \frac{| w^\top s_i |}{\| w \|_2}}
\addConstraint{\lambda_k}{=\frac{1}{c_k} \quad \forall k \in \{1,\dots,N\}}
\addConstraint{\lambda_0}{=0}
\addConstraint{|w_i|}{\leq M }
\addConstraint{w \cdot S_0}{\leq -\epsilon}
\addConstraint{z_i}{= [1-\lambda_k c_i]}
\end{maxi}
\begin{proof}
The NLP formulation follows from the definition of $v^{wc}_s$ and the fact that $\lambda^* \in \{ \frac{1}{c_k }: k \in \{1,\dots,N\} \} \cup \{ \lambda_0 := 0 \}$.
\end{proof}
\noindent The analog to Theorem 2.2 is given below.
\begin{theorem}
The critical radius $\delta^{wc}_{\alpha}$ can be expressed as $\inf \{\delta \geq 0 :v^{wc}_s(\delta) \leq \alpha^{wc}\}$. Furthermore, $\delta^{wc}_{\alpha}$ can be explicitly computed via binary search. Let $\delta < \delta^{wc}_{\alpha}$. For $Q \in \mathcal{U}_{\delta^{wc}_{\alpha}}(Q_N)$, it follows that $Q$ admits (strong) statistical arbitrage, for level $\alpha \geq v^{wc}_s(\delta^{wc}_{\alpha})$. For $Q \notin \mathcal{U}_{\delta^{wc}_{\alpha}}(Q_N)$, it follows that $Q$ may not admit (strong) statistical arbitrage for level $\alpha < v^{wc}_s(\delta^{wc}_{\alpha})$. 
\end{theorem}
\begin{proof}
This characterization of the critical radius $\delta^{wc}_{\alpha}$ follows from the condition (\ref{strongSADwc}) as well as the definition of (strong) statistical arbitrage. The asymptotic properties of $v^{wc}_{s}$ are such that $v^{wc}_{s}(0) > 0$ and $\lim_{\delta \to \infty} v^{wc}_{s}(\delta) = 0$. Furthermore, since $v^{wc}_{s}(\delta)$ is a non-increasing function of $\delta$, it follows that $\delta^{wc}_{\alpha}$ can be computed via binary search.
\end{proof}
\noindent One can view critical radius $\delta^{wc}_{\alpha}$ as a relative measure of the \textit{degree} of (strong) statistical arbitrage in reference measure $Q_N$. Those $Q_N$ which are ``close" to not admitting statistical arbitrage of level $\alpha^{wc}$ will have a relatively smaller value of $\delta^{wc}_{\alpha}$. 

\subsection{Worst Case Distribution for SA Problem}
The characterization of worst case distributions for NA problems carries over into the SA context. In particular, one is interested in worst case distributions $Q^{\alpha}_w \in WC(w,\delta^{\alpha})$ such that $\mathbb{E}^{Q^{\alpha}} [\, \mathbbm{1}_{\{ w \cdot S_1 \geq 0 \}} \,] = \inf_{ Q \in \mathcal{U}_{\delta^\alpha}(Q_N) } \mathbb{E}^Q[\, \mathbbm{1}_{\{ w \cdot S_1 \geq 0 \}} \,]$.
By selecting portfolio weights $w$ with their associated worst case distributions, it follows that
\[
\sup_{ w \in \Gamma_{s} } \: { \mathbb{E}^{Q^{\alpha}_w} [ \, \mathbbm{1}_{\{ w \cdot S_1 \geq 0 \}} \, ] } \leq \alpha^{wc}.
\]
Applying the greedy algorithm to $\mathbbm{1}_{\{ w \cdot S_1 < 0 \}} = 1 -\mathbbm{1}_{\{ w \cdot S_1 \geq 0 \}}$, one \textit{can} recover $Q^{\alpha}_{w}$ that is the most punitive for $w$ and admits statistical arbitrage of level at most $\alpha^{wc}$ for a given $w \in \Gamma_s$. See Section 6.2 for a concrete example.

\subsection{Portfolio Restrictions, SA Under No Short Sales}
The portfolio restrictions for NA problems apply in the SA context as well. We refer the reader to Section 2.3 and do not duplicate the material here.
The Farkas Lemma characterization of classical weak (strong) no arbitrage via the existence (and uniqueness for complete markets) of risk neutral measures does not yield any \textit{new} relationships in the context of statistical arbitrage under no short sales. As such, we do not establish any \textit{new} results in this subsection. Note that the theorem given in Section 2.4 still holds for probability measures $Q^\alpha \: \text{for} \: \alpha \in (0,1)$; in words, it holds for market models that admit statistical arbitrage but \textit{not} classical arbitrage. 


\subsection{Nearest SA Problem}
As above, the Farkas Lemma characterization does not yield any \textit{new} relationships for the nearest no-arbitrage problem in the context of statistical arbitrage. However, the nuances of how one uses the existing results in Section 2.5 (vs. Section 2.4) are different. In particular, one can apply those results for probability measures $Q^\alpha \: \text{for} \: \alpha = 1$; in words, it holds for market models that \textit{admit} classical arbitrage. 

\subsection{Alternate Robust SA Conditions}
The concept of exchanging the order of the $\sup$ and $\inf$ operators for the robust NA conditions (see Section 2.6) can be extended to cover SA. As before, exchanging the order of the operators gives the robust SA best case conditions
\begin{equation*}\label{strongRSAPbc}
\sup_{ Q \in \mathcal{U}_{\delta}(Q_N) } \: \sup_{ w \in \Gamma_s } \: \mathbb{E}^Q [ \, \mathbbm{1}_{\{ w \cdot S_1 \geq 0 \}} \, ] \leq \alpha^{bc}. \tag{RSSAP\textsuperscript{bc}}
\end{equation*}
Similarly, an alternate formulation of the robust SA worst case conditions is
\begin{equation*}\label{strongRSAPwc}
\inf_{ Q \in \mathcal{U}_{\delta}(Q_N) } \: \sup_{ w \in \Gamma_s } \: \mathbb{E}^Q [ \, \mathbbm{1}_{\{ w \cdot S_1 \geq 0 \}} \, ] \geq \alpha^{wc}. \tag{RSSAP\textsuperscript{wc}}
\end{equation*}
The intuitive meaning of these formulations is that the market adversary \textit{first} chooses a punitive distribution $Q \in \mathcal{U}_{\delta}(Q_N)$ and \textit{then} the portfolio manager chooses an optimal $w \in \Gamma_s$. Although one can invoke the $\min$-$\max$ inequality to establish the relation 
\[
\inf_{ Q \in \mathcal{U}_{\delta}(Q_N) } \: \sup_{ w \in \Gamma_s } \: \mathbb{E}^Q [ \, \mathbbm{1}_{\{ w \cdot S_1 \geq 0 \}} \, ] \geq \:
\sup_{ w \in \Gamma_s } \: { \inf_{ Q \in \mathcal{U}_{\delta}(Q_N) } \mathbb{E}^Q [ \, \mathbbm{1}_{\{ w \cdot S_1 \geq 0 \}} \, ] },
\]
finding a method to compute the LHS of \ref{strongRSAPbc} or \ref{strongRSAPwc} is not really achievable (to our knowledge) since the inner problem is \textbf{NP} Hard (see Section 5 for a proof) and the outer problem is infinite dimensional.

\section{Applications}
Section 4 presents applications of the theory developed in Sections 2 and 3 to robust option pricing and robust portfolio selection. In the latter we consider two examples: the classical Markowitz problem and a more modern view of risk using \text{CVaR} (as opposed to variance) as the measure of risk.

\subsection{Robust Option Pricing}
This subsection is a refinement (simplification) of the result for robust pricing of European options given in \citet{Bartl17}. For clarity, we adopt the notation and problem setup of Example 2.14 (Robust Call) \citep{Bartl17}. The approach taken there is to add an additional constraint on the probability measure $\mu$ to reside within Wasserstein radius $\delta$ of the reference (arbitrage-free) measure $\mu_0$. For this example, let us assume $\mu_0$ is arbitrage-free, distance function $d_c$ is the second order Wasserstein distance with associated quadratic cost function $c(x,y) = (x-y)^2/2$, $\mathcal{M}_1(\mathbb{R})$ denotes the set of probability measures on $\mathbb{R}$, and the penalty function is $\phi := \infty \mathbbm{1}_{(\delta,\infty]}$ with associated convex conjugate $\phi^*(\lambda) = \lambda \delta$. The authors show that the robust call option with maturity $T$, strike $k$, on a single asset, satisfies the relation:
\begin{equation}\label{call:rob}
\text{CALL}^{\text{robust}}(k) = \sup_{\{ \mu \in \mathcal{M}_1(\mathbb{R}): \int_{\mathbb{R}} S d\mu = s \}} \text{CALL}(k) - \phi(d_c(\mu_0,\mu)) = \inf_{\beta \in \mathbb{R}}   \inf_{\lambda > 0} \left\{ \lambda \delta + \text{CALL}(k - (2\beta+1)/(2\lambda)) + \beta^2/(2\lambda) \right\}
\end{equation}
where $\beta$ denotes the Lagrange multiplier for the arbitrage-free probability measure constraint $\{ \mu \in \mathcal{M}_1(\mathbb{R}): \int_{\mathbb{R}} S d\mu = s \}$, and $\lambda$ denotes the Lagrange multiplier for the Wasserstein distance constraint $d_c(\mu_0,\mu) \leq \delta$. Here $\text{CALL}(\tilde{k})$ denotes the non-robust call option price for strike $\tilde{k}$. Now let us assume that we have calculated the critical radius $\delta^*_{w(s)}$ for this problem (assume the reference measure $\mu_0$ is empirical) and we have chosen $\delta^\alpha < \min(\delta^*_w,\delta^*_s)$. Here $\delta^\alpha$ denotes the radius of a Wasserstein ball of probability measures that allow statistical arbitrage (up to some level $\alpha < 1$) but not classical arbitrage. It follows from Theorem 2.2 that the arbitrage-free probability measure constraint is not needed, hence one can simply set $\beta := 0$ in the above formula \ref{call:rob} to reduce it to the simpler formula:
\begin{equation}\label{call:rob2}
\text{CALL}^{\text{robust}}(k) = \inf_{\lambda > 0} G(\lambda) := \big\{ \lambda \delta^\alpha + \text{CALL}(k - 1/(2\lambda)) \big\}.
\end{equation}
Note that in formula \ref{call:rob2} above, $G(\lambda)$ is convex in $\lambda$. Once again, following the approach in our earlier work \citep{Singh19xva1}, we can simplify further to arrive at the following result.
\begin{prop}
$
Let \: \lambda^{*} = \sup_{\lambda \geq 0} \{ \lambda : \delta^\alpha - \frac{1}{N} [ \sum_{i \in J^+_1(\lambda)} 1/(2\lambda^2) ] \leq 0 \} = \inf_{\lambda \geq 0} \{  \lambda :  \delta^\alpha - \frac{1}{N} [ \sum_{i \in J_1(\lambda)} 1/(2\lambda^2) ] \geq 0 \}, where\\
$
$J^+_1(\lambda) \: = \{ i \in \{1,\dots,N\} : [1/(2\lambda) + s_i - k] > 0 \}$, 
$J_1(\lambda) \: = \{ i \in \{1,\dots,N\} : [1/(2\lambda) + s_i - k] \geq 0 \}$. 
\end{prop}
\begin{proofsketch}
This result follows from writing down the first order conditions for left and right derivatives for convex objective function $G(\lambda)$. Inspection of the left and right derivatives for $G(\lambda)$ reveals that they will cross zero (as $\lambda$ sweeps from 0 to $\infty$) and hence the $\sup$ and $\inf$ operators will apply over non-empty sets. For each index $i \in J^+_1 (J_1)$ we pick up another $1/(2\lambda^2)$ term in the left (right) derivative. Search on $\lambda$ (from the left or the right) until we find $\lambda^*$ such that $0 \in \partial G(\lambda^*)$.
\end{proofsketch}
\begin{proof}
The first order optimality condition says
\[ \delta^\alpha - \frac{1}{N} \sum_{i \in J^+_1(\lambda)} 1/(2\lambda^2) \leq 0 \leq \delta^\alpha - \frac{1}{N} \sum_{i \in J_1(\lambda)} 1/(2\lambda^2). \]
Note the LHS is an increasing function in $\lambda$. Hence one can write
\[ \lambda^* = \sup_{\lambda \geq 0} \{ \lambda : \delta^\alpha - \frac{1}{N} \sum_{i \in J^+_1(\lambda)} 1/(2\lambda^2) \leq 0 \}. \]
Similarly the RHS is also an increasing function in $\lambda$. Equivalently, one can write
\[ \lambda^* = \inf_{\lambda \geq 0} \{ \lambda : \delta^\alpha - \frac{1}{N} \sum_{i \in J_1(\lambda)} 1/(2\lambda^2) \geq 0 \}. \]
\end{proof}

\begin{corollaryP}
\[ \text{CALL}^{\text{robust}}(k) = G(\lambda^*) := \big[ \lambda^* \delta^\alpha + \text{CALL}(k - 1/(2\lambda^*)) \big] \]
where  $\lambda^*$ is given by Proposition 4.1 above.
\end{corollaryP}
\begin{proof}
This follows by direct substitution of $\lambda^*$ from Proposition 4.1 into formula \ref{call:rob2} above.
\end{proof}


\subsection{Robust Portfolio Selection}
\subsubsection{Robust Markowitz Portfolio Selection} 
This subsection is a refinement of the result(s) for robust Markowitz (mean variance) portfolio selection given in \citet{blanchetMV}. For clarity, we adopt the notation and problem setup of that paper. The convex primal problem is a distributionally robust Markowitz problem given by
\begin{equation}
\min_{\phi \in \mathcal{F}_{\delta,\bar{r}}(N)} \:\: \max_{P \in \mathcal{U}_\delta(P_N)} \{ \phi^\top \mbox{\text{Var}}_P(R) \phi \}
\end{equation}
where $\phi \in \mathbb{R}^d$ denotes the portfolio weight vector, $R \in \mathbb{R}^d$ denotes the random (gross) asset returns, $P_N$ denotes the empirical measure, $\mathcal{U}_\delta(P_N)$ denotes the uncertainty set for probability measures, with associated cost function $c(u,v) = \| v - u \|^2_q$ for $q \geq 1$, $Var_P(R)$ denotes the covariance matrix of returns under $P$, and $\mathcal{F}_{\delta,\bar{r}}(N) = \{ \phi : \phi^\top 1 = 1 \: ; \: \min_{P \in \mathcal{U}_\delta(P_N)} \mathbb{E}_P (\phi^\top R) \geq \bar{r} \}$ denotes the feasible region for portfolios. Using Lagrangian duality techniques (similar to this paper) the authors show that this primal problem is equivalent to the convex dual problem
\begin{equation}
\min_{\phi \in \mathcal{F}_{\delta,\bar{r}}(N)} \:\: \bigg( \sqrt{ \phi^\top \mbox{\text{Var}}_{P_N}(R) \phi } + \sqrt{\delta} \| \phi \|_p  \bigg)^2
\end{equation}
in terms of optimal value and solution(s), with $1/p + 1/q = 1$. Following the approach in the previous subsection, let us assume the reference measure $P_N$ is arbitrage-free and we have chosen $\delta^\alpha < \min(\delta^*_w,\delta^*_s)$. Again $\delta^\alpha$ denotes the radius of a Wasserstein ball of probability measures that allow statistical arbitrage (up to some level $\alpha < 1$) but not classical arbitrage. It follows from Theorem 2.2 that the arbitrage-free probability measure constraint is not needed hence the \textit{arbitrage-free} primal problem 
\begin{equation}
\min_{\phi \in \mathcal{F}_{\delta^\alpha,\bar{r}}(N)} \:\: \max_{P \in \mathcal{\tilde{U}}_{\delta^\alpha}(P_N)} \{ \phi^\top \mbox{\text{Var}}_P(R) \phi \}
\end{equation}
where $\mathcal{\tilde{U}}_{\delta^\alpha}(P_N) = \mathcal{U}_{\delta^\alpha}(P_N) \cap \{ P : \sup_{ \phi \in \{ \Gamma_w \cup \Gamma_s \} } \: \mathbb{E}^P [ \, \mathbbm{1}_{\{ \phi \cdot S_T \geq 0 \}} \, ] < 1 \}$ is equivalent to the primal and dual problems above. In this setting $R = R(S_0,S_T)$ is the random vector of asset returns calculated based on initial asset prices $S_0 \in \mathbb{R}^d$ and terminal asset prices $S_T \in \mathbb{R}^d$.
\subsubsection{Robust Mean Risk Portfolio Selection}
This subsection is a refinement of the result(s) for robust mean risk portfolio optimization given in \citet{Esfahani17}. For clarity, we adopt the notation and problem setup of that paper. Let $\xi \in \mathbb{R}^m$ denote a random vector of (gross) asset returns and $x \in \mathbb{X}$ denote a vector of portfolio percentage weights ranging over the probability simplex $\mathbb{X} = \{ x \in \mathbb{R}^m_{+} : \sum_{i=1}^m x_i = 1 \}$. Thus we consider a ``long only" portfolio. However, the reader is advised that  today's market includes securities such as exchange traded funds (ETFs) that behave like short positions hence the long portfolio setting is not as restrictive as it might seem at first glance. The portfolio return is given by $\langle x, \xi \rangle$. A single stage stochastic program which minimizes a weighted sum of the mean and \text{CVaR} of portfolio loss at confidence level $\bar{\alpha} \in (0,1]$, given the investor's risk aversion $\rho \in \mathbb{R}_{+}$ and distribution $\mathbb{P}$ is given by
\begin{equation}
J^* = \inf_{x \in \mathbb{X}} \mathbb{E}^{\mathbb{P}} [ - \langle x, \xi \rangle + \rho \: \text{CVaR}_{\bar{\alpha}}( - \langle x, \xi \rangle ) ].
\end{equation}
Substituting the formal definition of \text{CVaR} into the above, they show that 
\begin{align}
J^* &= \inf_{x \in \mathbb{X}} \mathbb{E}^{\mathbb{P}} [ - \langle x, \xi \rangle ] + \rho \inf_{\tau \in \mathbb{R}} \mathbb{E}^{\mathbb{P}} [ \tau + (1/\bar{\alpha}) \:\: \max ( - \langle x, \xi \rangle - \tau, 0) ] \\
	  &= \inf_{x \in \mathbb{X}, \tau \in \mathbb{R}} \mathbb{E}^{\mathbb{P}} [  \max_{k \leq K} a_k \langle x, \xi \rangle + b_k \tau ]
\end{align}
where $K = 2, \: a_1 = -1, \: a_2 = -1 - (\rho/\bar{\alpha}), \: b_1 = \rho, \: b_2 = \rho (1 - (1/\bar{\alpha}))$.

For Wasserstein ambiguity set $\mathbb{B}_{\epsilon}(\hat{\mathbb{P}}_N)$ of radius $\epsilon$ about reference measure $\hat{\mathbb{P}}_N$, the authors formulate the distributionally robust primal problem
\begin{equation}
\hat{J}_N(\epsilon) := \inf_{x \in \mathbb{X}} \:\: \sup_{\mathbb{Q} \in \mathbb{B}_{\epsilon}(\hat{\mathbb{P}}_N)} \mathbb{E}^{\mathbb{Q}} [ - \langle x, \xi \rangle + \rho \: \text{CVaR}_{\bar{\alpha}}( - \langle x, \xi \rangle ) ]
\end{equation}
Applying techniques of Lagrangian duality, Esfahani and Kuhn formulate the equivalent dual problem
\begin{equation}
\hat{J}_N(\epsilon) = 
	\begin{cases}
		\inf\limits_{x,\tau,\lambda,s_i,\gamma_{ik}} \:\: \lambda \epsilon + \frac{1}{N} \sum_{i=1}^N s_i \\
		\text{such that} \:\: x \in \mathbb{X}, \\
		\quad \quad \quad \quad b_k \tau + a_k \langle x, \hat{\xi_i} \rangle + \langle \gamma_{ik}, d - C \hat{\xi_i} \rangle \leq s_i, \\
		\quad \quad \quad \quad \| C^\top \gamma_{ik} - a_k x \|_{*} \leq \lambda, \\
		\quad \quad \quad \quad \gamma_{ik} \geq 0
	\end{cases}
\end{equation}
$\forall i \leq N, \forall k \leq K$. Following the approach in the previous subsection, let us assume the reference measure $\hat{\mathbb{P}}_N$ is arbitrage-free and we have chosen $\epsilon^\alpha < \min(\epsilon^*_w,\epsilon^*_s)$. As before $\epsilon^\alpha$ denotes the radius of a Wasserstein ball of probability measures that allow statistical arbitrage (up to some level $\alpha < 1$) but not classical arbitrage. It follows from Theorem 2.2 that the arbitrage-free probability measure constraint is not needed hence the \textit{arbitrage-free} primal problem 
\begin{equation}
 \inf_{x \in \mathbb{X}} \:\: \sup_{\mathbb{Q} \in \tilde{\mathbb{B}}_{\epsilon^\alpha}(\hat{\mathbb{P}}_N)} \mathbb{E}^{\mathbb{Q}} [ - \langle x, \xi \rangle + \rho \: \text{CVaR}_{\bar{\alpha}}( - \langle x, \xi \rangle ) ]
\end{equation}
where $\tilde{\mathbb{B}}_{\epsilon^\alpha}(\hat{\mathbb{P}}_N) = \mathbb{B}_{\epsilon^\alpha}(\hat{\mathbb{P}}_N) \cap \{ \mathbb{Q} : \sup_{\phi(x) \in \{ \Gamma_w \cup \Gamma_s \} } \: \mathbb{E}^{\mathbb{Q}} [ \, \mathbbm{1}_{\{ \phi(x) \cdot \tilde{S}_T \geq 0 \}} \, ] < 1 \}$ is equivalent to the primal and dual problems above. In this setting, $\xi = \xi(S_0,S_T)$ is the random vector of asset returns calculated based on initial asset prices $S_0 \in \mathbb{R}^m$ and terminal asset prices $S_T \in \mathbb{R}^m$. Also, $\tilde{S}_0 = (S_0,B_0)$ appends the initial cash bond (borrowing) $B_0$ used to purchase the portfolio (at zero or negative cost) and $\tilde{S}_T = (S_T,B_T)$ appends the bond repayment (principal plus interest) at the end of the investment period.  Finally, $\phi(x) \in \mathbb{R}^{m+1} \: \text{for} \:\, x \in \mathbb{X}$  denotes the portfolio weight vector corresponding to the portfolio purchase and cash loan. By construction, the first $m$ components of $\phi$ are non-negative whereas the last component has a negative sign.

\section{Complexity of NA Problem}
This section gives formal proofs for the complexity of the No-Arbitrage Problem.
We establish that the weak and strong no-arbitrage problems are $\mathbf{NP}$ Hard.
\noindent The approach taken here is to use reduction on the known $\mathbf{NP}$ complete closed (open) hemisphere decision problem \citep{johnson1978densest}.
The optimization problem, using the notation of this paper, is stated below \citep{avis2005graph}.
\begin{enumerate}
\item closed hemisphere:
\[ \text{Find} \: w \in \mathbb{R}^n \: \text{such that} \:\: \mathbf{card}( \{i : s_i \in S ; \: w \cdot s_i \geq 0 \} ) \: \text{is maximized.} \]
\item open hemisphere:
\[ \text{Find} \: w \in \mathbb{R}^n \: \text{such that} \:\: \mathbf{card}( \{i : s_i \in S ; \: w \cdot s_i > 0 \} ) \: \text{is maximized.} \]
To complete the problem statement, note that the set $S$ above denotes a finite subset of $\mathbb{Q}^n$ containing $N$ points.
It follows that the \textit{mixed} hemisphere problem (where $c_i \geq 0 \:\:\: \forall i$) is also $\mathbf{NP}$ complete.
\item mixed hemisphere:
\begin{equation*}\label{mixed1}
 \sup_{ w \in \mathbb{R}^n } \: \big[ \sum_{i=1}^N \mathbbm{1}_{\{ w \cdot s_i \geq 0 \}} + \sum_{i=1}^N c_i \mathbbm{1}_{\{ w \cdot s_i < 0 \}} \big]. \tag{M1}
\end{equation*}
\end{enumerate}
One can write a \textit{simplified} version of the weak and strong no-arbitrage optimization problems as follows (see Section 2.1.3). To construct these simplified versions, we have fixed $\lambda^*$ to a constant, relabeled $[1-\lambda c_i]^+$ as $c_i$, and omitted the initial cost constraint $w \cdot S_0 = \kappa$. Recall $\kappa$ is zero in the weak case, but strictly less than zero (for arbitrary $\kappa$) in the strong case.
\begin{equation*}\label{mixed2}
 \sup_{ w \in \mathbb{R}^n } \: F(w) := \big[ \sum_{i=1}^N \mathbbm{1}_{\{ w \cdot s_i \geq 0 \}} + \sum_{i=1}^N c_i \mathbbm{1}_{\{ w \cdot s_i < 0 \}} \big]. \tag{WD3,SD3}
\end{equation*}
However, there is some work to be done to incorporate the initial cost constraint \textit{back} to formulate the no-arbitrage problems. First, think of the unconstrained initial cost as the union of three possibilities: (i) $w \cdot S_0 < 0$, (ii) $w \cdot S_0 = 0$, (iii) $w \cdot S_0 > 0$.
Some thought suggests that the following proposition holds.
\begin{prop}
Asumming $\mathbf{P} \neq \mathbf{NP}$, there can be no polynomial time algorithm to solve the simplified no-arbitrage problem under initial cost constraint $w \cdot S_0 \leq \kappa$ \, for $\kappa \in \mathbb{R}$.
\end{prop}
\begin{proof}
Proceed by contradiction. Suppose there is a polynomial time algorithm $A$ that can solve the following problem:
\begin{equation*}\label{mixed2}
 \sup_{ w \in \mathbb{R}^n, w \cdot S_0 \leq \kappa } \: F(w). \tag{M2}
\end{equation*}
Exploiting symmetry, one can then also use algorithm $A$ to solve this problem:
\begin{equation*}\label{mixed3}
 \sup_{ w \in \mathbb{R}^n, w \cdot S_0 \geq \kappa } \: F(w). \tag{M3}
\end{equation*}
Returning the better answer now gives us a polynomial time algorithm to solve the mixed hemisphere problem which contradicts $\mathbf{P} \neq \mathbf{NP}$. Hence it must be that there is no polynomial time algorithm $A$ to solve either \ref{mixed2} or \ref{mixed3}.
\end{proof}

\begin{corollaryP}
Asumming $\mathbf{P} \neq \mathbf{NP}$, there can be no polynomial time algorithms to solve \textbf{both} the weak and strong no-arbitrage problems.
\end{corollaryP}
\begin{proof}
This follows directly from the definitions of the weak and strong no-arbitrage conditions (see Section 1.3.1).
\end{proof}

\begin{corollaryP}
Asumming $\mathbf{P} \neq \mathbf{NP}$, there can be no polynomial time algorithms to solve \textbf{either} the weak \textbf{or} strong no-arbitrage problems.
\end{corollaryP}
\begin{proof}
Recall that weight vectors $w$ satisfy the homogeneity property. Hence the optimal solution to the strong no-arbitrage problem is invariant to the actual choice of $\kappa$ up to the sign. In other words, we have the following relation:
\begin{equation*}\label{mixed4}
 \sup_{ w \in \mathbb{R}^n, \: w \cdot S_0 < 0 } \:  F(w) = \sup_{ w \in \mathbb{R}^n, \: w \cdot S_0 = \kappa < 0 (\kappa \: \text{arbitrary})} \: F(w). \tag{M4}
\end{equation*}
Furthermore, the RHS formulation above is equivalent in form to the weak no-arbitrage problem.
\end{proof}

\section{Computational Study}

This computational study uses the Matlab \textit{fminimax} and \textit{fmincon} solvers to work out a couple of concrete examples to find the critical radii at the cusp of (statistical) arbitrage assuming short sales are \textit{allowed}. Best (worst) case distributions and optimal portfolios are computed as well. Suitable values (for the problem instances below) for $M$ range from 100 to 10,000 and for $\epsilon$ from 0.001 to 0.0001. Other choices may be suitable. 
Recall that Matlab \textit{fminimax} solves to \textit{local} optimality using a sequential quadratic programming (SQP) method with modifications \citep{fletcher2010sequential}. Similarly, \textit{fmincon} solves to local optimality using gradient based techniques. Our algorithm incorporates a few \textit{additional} features to improve the robustness of the approach. These are listed next.
\begin{enumerate}
\item multi search: multiple search paths (that evolve candidate solutions) are used, similar to a genetic algorithm.
\item hot start: the optimal portfolio from the previous run $\delta_{prev}$ becomes the initial portfolio for the next run $\delta_{next}$. 
\item function smoothing: the indicator function can be relaxed using a sigmoid with appropriate scale factor.
\end{enumerate}

\noindent As mentioned in Section 2, developing an approach to solve for \textit{global} optimality would be a topic for further research. Meanwhile, for this section, the computed values for $v_{w(s)}$ and corresponding critical values for $\delta^*_{w(s)}$ represent \textit{local} optimality (upper bounds for globally optimal $\delta^*_{w(s)}$). This comment also applies for the statistical arbitrage calculations for $\delta^{bc}_\alpha$ and $\delta^{wc}_\alpha$.

\subsection{Binomial Tree Asset Pricing}
For the first example, consider the simple setting of a one-period binomial tree asset pricing model. There is a riskless bond priced at par at time zero that earns a deterministic risk free rate of return r at time 1. In addition there is a risky asset (stock) with initial price $s_0$ and time 1 price $s_u = u s_0$ that occurs with probability $p = 1/2$ and price $s_d = d s_0$ that occurs with probability $q = 1-p = 1/2$. The (weak) no-arbitrage conditions can be stated as: $0 < d < 1+r < u$ \citep{shreve2005stochastic}. Let us mock up an example to satisfy these conditions. Consider the problem setting below. Here $0 < d=0.966... < 1+r \in \{ 0.995, 1.005 \} < u = 1.0333...$ thus the conditions are satisfied. Intuitively the investor could either make or lose money depending on what happens.

\begin{figure}[!htb]
\caption{One-Period Binomial Tree}
\begin{adjustbox}{center}
\begin{tikzpicture}[>=stealth,sloped]
    \matrix (tree) [%
      matrix of nodes,
      minimum size=1cm,
      column sep=3.5cm,
      row sep=1cm,nodes={text width=8em}
          ]
    {
          &   &  \\
          & {Stock =\$310\\ Bond =\$100.5} &   \\
     {Stock = \$300\\ Bond = \$100} &   &  \\
          &  {Stock =\$290\\ Bond =\$99.5} &   \\
          &   &  \\
    };
    \node[bullet,right=0mm of tree-3-1.east](b-3-1){};
    \node[bullet,left=0mm of tree-2-2.west](b-2-2){};
    \node[bullet,left=0mm of tree-4-2.west](b-4-2){};
    \draw[->] (b-3-1) -- (b-2-2) node [midway,above] {$p$};
    \draw[->] (b-3-1) -- (b-4-2) node [midway,below] {$q=(1-p)$};
\end{tikzpicture}
\end{adjustbox}
\end{figure}
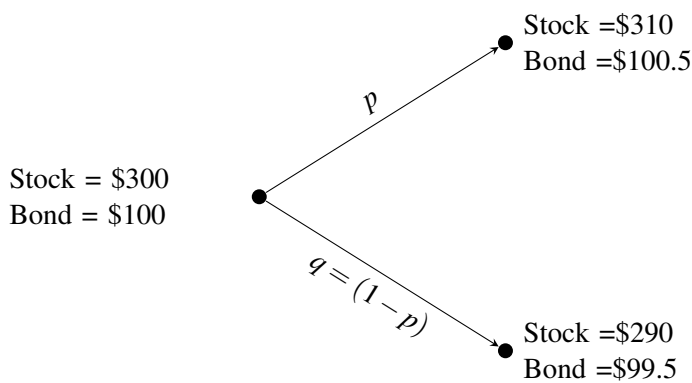


Solving NLP N\_WNA (see Theorem 2.1) for various values of $\delta$ gives the results in Table 2 (including the optimal portfolios). The critical radius $\delta^*_w$ from Theorem 2.2 is \textit{at most} 1.5. Solving NLP N\_SNA (see Corollary 2.1.1) for various values of $\delta$ gives the results in Table 3. The critical radius $\delta^*_s$ is \textit{at most} 1.5 as well. For this problem setup it appears that weak and strong arbitrage occur together. A plot of these values (from both tables) is shown 
in Figure 2 below.

\begin{table}[!htb]
\begin{center}
\caption{$v_w(\delta) < 1$: Weak No-Arbitrage Condition}
\begin{tabular}{ |r|r|r|r|r|r|r|r| }
 \hline
$\delta$ & 0.001 & 0.1 & 0.25 & 0.5 & 1.0 & 1.25 & 1.5 \\
 \hline
$v_w$ & 0.50 & 0.54 & 0.59 & 0.69 & 0.87 & 0.97 & 1.0 \\
 \hline
$w_{stock}$ & 1.3 & -0.7 & -0.7 & -0.7 & -0.7 & -0.7 & -0.5 \\
 \hline
$w_{bond}$ & -3.9 & 2.1 & 2.1 & 2.1 & 2.1 & 2.1 & 1.5 \\
 \hline
\end{tabular}
\end{center}
\end{table}


\begin{table}[!htb]
\begin{center}
\caption{$v_s(\delta) < 1$: Strong No-Arbitrage Condition}
\begin{tabular}{ |r|r|r|r|r|r|r|r| }
 \hline
$\delta$ & 0.001 & 0.1 & 0.25 & 0.50 & 1.0 & 1.25 & 1.5 \\
 \hline
$v_s$ & 0.50 & 0.54 & 0.59 & 0.69 & 0.87 & 0.97 & 1.0 \\
 \hline
$w_{stock}$ & 1.3 & 188 & 188 & 188 & -300 & -300 & -82 \\
 \hline
$w_{bond}$ & -3.9 & -565 & -565 & -565 & 899 & 899 & 247 \\
 \hline
\end{tabular}
\end{center}
\end{table}

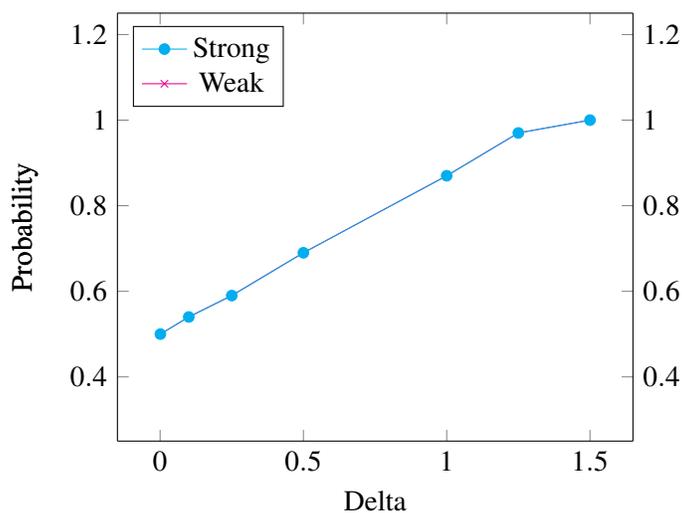
\begin{figure}[H]
\caption{Arbitrage Probabilities for One-Period Binomial Asset Pricing}
\begin{center}
\begin{tikzpicture}
	\begin{axis}[legend pos=north west,
		xlabel=Delta,
		ylabel=Probability,
		ymin = 0.25, ymax = 1.25,
		axis y line* = left ]
	\addplot[color=magenta,mark=x] coordinates {
		(0.001,0.5)
		(0.1,0.54)
		(0.25,0.59)
		(0.50,0.69)
		(1.0,0.87)
		(1.25,0.97)
		(1.5,1.0)
	}; \label{plot3_y1}
	\end{axis}

	\begin{axis}[legend pos=north west,
		xlabel=Delta,
		ylabel=Probability,
		ymin = 0.25, ymax = 1.25,
		axis y line* = right,
		axis x line = none ]
	\addplot[color=cyan,mark=*] coordinates {
		(0.001,0.5)
		(0.1,0.54)
		(0.25,0.59)
		(0.50,0.69)
		(1.0,0.87)
		(1.25,0.97)
		(1.5,1.0)
	}; \label{plot3_y2}
	\addlegendimage{/pgfplots/refstyle=plot3_y1}\addlegendentry{Strong}
    \addlegendimage{/pgfplots/refstyle=plot3_y2}\addlegendentry{Weak}

	\end{axis}
	
\end{tikzpicture}
\end{center}
\end{figure}

\subsection{Pairs Trading}

A typical example of a pairs trade would be to trade a linear combination of cointegrated tickers. The idea is to exploit temporary divergence from the long run relationship in the belief that convergence to the long run mean will result in a profitable trading strategy \citep{wojcik2005pairs}. The following annual data set of month end closing prices is taken from Yahoo finance website.

\begin{table}[!htb]
\begin{center}
\caption{U.S. Tech Pair Market Data 2019}
\begin{tabular}{ |c|c|c|c|c|c|c| }
 \hline
Date & 04/01 & 05/01 & 06/01 & 07/01 & 08/01 & 09/01 \\
 \hline
 Google & 1,188.48 & 1,103.63 & 1,080.91 & 1,216.68 & 1,188.10 & 1,219.00 \\
 \hline
 Amazon & 1,926.52 & 1,775.07 & 1,893.63 & 1,866.78 & 1,776.29 & 1,735.91 \\
 \hline
\end{tabular}
\end{center}
\end{table}

\begin{table}[!htb]
\begin{center}
\caption{U.S. Tech Pair Market Data 2019/2020}
\begin{tabular}{ |c|c|c|c|c|c|c| }
 \hline
Date & 10/01 & 11/01 & 12/01 & 01/01 & 02/01 & 03/01 \\
 \hline
 Google & 1,260.11 & 1,304.96 & 1,337.02 & 1,434.23 & 1,339.33 & 1,298.41  \\
 \hline
 Amazon & 1,776.66 & 1,800.80 & 1,847.84 & 2,008.72 & 1,883.75 & 1,901.09 \\
 \hline
\end{tabular}
\end{center}
\end{table}

\noindent A plot of this market data is shown in Figure 3 below.

\begin{figure}[H]
\caption{U.S. Tech Pair Market Data}
\begin{center}
\begin{tikzpicture}
	\begin{axis}[legend pos=north west,
		xlabel=Month,
		ylabel=Closing Prices,
		ymin = 1000, ymax = 1500,
		axis y line* = left ]
	\addplot[color=green,mark=x] coordinates {
		(1,1188)
		(2,1104)
		(3,1081)
		(4,1217)
		(5,1188)
		(6,1219)
		(7,1260)
		(8,1305)
		(9,1337)
		(10,1434)
		(11,1339)
		(12,1298)
	}; \label{plot1_y1}
	\end{axis}

	\begin{axis}[legend pos=north west,
		xlabel=Month,
		ylabel=Closing Prices,
		ymin = 1700, ymax = 2200,
		axis y line* = right ]
	\addplot[color=blue,mark=*] coordinates {
		(1,1927)
		(2,1775)
		(3,1894)
		(4,1867)
		(5,1776)
		(6,1736)
		(7,1777)
		(8,1801)
		(9,1848)
		(10,2009)
		(11,1884)
		(12,1901)
	}; \label{plot1_y2}
	\addlegendimage{/pgfplots/refstyle=plot1_y1}\addlegendentry{Amazon}
    \addlegendimage{/pgfplots/refstyle=plot1_y2}\addlegendentry{Google}

	\end{axis}
	
\end{tikzpicture}
\end{center}
\end{figure}
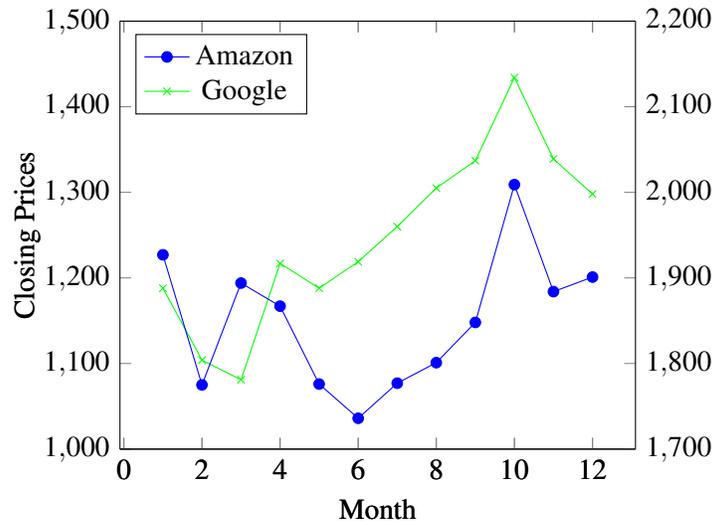

Solving NLP N\_SNA for various values of $\delta$ gives the results in Table 6. A plot of these values is shown in Figure 4 below. The entire 12 point data set is used as the support for the time 1 distribution. The arithmetic average is used for the time 0 prices. The data tuples of closing prices are assigned to the (uniform) discrete distribution  for time 1.

\begin{table}[!htb]
\begin{center}
\caption{$v_s(\delta)$: SA Best Case}
\begin{tabular}{ |r|r|r|r|r|r|r|r|r| }
 \hline
$\delta$ & 0.001 & 1 & 2 & 5 & 10 & 20 & 31 & 31.7 \\
 \hline
$v_s$  & 0.58 & 0.67 & 0.69 & 0.77 & 0.83 & 0.93 & 0.99 & 1.0 \\
 \hline
$w_{google}$ & 10.1 & 100.0 & 100.0 & 100.0 & 100.0 & 100.0 & 100.0 & 100.0 \\
 \hline
$w_{amazon}$ & -6.9 & -67.5 & -67.5 & -67.5 & -67.5 & -67.5 & -67.5 & -67.5 \\
 \hline
\end{tabular}
\end{center}
\end{table}

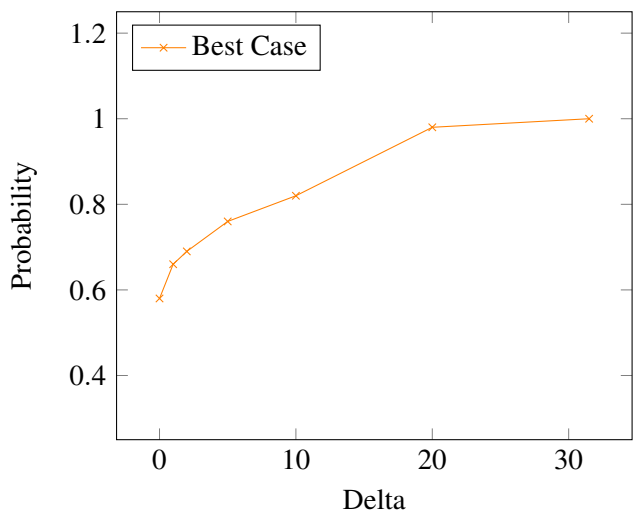
\begin{figure}[!htb]
\caption{Arbitrage Probabilities for U.S. Tech Pair}
\begin{center}
\begin{tikzpicture}
	\begin{axis}[legend pos=north west,
		xlabel=Delta,
		ylabel=Probability,
		ymin = 0.25, ymax = 1.25 ]
	\addplot[color=orange,mark=x] coordinates {
		(0.001,0.58)
		(1.0,0.66)
		(2.0,0.69)
		(5.0,0.76)
		(10,0.82)
		(20,0.98)
		(31.5,1.0)
	}; \label{plot4_y1}

	\addlegendimage{/pgfplots/refstyle=plot4_y1}\addlegendentry{Best Case}
	\end{axis}
	
\end{tikzpicture}
\end{center}
\end{figure}

\noindent A plot of the best case (bc) distribution is shown in Figure 5 below. Recall the robust (strong) no-arbitrage conditions are
\[
\sup_{ w \in \Gamma_{w(s)} } \: { \sup_{ Q \in \mathcal{U}_{\delta}(Q_N) } \mathbb{E}^Q [ \, \mathbbm{1}_{\{ w \cdot S_1 \geq 0 \}} \, ] } < 1. 
\]
The best case distribution has the property that the inner $\sup$ evaluates to 1 for $\delta \geq \delta^* = 31.7$ (critical radius from Table 6 above). Using the optimal portfolio $w^* = \{100.0, -67.5\}$ from Table 6, corresponding to $\delta = \delta^*$, the outer $\sup$ also evaluates to 1. Using the greedy algorithm discussed in Section 2.2 one recovers an \textit{arbitrage} distribution. From the plot in Figure 5 it is clear that Google dominates Amazon which allows for the profit making opportunity. \par
Also from Table 6 one case see that for $\alpha=0.99$ the critical radius is $\delta^{bc}_\alpha = 31$. It turns out that point 3 is the most expensive to move towards the arbitrage conditions, as Amazon dominates Google here (instead of the other way around). Moving 95\% of its mass towards the new values (and using the arbitrage admissible distribution for the remaining points) recovers the \textit{statistical arbitrage} distribution for $\alpha=0.99$. See Table 7 for the detailed probability mass function (PMF) for $\alpha=0.99$. Recall that one point mass from the reference distribution can be split into two pieces according to the budget constraint $\delta$. In this case, this happens for point 3. 95\% of its mass is moved towards the new values in point 13.

\begin{figure}[!htb]
\caption{U.S. Tech Pair Best Case Distribution}
\begin{center}
\begin{tikzpicture}
	\begin{axis}[legend pos=north west,
		xlabel=Month,
		ylabel=Closing Prices,
		ymin = 1100, ymax = 1500,
		axis y line* = left ]
	\addplot[color=green,mark=x] coordinates {
		(1,1265)
		(2,1169)
		(3,1217)
		(4,1247)
		(5,1196)
		(6,1219)
		(7,1260)
		(8,1305)
		(9,1337)
		(10,1434)
		(11,1339)
		(12,1298)
	}; \label{plot2_y1}
	\end{axis}

	\begin{axis}[legend pos=north west,
		xlabel=Month,
		ylabel=Closing Prices,
		ymin = 1700, ymax = 2100,
		axis y line* = right ]
	\addplot[color=blue,mark=*] coordinates {
		(1,1875)
		(2,1731)
		(3,1802)
		(4,1847)
		(5,1771)
		(6,1736)
		(7,1777)
		(8,1801)
		(9,1848)
		(10,2009)
		(11,1884)
		(12,1901)
	}; \label{plot2_y2}
	\addlegendimage{/pgfplots/refstyle=plot2_y1}\addlegendentry{Amazon}
    \addlegendimage{/pgfplots/refstyle=plot2_y2}\addlegendentry{Google}

	\end{axis}
	
\end{tikzpicture}
\end{center}
\end{figure}
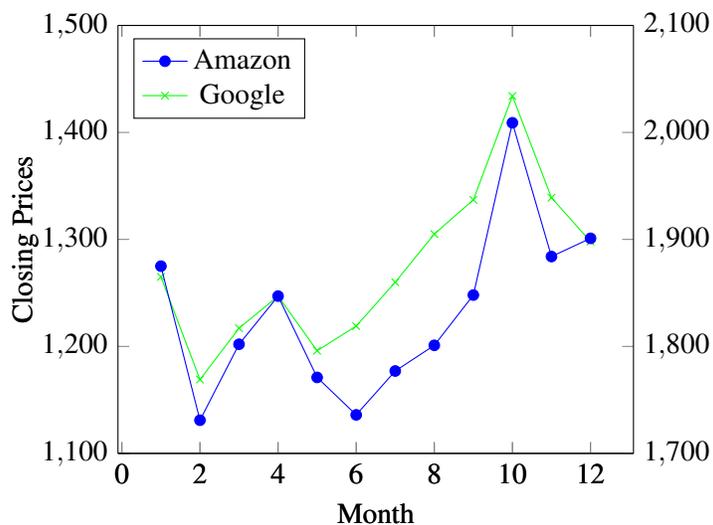

\begin{table}[H]
\begin{center}
\caption{Best Case PMF for $\alpha=0.99$}
\begin{tabular}{ |r|r|r|r|r|r|r|r|r|r|r|r|r|r|r|r| }
 \hline
Point & 1 & 2 & 3 & 4 & 5 & 6 & 7 & 8 & 9 & 10 & 11 & 12 & 13 \\
 \hline
Prob & 1/12 & 1/12 & 0.0041 & 1/12 & 1/12 & 1/12 & 1/12 & 1/12 & 1/12 & 1/12 & 1/12 & 1/12 & 0.0792 \\
 \hline
Google & 1,265 & 1,169 & 1,081 & 1,247 & 1,196 & 1,219 & 1,260 & 1,305 & 1,337 & 1,434 & 1,339 & 1,298 & 1,217 \\
 \hline
Amazon & 1,875 & 1,731 & 1,894 & 1,847 & 1,771 & 1,736 & 1,777 & 1,801 & 1,848 & 2,009 & 1,884 & 1,901 & 1,802 \\
 \hline
\end{tabular}
\end{center}
\end{table}

Switching to the worst case SA conditions, solving NLP N\_SSA for various values of $\delta$ gives the results in Table 8. A plot of these values is shown in Figure 6 below. The problem setup is the same as for the best case SA conditions above. A plot of the worst case (wc) distribution $Q^\alpha_{w^\alpha}$ for $\alpha=0$ is shown in Figure 7 below. The corresponding portfolio is $w^\alpha = \{ 99.11, -66.864 \}$. These results were calculated using the Matlab \textit{fmincon} solver, aplpying a grid search over $\lambda$ to solve the \textit{maximax} problem (which is convex in $\lambda$). Using the greedy algorithm discussed in Section 2.2 one recovers a \textit{no-win} distribution. From the plot in Figure 7 it is clear that \textit{neither} Amazon \textit{nor} Google dominates at all points but for the optimal portfolio $w^\alpha$, a quick check verifies that this distribution leads to a no-win situation, meaning $\mathbb{E}^{Q^\alpha_{w^\alpha}}  [ \, \mathbbm{1}_{\{ w^\alpha \cdot S_1 \geq 0 \}} \, ] = 0$. Our calculations show that for $\alpha=0$ the critical radius is $\delta^{wc}_\alpha = 31.4$. See Table 9 for the detailed probability mass function (PMF). Finally, Figure 8 shows the absolute values of these two positions in the optimal portfolio with weights $w^\alpha$. Here the dominance of the short Amazon position is easier to see. \par

\begin{table}[!htb]
\begin{center}
\caption{$v^{wc}_s(\delta)$: SA Worst Case}
\begin{tabular}{ |r|r|r|r|r|r|r|r|r| }
 \hline
$\delta$ & 0.001 & 1 & 2 & 5 & 10 & 20 & 31 & 31.4 \\
 \hline
$v^{wc}_s$  & 0.58 & 0.51 & 0.47 & 0.41 & 0.31 & 0.11 & 0.004 & 0.0 \\
 \hline
$w_{google}$ & 2.94 & 99.97 & 98.55 & 99.08 & 99.08 & 99.36 & 99.36 & 2.94 \\
 \hline
$w_{amazon}$ & -1.98 & -67.44 & -67.48 & -66.84 & -66.84 & -67.02 & -67.02 & -1.98 \\
 \hline
\end{tabular}
\end{center}
\end{table}

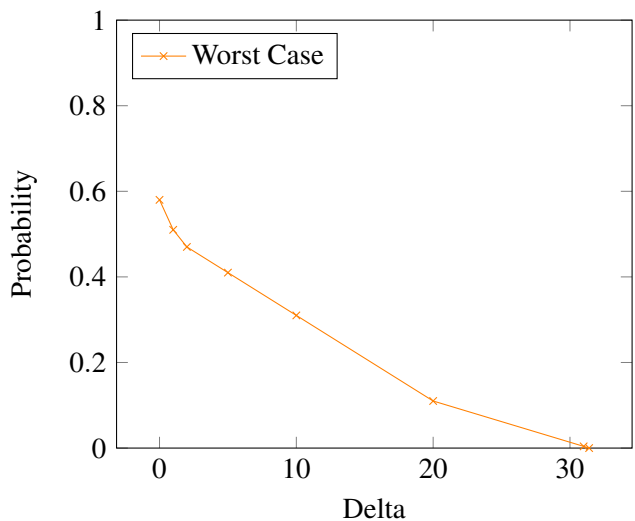
\begin{figure}[!htb]
\caption{Arbitrage Probabilities for U.S. Tech Pair}
\begin{center}
\begin{tikzpicture}
	\begin{axis}[legend pos=north west,
		xlabel=Delta,
		ylabel=Probability,
		ymin = 0.0, ymax = 1.0 ]
	\addplot[color=orange,mark=x] coordinates {
		(0.001,0.58)
		(1.0,0.51)
		(2.0,0.47)
		(5.0,0.41)
		(10,0.31)
		(20,0.11)
		(31,0.004)
		(31.4,0.0)
	}; \label{plot4wc_y1}
	\addlegendimage{/pgfplots/refstyle=plot4wc_y1}\addlegendentry{Worst Case}
	\end{axis}
	
\end{tikzpicture}
\end{center}
\end{figure}

\begin{figure}[!htb]
\caption{U.S. Tech Pair Worst Case Distribution}
\begin{center}
\begin{tikzpicture}
	\begin{axis}[legend pos=north west,
		xlabel=Month,
		ylabel=Closing Prices,
		ymin = 1000, ymax = 1400,
		axis y line* = left ]
	\addplot[color=green,mark=x] coordinates {
		(1,1188)
		(2,1104)
		(3,1081)
		(4,1217)
		(5,1188)
		(6,1186)
		(7,1218)
		(8,1243)
		(9,1275)
		(10,1380)
		(11,1292)
		(12,1288)
	}; \label{plot2wc_y1}
	\end{axis}

	\begin{axis}[legend pos=north west,
		xlabel=Month,
		ylabel=Closing Prices,
		ymin = 1700, ymax = 2100,
		axis y line* = right ]
	\addplot[color=blue,mark=*] coordinates {
		(1,1927)
		(2,1775)
		(3,1894)
		(4,1867)
		(5,1776)
		(6,1758)
		(7,1805)
		(8,1843)
		(9,1890)
		(10,2045)
		(11,1916)
		(12,1908)
	}; \label{plot2wc_y2}
	\addlegendimage{/pgfplots/refstyle=plot2wc_y1}\addlegendentry{Amazon}
    \addlegendimage{/pgfplots/refstyle=plot2wc_y2}\addlegendentry{Google}

	\end{axis}
	
\end{tikzpicture}
\end{center}
\end{figure}
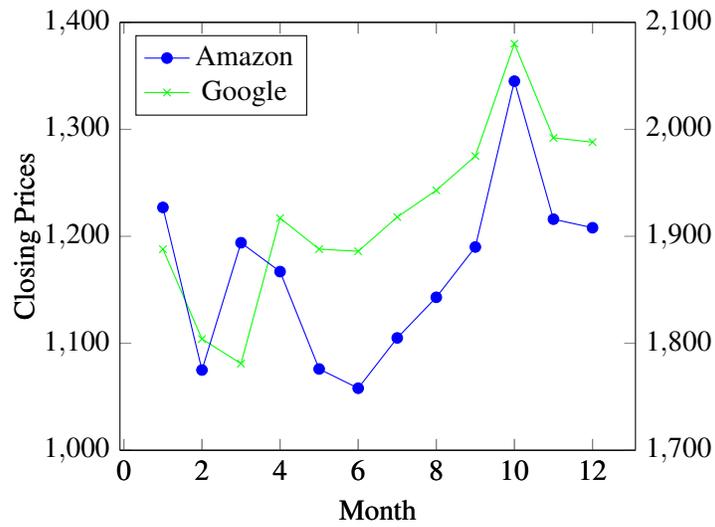

\begin{table}[H]
\begin{center}
\caption{Worst Case PMF for $\alpha=0$}
\begin{tabular}{ |r|r|r|r|r|r|r|r|r|r|r|r|r|r|r|r| }
 \hline
Point & 1 & 2 & 3 & 4 & 5 & 6 & 7 & 8 & 9 & 10 & 11 & 12 \\
 \hline
Prob & 1/12 & 1/12 & 1/12 & 1/12 & 1/12 & 1/12 & 1/12 & 1/12 & 1/12 & 1/12 & 1/12 & 1/12 \\
 \hline
Google & 1,188 & 1,104 & 1,081 & 1,217 & 1,188 & 1,186 & 1,218 & 1,243 & 1,275 & 1,380 & 1,292 & 1,288 \\
 \hline
Amazon & 1,927 & 1,775 & 1,894 & 1,867 & 1,776 & 1,758 & 1,805 & 1,843 & 1,890 & 2,045 & 1,916 & 1,908 \\
 \hline
\end{tabular}
\end{center}
\end{table}

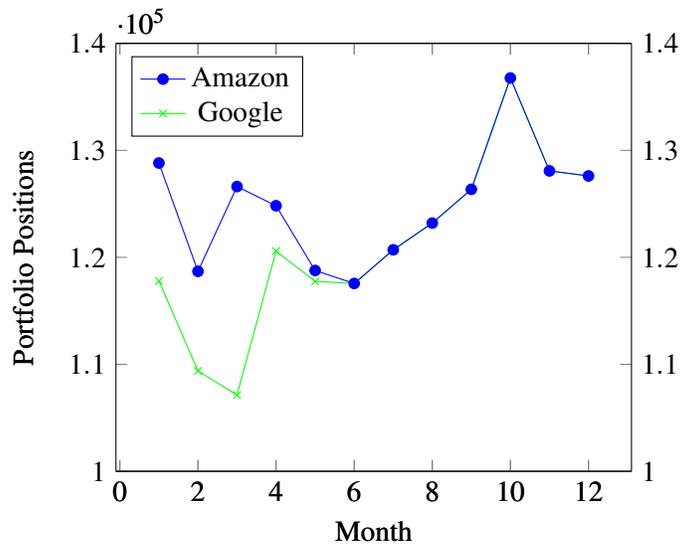
\begin{figure}[!htb]
\caption{U.S. Tech Pair Worst Case Positions: Absolute Values}
\begin{center}
\begin{tikzpicture}
	\begin{axis}[legend pos=north west,
		xlabel=Month,
		ylabel=Portfolio Positions,
		ymin = 100000, ymax = 140000,
		axis y line* = left ]
	\addplot[color=green,mark=x] coordinates {
		(1,117790)
		(2,109381)
		(3,107129)
		(4,120585)
		(5,117753)
		(6,117551)
		(7,120698)
		(8,123198)
		(9,126353)
		(10,136759)
		(11,128075)
		(12,127603)
	}; \label{plot2pos_y1}
	\end{axis}

	\begin{axis}[legend pos=north west,
		xlabel=Month,
		ylabel=Portfolio Positions,
		ymin = 100000, ymax = 140000,
		axis y line* = right ]
	\addplot[color=blue,mark=*] coordinates {
		(1,128815)
		(2,118688)
		(3,126616)
		(4,124820)
		(5,118770)
		(6,117555)
		(7,120702)
		(8,123202)
		(9,126357)
		(10,136763)
		(11,128079)
		(12,127607)
	}; \label{plot2pos_y2}
	\addlegendimage{/pgfplots/refstyle=plot2pos_y1}\addlegendentry{Amazon}
    \addlegendimage{/pgfplots/refstyle=plot2pos_y2}\addlegendentry{Google}

	\end{axis}
	
\end{tikzpicture}
\end{center}
\end{figure}

\subsection{Basket Trading}
Basket trading involves simultaneous trading of a basket of stocks. This example computes the critical radius for a small basket of U.S. equities from the S\&P 500 index used in the statistical arbitrage study by \citep{zhao2018optimal}. Table 10 below lists the stock tickers, names, and industries. Table 11 displays a \textit{partial} listing of the 5y historical market data set from March 2015 through March 2020 used in this study. As before, the arithmetic average is used for time 0 and the data tuples for time 1. Table 12 and Figure 9 display the optimal portfolios and best case arbitrage probabilities. Figures 10 and 11 show different views of the best case distribution for $\alpha =1 $ and optimal portfolio $w^\alpha = \{ 17.20, -0.24, 10.90, -11.19, -41.23, 2.07, -3.90 \}$. The quantiles in Figure 11b are $\{ 0.25, 0.5, 0.75 \}$ respectively.

\begin{table}[!htb]
\begin{center}
\caption{Basket Constituents}
\begin{tabular}{ |c|c|c|c| }
 \hline
  Ticker & Name & Industry & Market Cap (bn) \\
 \hline
 APA & Apache Corporation & Energy: Oil and Gas & 10.68 \\
 \hline
 AXP & American Express Company & Credit Services & 109.0 \\
 \hline
 CAT & Caterpillar Inc. & Farm Machinery & 74.94 \\
 \hline
 COF & Capital One Financial Corp. & Credit Services & 46.19 \\
 \hline
 FCX & Freeport-McMoRan Inc. & Copper & 17.33 \\
 \hline
 IBM & 1nternational Business Machines Corp. & Technology & 132.70 \\
 \hline
 MMM & 3M Company & Industrial Machinery & 90.33 \\
 \hline
\end{tabular}
\end{center}
\end{table}

\begin{table}[H]
\begin{center}
\caption{Basket 2019 Market Data}
\begin{tabular}{ |r|r|r|r|r|r|r|r| }
 \hline
Date & 06/01 & 07/01 & 08/01 & 09/01 & 10/01 & 11/01 & 12/01 \\
 \hline
 APA & 28.13 & 23.71 & 21.17 & 25.12 & 21.25 & 22.11 & 25.39 \\
 \hline
 AXP & 122.16 & 123.08 & 119.50 & 117.42 & 116.43 & 119.71 & 124.06 \\
 \hline
 CAT & 133.26 & 128.74 & 117.25 & 124.45 & 135.77 & 143.72 & 146.65 \\
 \hline
 COF & 89.62 & 91.28 & 85.56 & 90.26 & 92.51 & 99.22 & 102.51 \\
 \hline
 FCX & 11.45 & 10.91 & 9.10 & 9.48 & 9.73 & 11.34 & 13.07 \\
 \hline
 IBM & 133.31 & 143.31 & 131.02 & 142.24 & 130.80 & 131.51 & 132.65 \\
 \hline
 MMM & 168.77 & 170.11 & 157.45 & 161.53 & 162.11 & 166.80 & 174.84 \\
 \hline
\end{tabular}
\end{center}
\end{table}


\begin{table}[H]
\begin{center}
\caption{$v_s(\delta)$: SA Best Case}
\begin{tabular}{ |r|r|r|r|r|r|r|r| }
 \hline
$\delta$ & 0.001 & 0.01 & 0.05 & 0.1 & 0.5 & 1 \\
 \hline
$v_s$ & 0.67 & 0.71 & 0.75 & 0.81 & 0.95 & 1.0 \\
 \hline
$w_{apa}$ & 5.31 & 16.37 & 7.03 & -7.61 & 6.38 & 17.2 \\
 \hline
$w_{axp}$ & -3.90 & -7.50 & -6.44 & -14.45 & -4.75 & -0.24 \\
 \hline
$w_{cat}$ & -3.03 & 2.69 & -0.94 & 5.54 & 6.36 & 10.90 \\
 \hline
$w_{cof}$ & 5.29 & -11.00 & -5.27 & 14.53 & 1.66 & -11.19 \\
 \hline
$w_{fcx}$ & 9.59 & -7.39 & 17.31 & -92.13 & -52.97 & -41.23 \\
 \hline
$w_{ibm}$ & -7.12 & -6.56 & -5.83 & 3.00 & 0.68 & 2.07 \\
 \hline
$w_{mmm}$ & 4.82 & 9.08 & 7.86 & 3.38 & -0.41 & -3.90 \\
 \hline
\end{tabular}
\end{center}
\end{table}

\begin{figure}[H]
\caption{Arbitrage Probabilities for U.S. Equity Basket}
\begin{center}
\begin{tikzpicture}[scale=0.975]
	\begin{axis}[legend pos=north west,
		xlabel=Delta,
		ylabel=Probability,
		ymin = 0.25, ymax = 1.25 ]
	\addplot[color=black,mark=x] coordinates {
		(0.001,0.67)
		(0.01,0.71)
		(0.05,0.73)
		(0.1,0.76)
		(0.5,0.91)
		(1.0,1.0)
	}; \label{plot6_y1}
	\addlegendimage{/pgfplots/refstyle=plot6_y1}\addlegendentry{Best Case}
	\end{axis}
	
\end{tikzpicture}
\end{center}
\end{figure}

\begin{figure}[!htb]
\caption{Correlation Matrix for Equity Basket BC Distribution}
\centerline{\scalebox{0.15}[0.15]{\includegraphics{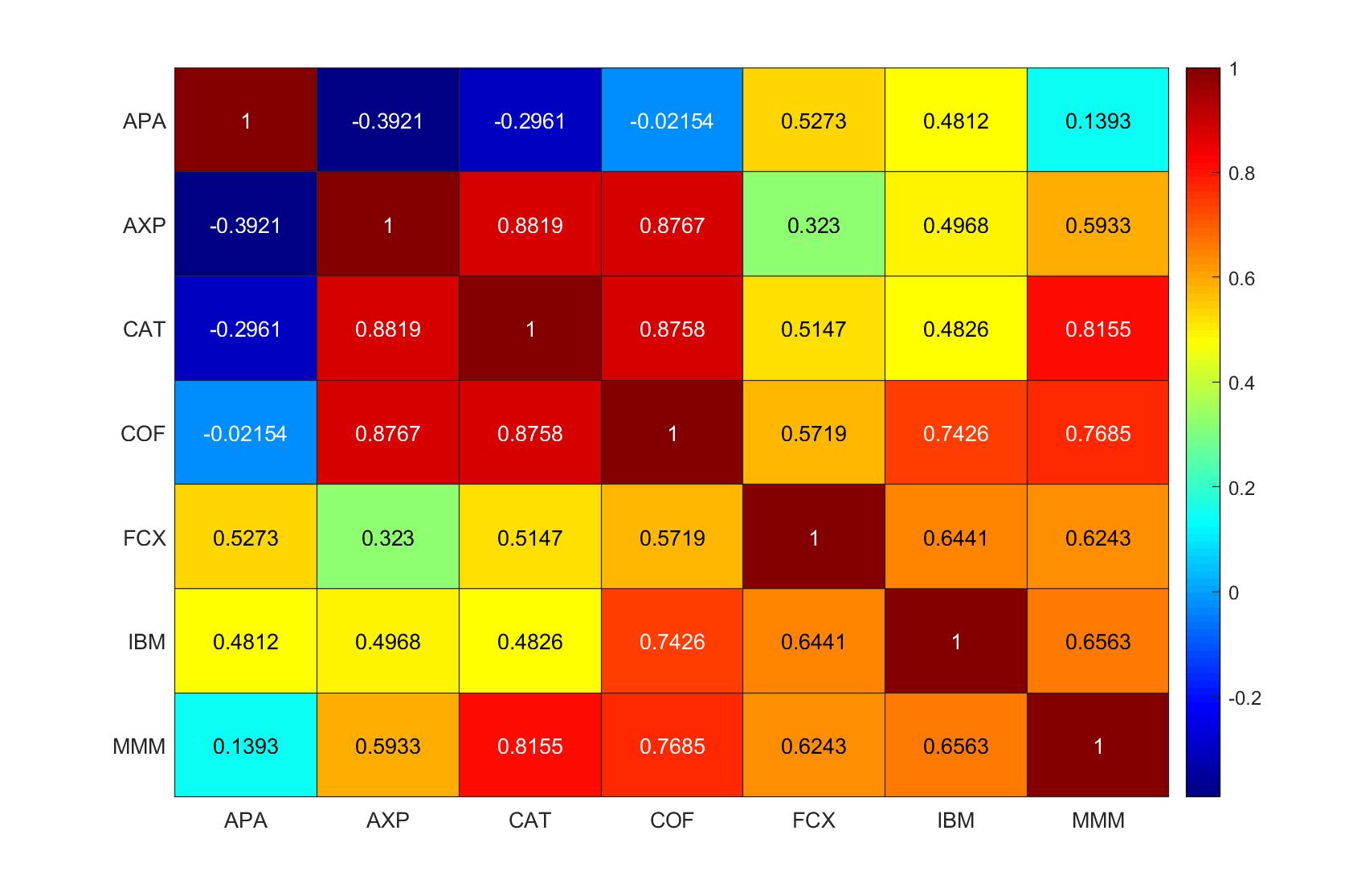}}}
\end{figure}

\begin{figure}[H]
	\centering
	\caption{Equity Basket BC Distribution}%
	\subfloat[Parallel Coords]{\scalebox{0.145}[0.15]{\includegraphics{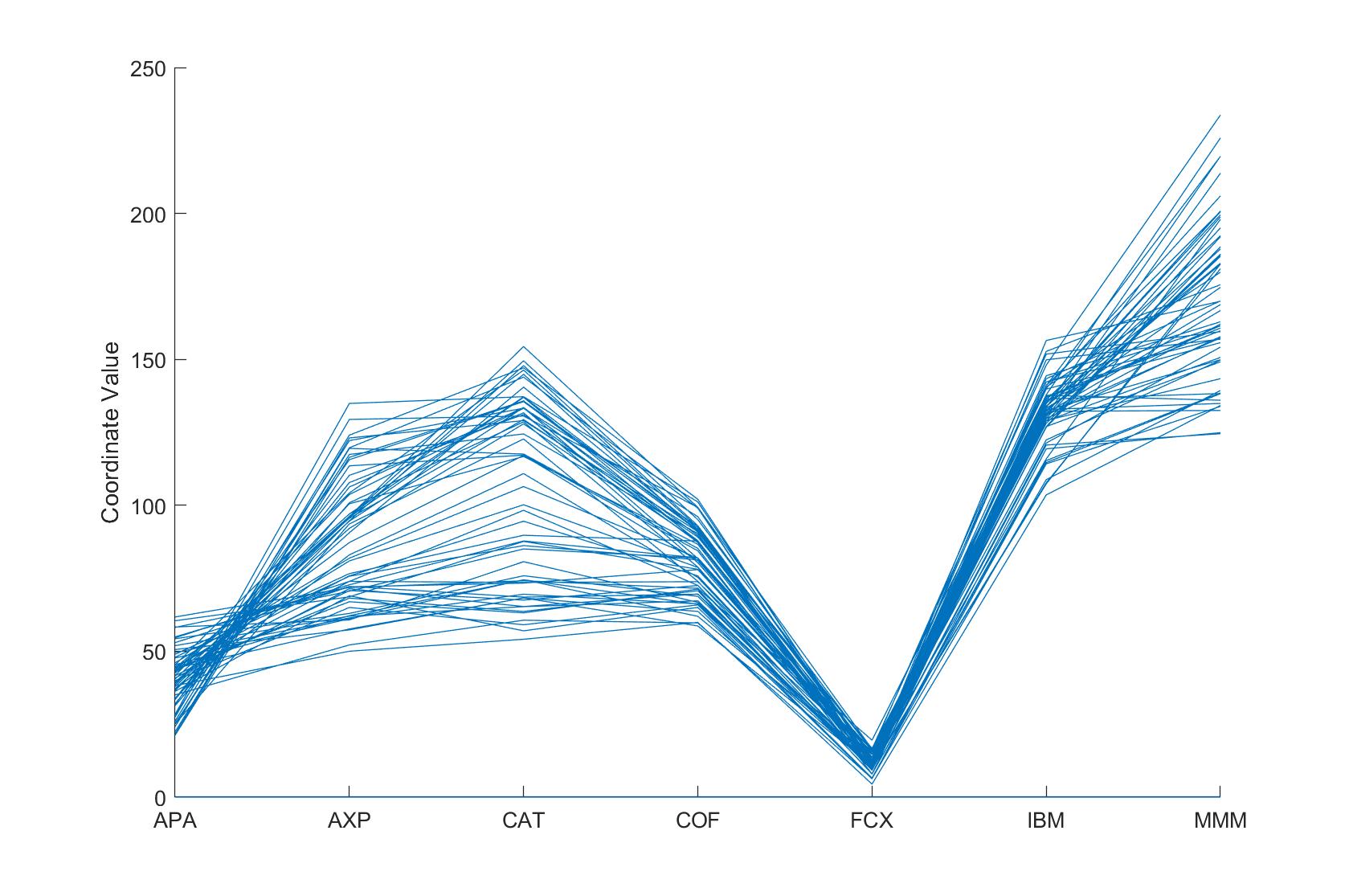}}}%
	\quad
	\subfloat[Quantiles]{\scalebox{0.145}[0.15]{\includegraphics{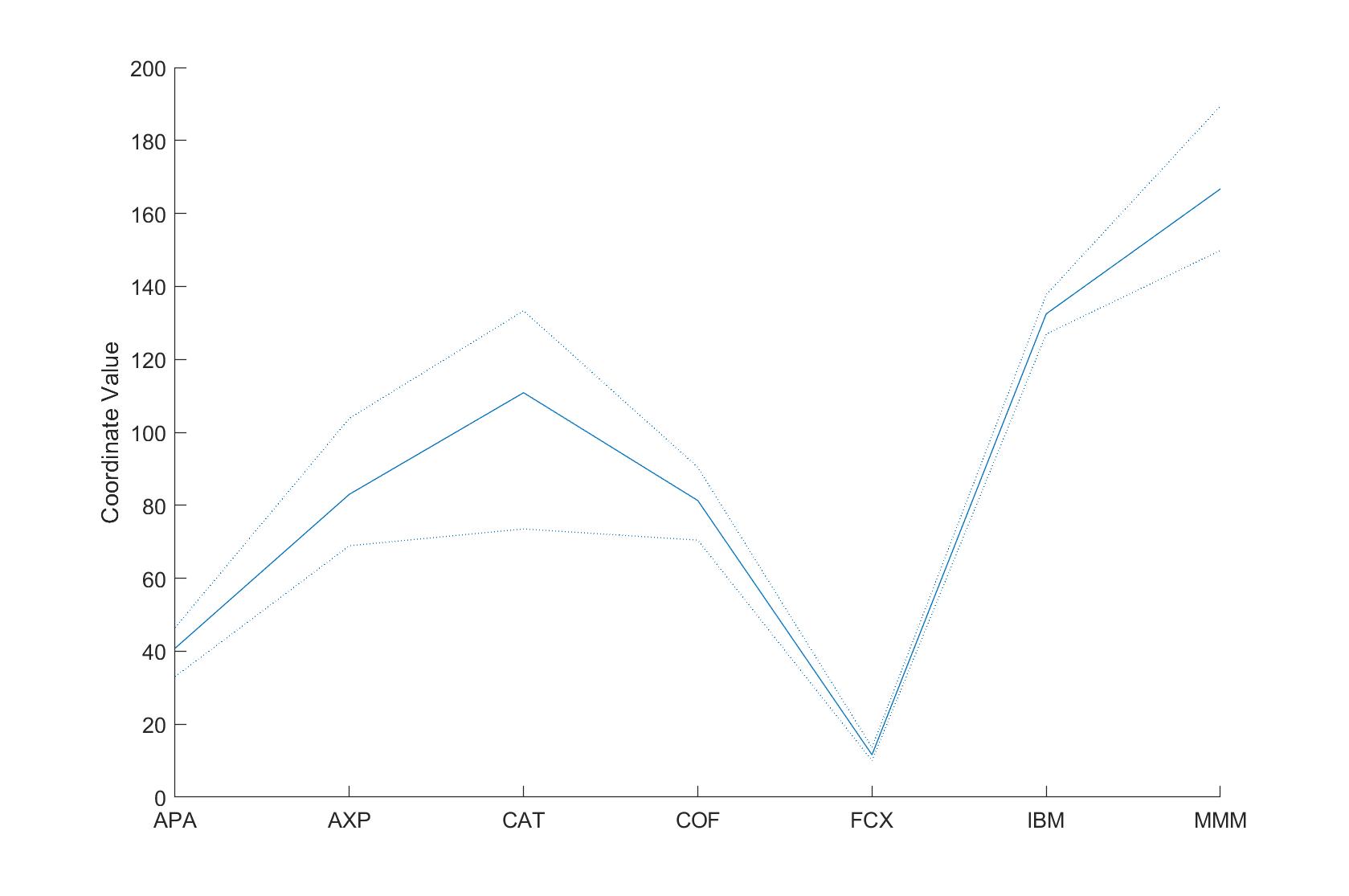}}}%
\end{figure}

%

Switching to the worst case SA conditions, solving NLP N\_SSA for various values of $\delta$ gives the results in Table 13. A plot of these values is shown in Figure 12 below. The problem setup is the same as for the best case SA conditions above. Note that the critical value $\delta^{wc}_{\alpha=0} = 13.5$ is significantly higher than the critical value $\delta^{bc}_{\alpha=1} = 1$. The reference distribution is much closer to admitting arbitrage than admitting a \textit{no-win} situation. Figures 13 and 14 show different views of the worst case distribution $Q^\alpha_{w^\alpha}$ for $\alpha=0$ and optimal portfolio $w^\alpha = \{ 7.45, 9.97, -6.84, 9.90, 3.14, -6.54, -3.11 \}$. The quantiles in Figure 14b are $\{ 0.25, 0.5, 0.75 \}$ respectively.

\begin{table}[H]
\begin{center}
\caption{$v^{wc}_s(\delta)$: SA Worst Case}
\begin{tabular}{ |r|r|r|r|r|r|r|r| }
 \hline
$\delta$ & 0.001 & 0.5 & 1.0 & 2.0 & 5.0 & 10.0 & 13.5 \\
 \hline
$v^{wc}_s$ & 0.68 & 0.53 & 0.46 & 0.39 & 0.27 & 0.07 & 0.0 \\
 \hline
$w_{apa}$ & 25.96 & 784.18 & -45.29 & -657.58 & -654.99 & 685.87 & 7.45 \\
 \hline
$w_{axp}$ & -43.41 & -948.06 & --895.60 & 379.42 & 353.38 & -540.94 & 9.97 \\
 \hline
$w_{cat}$ & -40.50 & 117.04 & -321.50 & 975.81 & 973.14 & -994.01 & -6.84 \\
 \hline
$w_{cof}$ & 87.52 & 667.87 & 166.52 & -115.33 & -103.22 & -0.26 & 9.90 \\
 \hline
$w_{fcx}$ & 52.76 & 484.11 & 360.76 & -50.30 & -36.49 & 55.17 & 3.14 \\
 \hline
$w_{ibm}$ & -46.11 & -523.48 & 143.70 & -834.35 & -810.86 & 820.05 & -6.54 \\
 \hline
$w_{mmm}$ & 31.10 & 280.00 & 450.25 & 61.49 & 50.99 & 90.97 & -3.11 \\
 \hline
\end{tabular}
\end{center}
\end{table}

\begin{figure}[H]
\caption{Arbitrage Probabilities for U.S. Equity Basket}
\begin{center}
\begin{tikzpicture}[scale=0.975]
	\begin{axis}[legend pos=north west,
		xlabel=Delta,
		ylabel=Probability,
		ymin = 0.0, ymax = 1.0 ]
	\addplot[color=black,mark=x] coordinates {
		(0.001,0.68)
		(0.5,0.53)
		(1.0,0.46)
		(2.0,0.39)
		(5.0,0.27)
		(10.0,0.07)
		(13.5,0.0)
	}; \label{plot6_y1}
	\addlegendimage{/pgfplots/refstyle=plot6_y1}\addlegendentry{Worst Case}
	\end{axis}
	
\end{tikzpicture}
\end{center}
\end{figure}

\begin{figure}[H]
\caption{Correlation Matrix for Equity Basket WC Distribution}
\centerline{\scalebox{0.2}[0.2]{\includegraphics{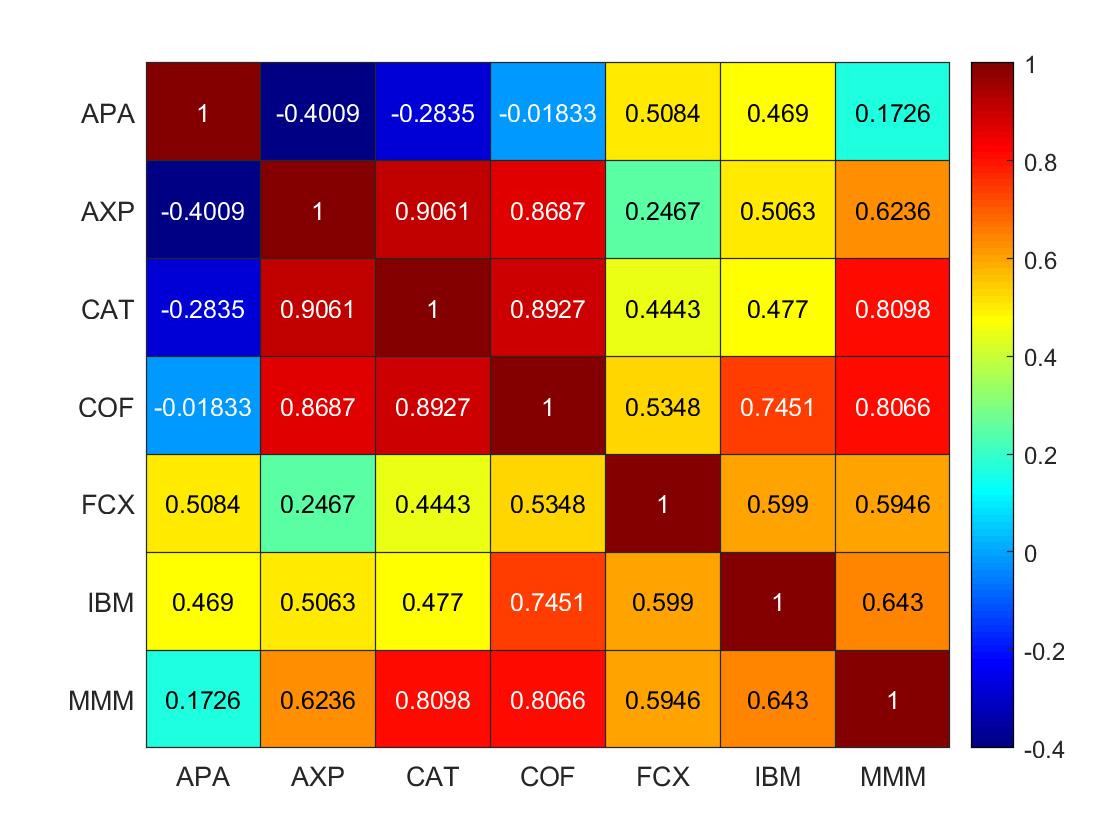}}}
\end{figure}

\begin{figure}[H]
	\centering
	\caption{Equity Basket WC Distribution}%
	\subfloat[Parallel Coords]{\scalebox{0.22}[0.2]{\includegraphics{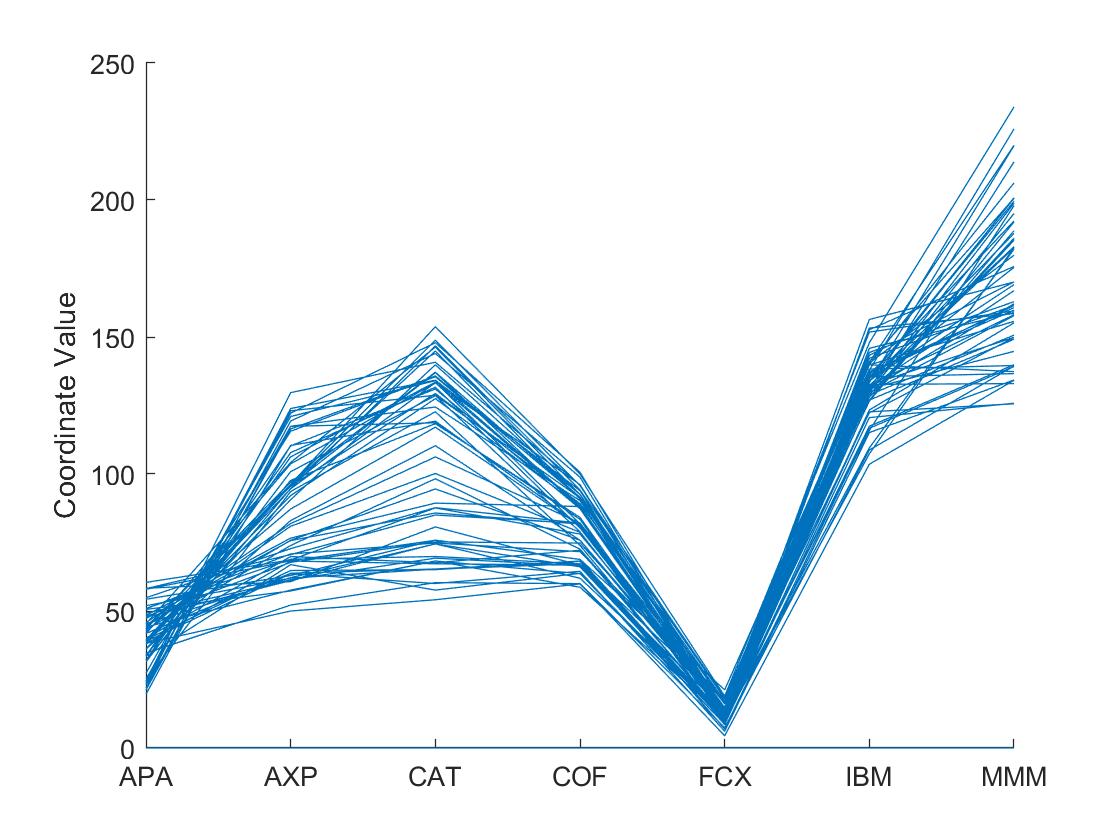}}}%
	\quad
	\subfloat[Quantiles]{\scalebox{0.22}[0.2]{\includegraphics{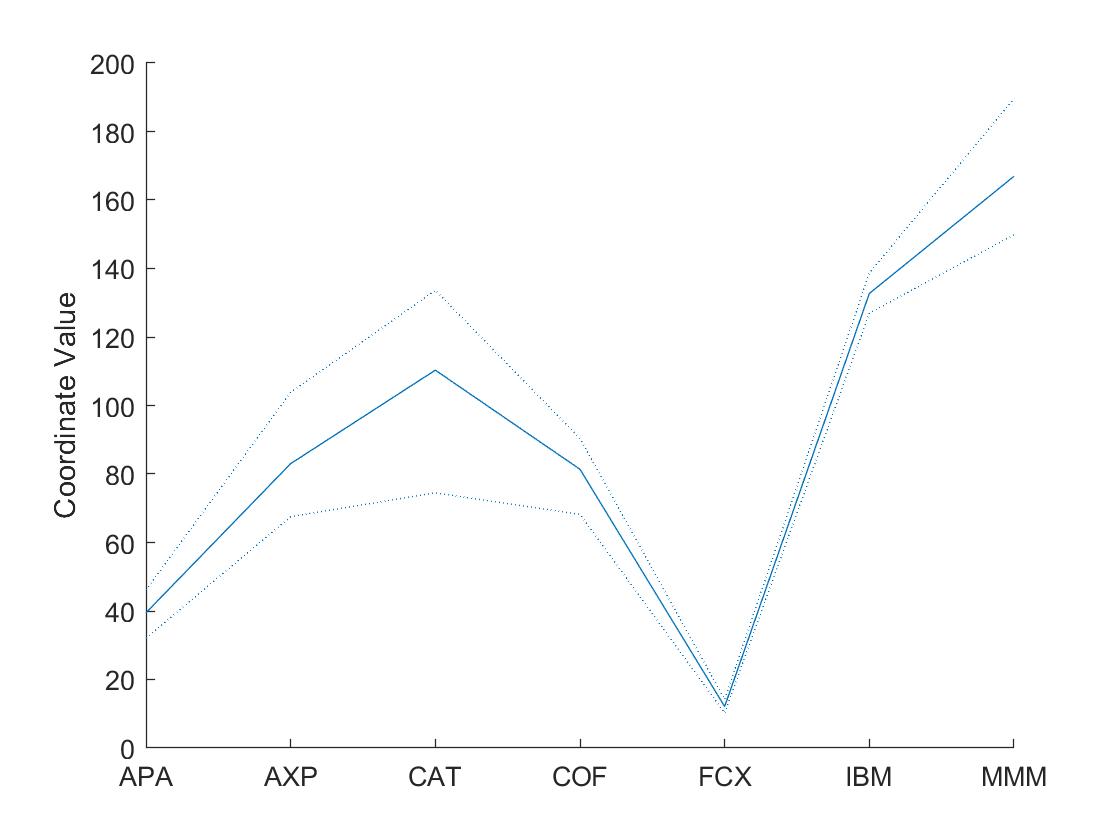}}}%
\end{figure}

%

As another example, let us consider a basket of stock indices, taken from the Market Watch financial website. In particular, we look at broad based equity indices ( Dow Jones 30, S\&P 500 ), the Nasdaq technology stock index ( IXIC ), the USO oil exchange traded fund (ETF), and a gold ETF (SGOL). Table 14 displays a \textit{partial} listing of the historical data set from March 2015 to March 2020 used in this study. As before, the arithmetic average is used for time 0 and the data tuples for time 1.

\begin{table}[H]
\begin{center}
\caption{Basket 2019 Market Data}
\begin{tabular}{ |r|r|r|r|r|r|r|r| }
 \hline
Date & 06/01 & 07/01 & 08/01 & 09/01 & 10/01 & 11/01 & 12/01 \\
 \hline
 DJI & 26,600 & 26,864 & 26,403 & 26,917 & 27,046 & 28,051 & 28,538 \\
 \hline
 GSPC & 2,942 & 2,980 & 2,926 & 2,977 & 3,038 & 3,141 & 3,231 \\
 \hline
 IXIC & 8,006 & 8,175 & 7,963 & 7,999 & 8,292 & 8,665 & 8,973 \\
 \hline
 USO & 12.04 & 12.04 & 11.46 & 11.34 & 11.30 & 11.62 & 12.81 \\
 \hline
 SGOL & 13.60 & 13.61 & 14.69 & 14.20 & 14.56 & 14.16 & 14.62 \\
 \hline
\end{tabular}
\end{center}
\end{table}

Solving NLP N\_SNA for various values of $\delta$ gives the results in Table 15 below. The arbitrage probability curve is plotted in Figure 15.  Different views of the best case distribution for $\alpha=1$ and optimal portfolio $w^\alpha = \{ -0.16, -5.13, 1.61, 179.35, 290.23 \}$ are shown in Figures 16 and 17.

\begin{table}[H]
\begin{center}
\caption{$v_s(\delta)$: SA Best Case}
\begin{tabular}{ |r|r|r|r|r|r|r|r| }
 \hline
$\delta$ & 0.001 & 0.01 & 0.05 & 0.1 & 0.5 & 0.6 \\
 \hline
$v_s$ & 0.68 & 0.68 & 0.70 & 0.81 & 0.99 & 1.0 \\
 \hline
$w_{dji}$ & 0.80 & 1.55 & 1.37 & 0.83 & 0.95 & -0.16 \\
 \hline
$w_{gspc}$ & -0.04 & -0.77 & 2.73 & 2.54 & -11.65 & -5.13 \\
 \hline
$w_{ixic}$ & -2.70 & -4.97 & -5.45 & -3.21 & 0.08 & 1.61 \\
 \hline
$w_{uso}$ & -6.01 & 4.42 & 105.07 & 21.27 & -392.33 & 179.35 \\
 \hline
$w_{sgol}$ & -0.01 & -17.15 & -232.03 & -334.96 & 1,000.00 & 290.23 \\
 \hline
\end{tabular}
\end{center}
\end{table}

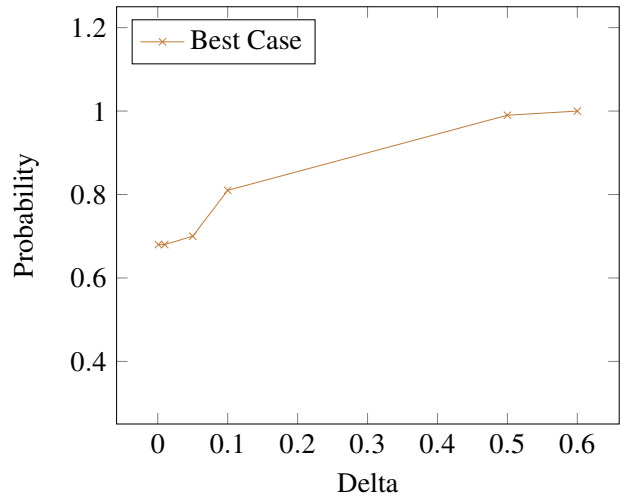
\begin{figure}[H]
\caption{Arbitrage Probabilities for Basket of Indices}
\begin{center}
\begin{tikzpicture}[scale=0.975]
	\begin{axis}[legend pos=north west,
		xlabel=Delta,
		ylabel=Probability,
		ymin = 0.25, ymax = 1.25 ]
	\addplot[color=brown,mark=x] coordinates {
		(0.001,0.68)
		(0.01,0.68)
		(0.05,0.70)
		(0.1,0.81)
		(0.5,0.99)
		(0.6,1.0)
	}; \label{plot7_y1}
	\addlegendimage{/pgfplots/refstyle=plot7_y1}\addlegendentry{Best Case}
	\end{axis}
	
\end{tikzpicture}
\end{center}
\end{figure}

\begin{figure}[H]
\caption{Correlation Matrix for Indices BC Distribution}
\centerline{\scalebox{0.2}[0.2]{\includegraphics{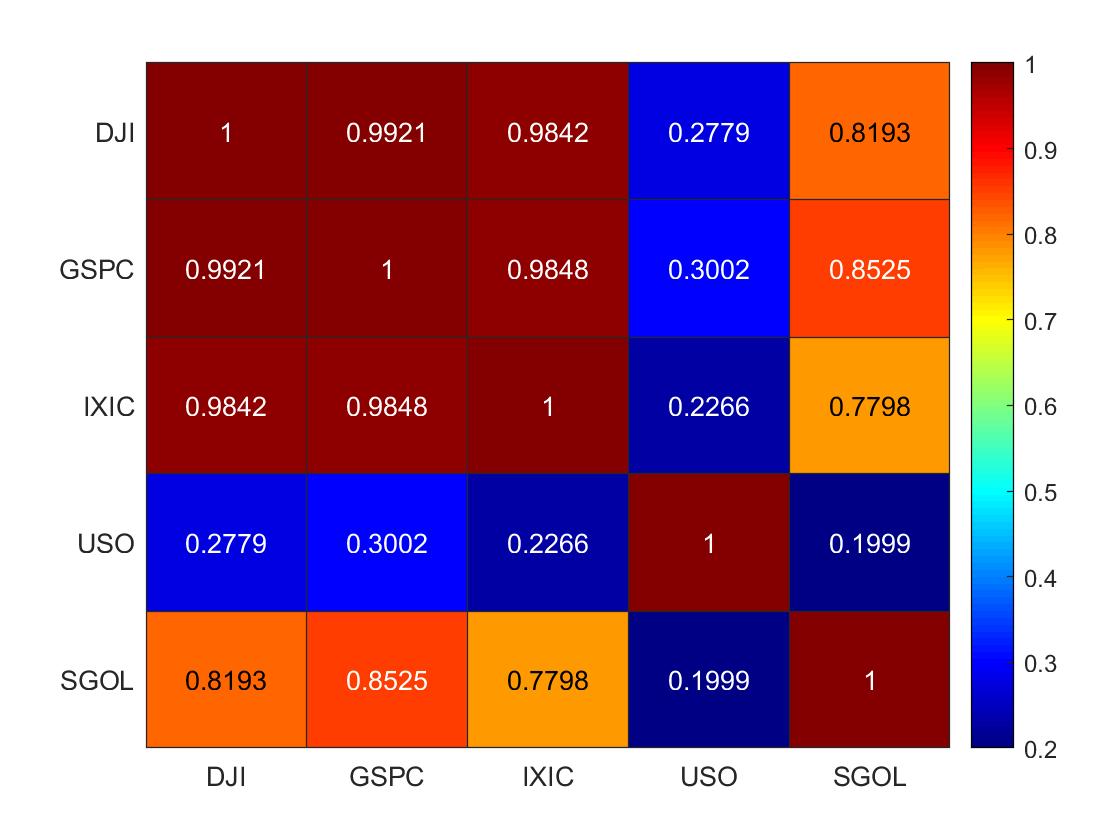}}}
\end{figure}

\begin{figure}[H]
	\centering
	\caption{Indices BC Distribution}%
	\subfloat[Parallel Coords]{\scalebox{0.2}[0.2]{\includegraphics{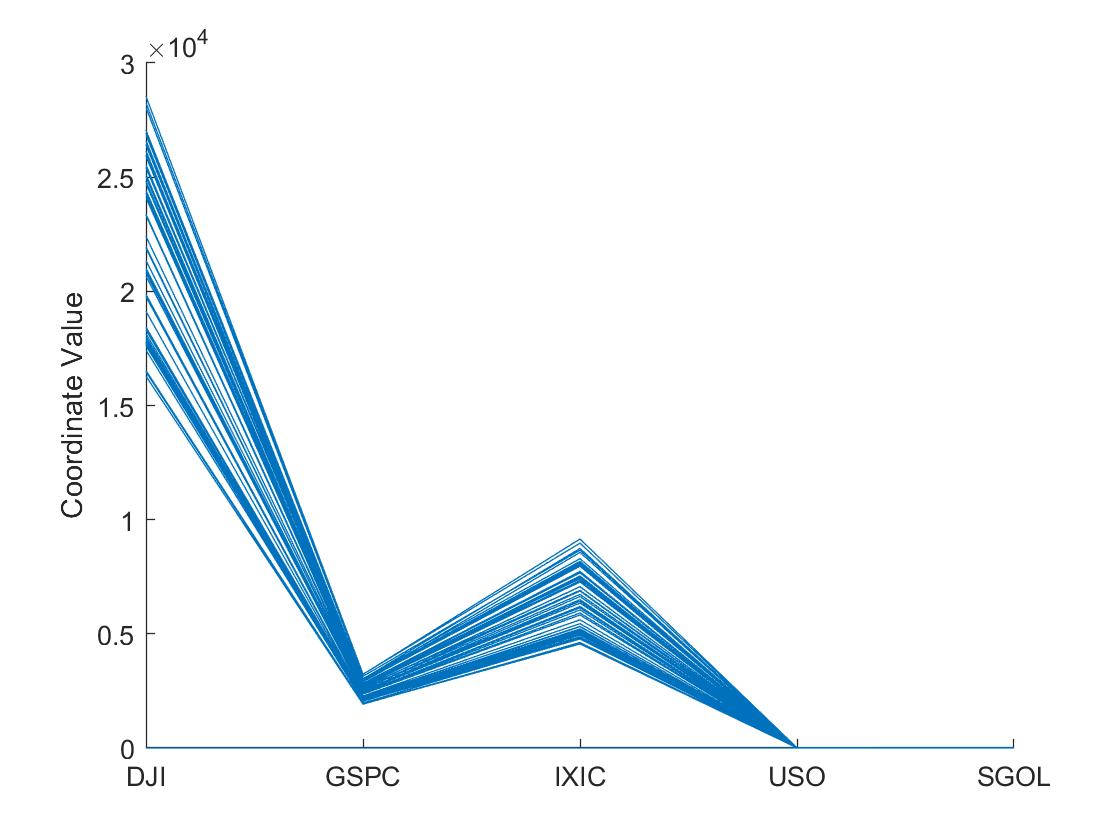}}}%
	\quad
	\subfloat[Quantiles]{\scalebox{0.2}[0.2]{\includegraphics{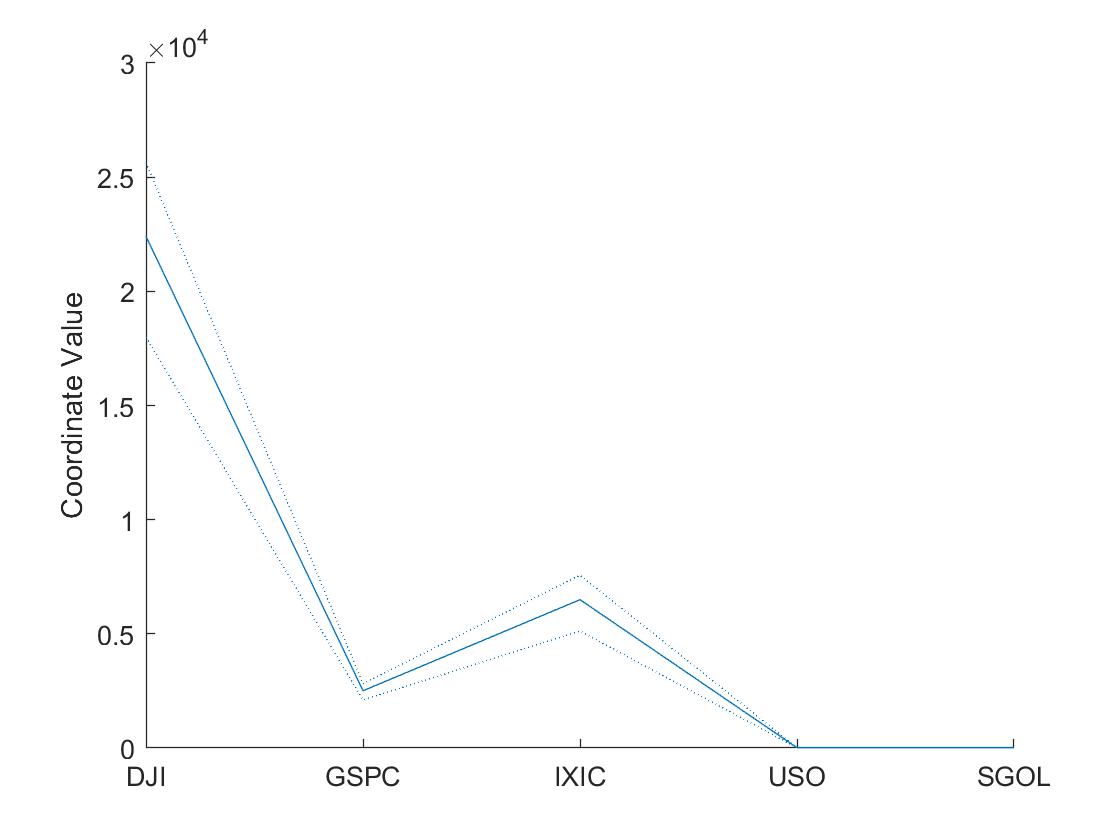}}}%
\end{figure}

%

Switching to the worst case gives the results in Table 16 below. The arbitrage probability curve is plotted in Figure 18. Note that the critical value $\delta^{wc}_{\alpha=0} = 32.6$ is significantly higher than the critical value $\delta^{bc}_{\alpha=1} = 0.6$. As before, the reference distribution is much closer to admitting arbitrage than admitting a \textit{no-win} situation. Different views of the worst case distribution for $\alpha=0$ and optimal portfolio $w^\alpha = \{ 1.84, 8.42, -9.52, 9.0, 3.0 \}$ are shown in Figures 19 and 20.

\begin{table}[H]
\begin{center}
\caption{$v^{wc}_s(\delta)$: SA Worst Case}
\begin{tabular}{ |r|r|r|r|r|r|r|r| }
 \hline
$\delta$ & 0.001 & 1.0 & 2.0 & 5.0 & 10.0 & 20.0 & 32.6 \\
 \hline
$v^{wc}_s$ & 0.68 & 0.64 & 0.61 & 0.55 & 0.45 & 0.25 & 0.0 \\
 \hline
$w_{dji}$ & 110.62 & 184.78 & 142.34 & 49.08 & 189.16 & 260.27 & 1.84 \\
 \hline
$w_{gspc}$ & -71.35 & -83.14 & -53.80 & -18.59 & -71.69 & -98.82 & 8.42 \\
 \hline
$w_{ixic}$ & -349.21 & -597.31 & -464.08 & -160.00 & -616.62 & -848.45 & -9.52 \\
 \hline
$w_{uso}$ & -55.20 & 29.16 & 16.30 & 5.79 & 14.07 & 46.47 & 9.00 \\
 \hline
$w_{sgol}$ & 19.75 & -232.03 & 5.59 & 5.33 & 4.45 & 23.25 & 3.00\\
 \hline
\end{tabular}
\end{center}
\end{table}

\begin{figure}[H]
\caption{Arbitrage Probabilities for Basket of Indices}
\begin{center}
\begin{tikzpicture}[scale=0.975]
	\begin{axis}[legend pos=north west,
		xlabel=Delta,
		ylabel=Probability,
		ymin = 0.0, ymax = 1.0 ]
	\addplot[color=brown,mark=x] coordinates {
		(0.001,0.68)
		(1.0,0.64)
		(2.0,0.61)
		(5.0,0.55)
		(10.0,0.45)
		(20.0,0.25)
		(31.6,0.0)
	}; \label{plot7wc_y1}
	\addlegendimage{/pgfplots/refstyle=plot7wc_y1}\addlegendentry{Worst Case}
	\end{axis}
	
\end{tikzpicture}
\end{center}
\end{figure}
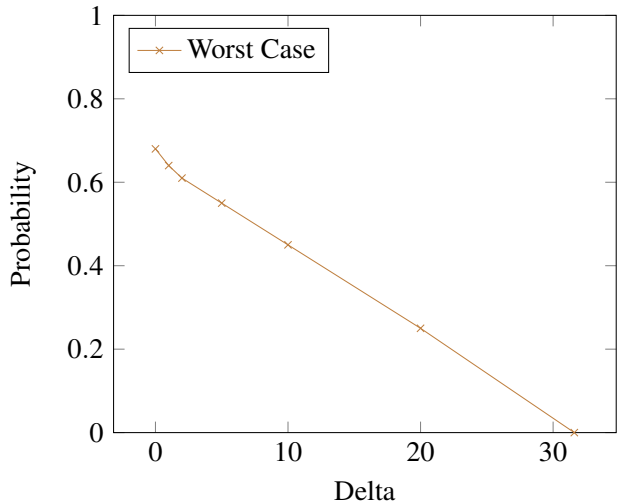

\begin{figure}[H]
\caption{Correlation Matrix for Indices WC Distribution}
\centerline{\scalebox{0.2}[0.2]{\includegraphics{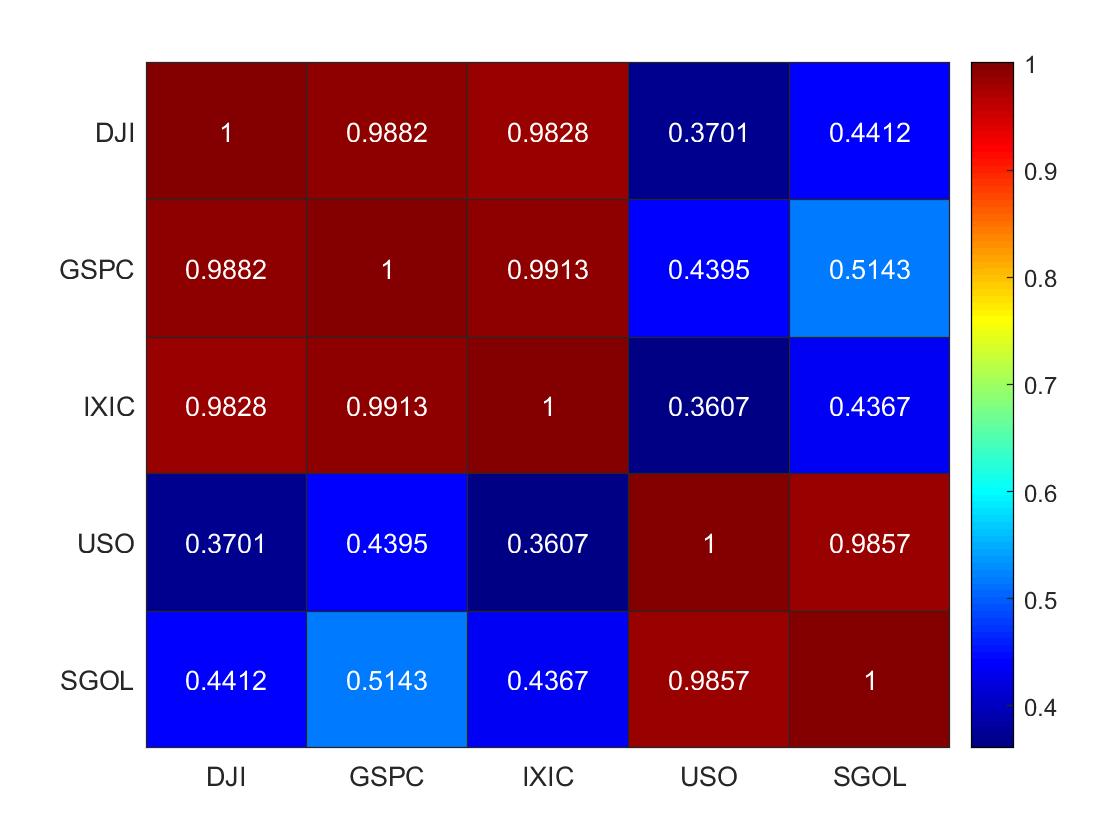}}}
\end{figure}

\begin{figure}[H]
	\centering
	\caption{Indices WC Distribution}%
	\subfloat[Parallel Coords]{\scalebox{0.22}[0.2]{\includegraphics{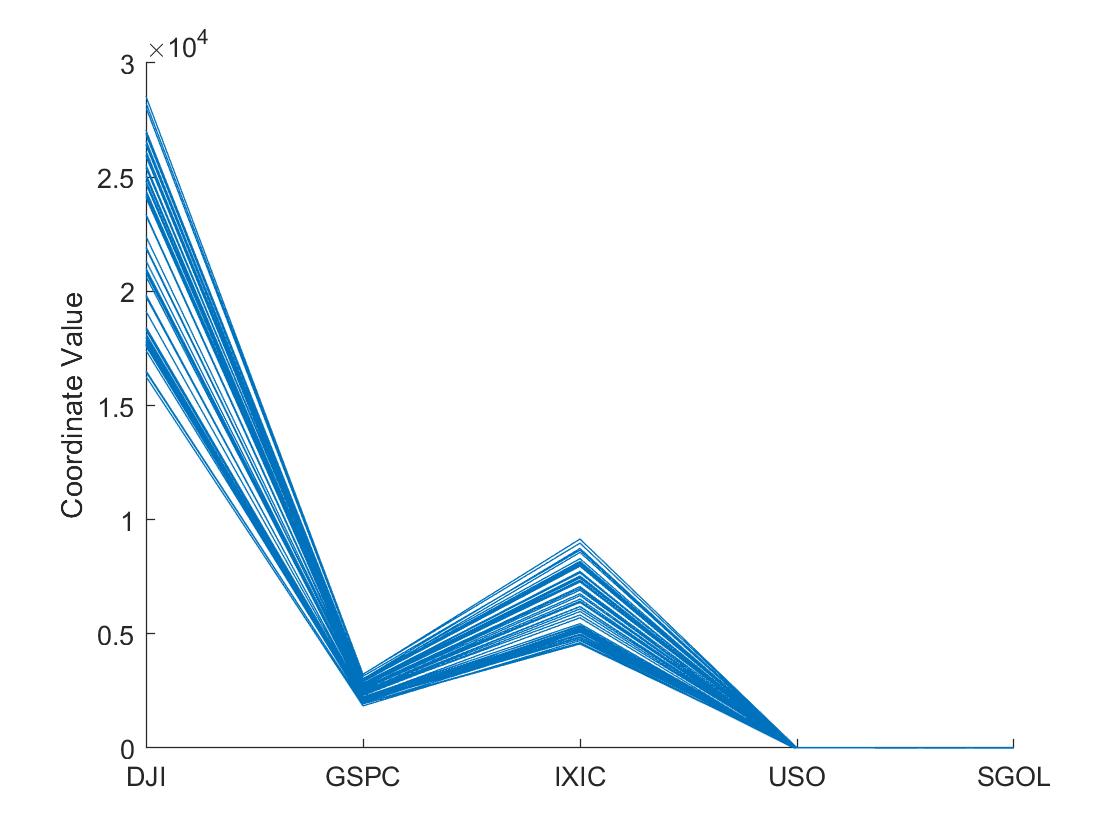}}}%
	\quad
	\subfloat[Quantiles]{\scalebox{0.22}[0.2]{\includegraphics{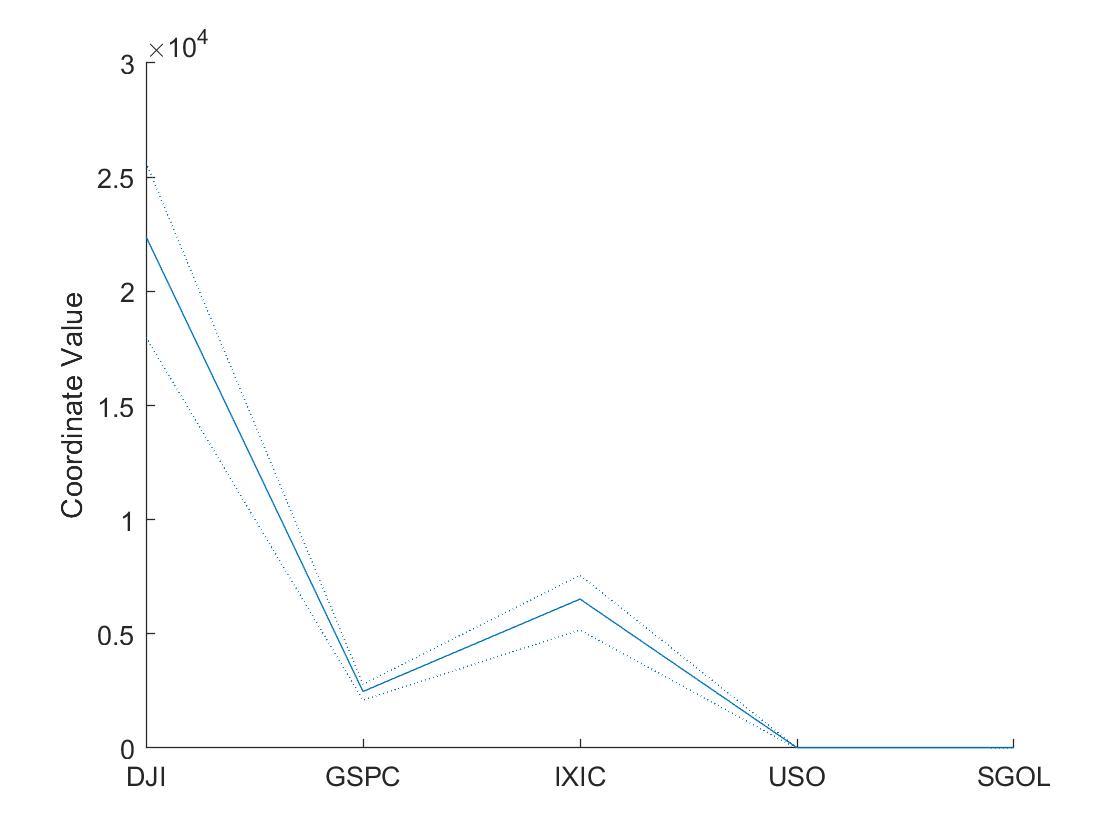}}}%
\end{figure}

%

\subsection{Nearest NA Problem}
This subsection looks at a couple of concrete examples for the nearest NA problem discussed in Section 2.5. In particular, short sales are allowed so we consider the problem setting of Section 2.5.1. The first example is a simple one-period binomial tree asset pricing model. The second is a one-period pairs trading example using the Russell 2000 small-cap index and the S\&P 500 index. The third example looks at basket trading using the index basket from Section 6.3.

\subsubsection{Binomial Tree Asset Pricing}
For this example, we again consider the simple setting of a one-period binomial tree asset pricing model. There is a riskless bond priced at par at time zero that earns a deterministic risk free rate of return $r$ at time 1. In addition there is a risky asset (stock) with initial price $s_0$ and time 1 price $s_u = u s_0$ that occurs with probability $p = 1/2$ and price $s_d = d s_0$ that occurs with probability $q = 1-p = 1/2$. The (weak) no-arbitrage conditions can be stated as: $0 < d < 1+r < u$ \citep{shreve2005stochastic}. Let us mock up an example to violate this. Consider the problem setting below. Here $0 < 1+r = 1.01 < d = u = 1.01333...$ thus the conditions are violated. Intuitively the investor could always make money by going long the stock and borrowing via the bond.

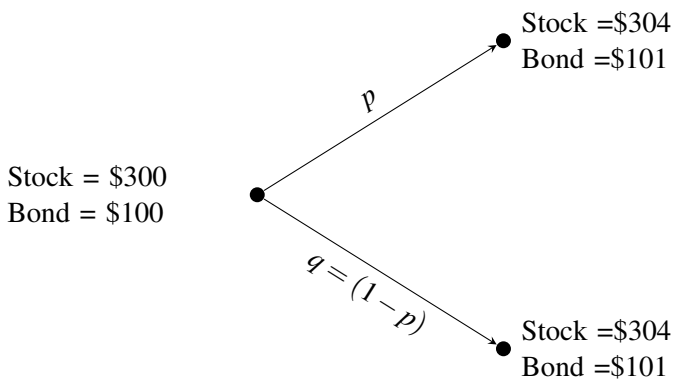
\begin{figure}[!htb]
\caption{One-Period Binomial Tree}
\begin{adjustbox}{center}
\begin{tikzpicture}[>=stealth,sloped]
    \matrix (tree) [%
      matrix of nodes,
      minimum size=1cm,
      column sep=3.5cm,
      row sep=1cm,nodes={text width=8em}
          ]
    {
          &   &  \\
          & {Stock =\$304\\ Bond =\$101} &   \\
     {Stock = \$300\\ Bond = \$100} &   &  \\
          &  {Stock =\$304\\ Bond =\$101} &   \\
          &   &  \\
    };
    \node[bullet,right=0mm of tree-3-1.east](b-3-1){};
    \node[bullet,left=0mm of tree-2-2.west](b-2-2){};
    \node[bullet,left=0mm of tree-4-2.west](b-4-2){};
    \draw[->] (b-3-1) -- (b-2-2) node [midway,above] {$p$};
    \draw[->] (b-3-1) -- (b-4-2) node [midway,below] {$q=(1-p)$};
\end{tikzpicture}
\end{adjustbox}
\end{figure}

Solving the penalty relaxation problem \ref{NstrongPR} using Neos / Baron nonlinear programming (NLP) solver \citep{byrdintegrated} for a set of values for $\beta$ gives the results in Table 17. Using a subgradient method we find the solution to the tight relaxation problem \ref{NstrongPRT} to be $\delta^*_{nst} \approx 0.316$. The corresponding values 
for $\tilde{X}^*$ and $q^*$ are shown as well.
\begin{table}[!htb]
\begin{center}
\caption{Min Distance to Arbitrage-Free Measure}
\begin{tabular}{ |c|c|c|c|c|c|c|c|c|c|c|c| }
 \hline
$\beta$ & 1 & 2 & 4 & 8 & 16 & 32 & 64 & 128 & 256 & 512 & 1024 \\
 \hline
$\delta^*_{nsr}$ & 0.098 & 0.188 & 0.252 & 0.284 & 0.300 & 0.308 & 0.312 & 0.314 & 0.315 & 0.316 & 0.316 \\
 \hline
\end{tabular}
\end{center}
\end{table}
\noindent Calculations show that
\[
p = 
\begin{bmatrix}
300 \\
100
\end{bmatrix}
\quad \wedge \quad X =
\begin{bmatrix} 
304 & 304 \\ 
101 & 101 
\end{bmatrix} 
\quad \implies \quad
\tilde{X}^* = 
\begin{bmatrix} 
303.9 & 303.9 \\ 
101.3 & 101.3 
\end{bmatrix}
\quad \wedge \quad q^* = 
\begin{bmatrix} 
0.493583 \\ 
0.493583 
\end{bmatrix}
\quad \implies \quad
\| p - \tilde{X}^* q^* \|^2 = \num{3.058e-08}
\]
For the complete markets problem, using the Neos / Knitro solver,
\[
p = 
\begin{bmatrix}
300 \\
100
\end{bmatrix}
\quad \wedge \quad X =
\begin{bmatrix} 
304 & 304 \\ 
101 & 101 
\end{bmatrix} 
\quad \implies \quad
\tilde{X}^* = 
\begin{bmatrix} 
304.0 & 303.9 \\ 
101.0 & 101.3 
\end{bmatrix}
\quad \wedge \quad q^* = 
\begin{bmatrix} 
0 \\
0.987167 
\end{bmatrix}
\quad \implies \quad
\| p - \tilde{X}^* q^* \|^2 = \num{8.786e-23}
\]
with minimal distance $\delta^*_{cns} \approx 0.316$. As another example (using the subgradient method) we find that
\[
p = 
\begin{bmatrix}
300 \\
100
\end{bmatrix}
\quad \wedge \quad X =
\begin{bmatrix} 
309 & 306 \\ 
101 & 101 
\end{bmatrix} 
\quad \implies \quad
\tilde{X}^* = 
\begin{bmatrix} 
309 & 305.7 \\ 
101.0 & 101.899 
\end{bmatrix}
\quad \wedge \quad q^* = 
\begin{bmatrix} 
0 \\ 
0.981354 
\end{bmatrix}
\quad \implies \quad
\| p - \tilde{X}^* q^* \|^2 = \num{1.038e-08}
\]
with tight relaxation $\delta^*_{nst} \approx 0.949$. For the complete markets problem, we arrive at essentially the same solution.

\subsubsection{Pairs Trading}
This example uses the Russell 2000 and S\&P 500 indices to conduct pairs trading on an annual data set of month end closing prices from the Yahoo website, as shown in Tables 18 and 19. A plot of this market data is shown in Figure 22. To satisfy the (strong) arbitrage conditions, an initial asset price vector $S_0 = \{ 1,660, 2,750 \}$ is selected. The portfolio $w^* = \{ -1.0, 0.6 \}$ satisfies the (strong) arbitrage condition, for time 1 asset price vector $S_1$ following a uniform discrete distribution with the annual data set as its support. Converting to the nearest NA problem setting of $p = X q$, this support is used as the scenario matrix X and the initial asset price vector $S_0$ is used as the price vector $p$.

\begin{table}[!htb]
\begin{center}
\caption{Russell 2k and S\&P 500 Market Data 2019}
\begin{tabular}{ |c|c|c|c|c|c|c| }
 \hline
Date & 04/01 & 05/01 & 06/01 & 07/01 & 08/01 & 09/01 \\
 \hline
 Russell 2k & 1,591 & 1,466 & 1,567 & 1,577 & 1,495 & 1,523 \\
 \hline
S\&P 500 & 2,946 & 2,752 & 2,942 & 2,980 & 2,926 & 2,977 \\
 \hline
\end{tabular}
\end{center}
\end{table}
\begin{table}[!htb]
\begin{center}
\caption{Russell 2k and S\&P 500 Market Data 2019/2020}
\begin{tabular}{ |c|c|c|c|c|c|c| }
 \hline
Date & 10/01 & 11/01 & 12/01 & 01/01 & 02/01 & 03/01 \\
 \hline
 Russell 2k & 1,562 & 1,625 & 1,668 & 1,614 & 1,476 & 1,153  \\
 \hline
 S\&P 500 & 3,038 & 3,141 & 3,230 & 3,226 & 2,954 & 2,585 \\
 \hline
\end{tabular}
\end{center}
\end{table}

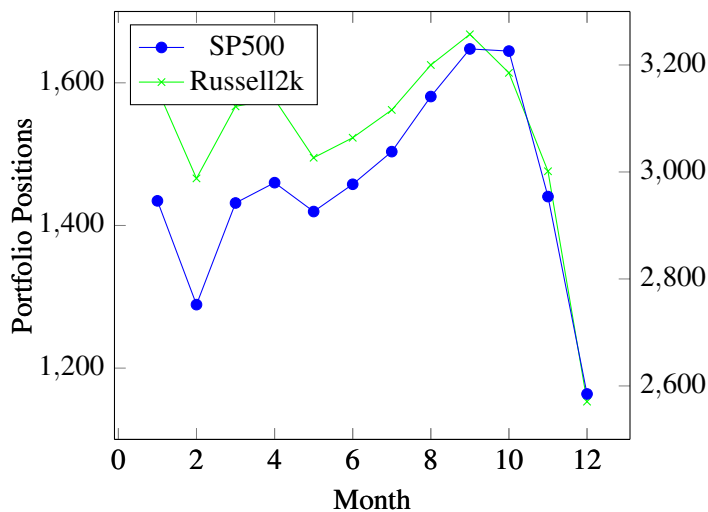
\begin{figure}[!htb]
\caption{Russell 2k and S\&P 500 Market Data}
\begin{center}
\begin{tikzpicture}
	\begin{axis}[legend pos=north west,
		xlabel=Month,
		ylabel=Portfolio Positions,
		ymin = 1100, ymax = 1700,
		axis y line* = left ]
	\addplot[color=green,mark=x] coordinates {
		(1,1591)
		(2,1466)
		(3,1567)
		(4,1577)
		(5,1495)
		(6,1523)
		(7,1562)
		(8,1625)
		(9,1668)
		(10,1614)
		(11,1476)
		(12,1153)
	}; \label{plot12pos_y1}
	\end{axis}

	\begin{axis}[legend pos=north west,
		xlabel=Month,
		ylabel=Portfolio Positions,
		ymin = 2500, ymax = 3300,
		axis y line* = right ]
	\addplot[color=blue,mark=*] coordinates {
		(1,2946)
		(2,2752)
		(3,2942)
		(4,2980)
		(5,2926)
		(6,2977)
		(7,3038)
		(8,3141)
		(9,3230)
		(10,3226)
		(11,2954)
		(12,2585)
	}; \label{plot12pos_y2}
	\addlegendimage{/pgfplots/refstyle=plot12pos_y1}\addlegendentry{SP500}
    \addlegendimage{/pgfplots/refstyle=plot12pos_y2}\addlegendentry{Russell2k}
	\end{axis}
	
\end{tikzpicture}
\end{center}
\end{figure}

Solving the penalty relaxation problem \ref{NstrongPR} using Neos / Knitro nonlinear programming (NLP) solver \citep{byrdintegrated} for a set of values for $\beta$ gives the results in Table 20. Using a subgradient method we find the solution to the tight relaxation problem \ref{NstrongPRT} to be $\delta^*_{nst} \approx 160.36$. The corresponding values for $\tilde{X}^*$ and $q^*$ are shown as well.

\begin{table}[!htb]
\begin{center}
\caption{Min Distance to Arbitrage-Free Measure}
\begin{tabular}{ |c|c|c|c|c|c|c|c|c|c|c|c| }
 \hline
$\beta$ & 1 & 2 & 4 & 8 & 16 & 32 & 64 & 128 & 256 & 512 & 1024 \\
 \hline
$\delta^*_{nsr}$ & 160.09 & 160.23 & 160.29 & 160.33 & 160.35 & 160.35 & 160.36 & 160.36 & 160.36 & 160.36 & 160.36 \\
 \hline
\end{tabular}
\end{center}
\end{table}

\[
p = 
\begin{bmatrix}
1,660 \\
2,750
\end{bmatrix}
\quad \wedge \quad
X = 
\begin{bmatrix*}[r]
1,591 & 1,466 & 1,567 & 1,577 & 1,495 & 1,523 & 1,562 & 1,625 & 1,668 & 1,614 & 1,476 & 1,153 \\
2,946 & 2,752 & 2,942 & 2,980 & 2,926 & 2,977 & 3,038 & 3,141 & 3,230 & 3,226 & 2,954 & 2,585
\end{bmatrix*}
\implies
\]

\[
\tilde{X}^* = 
\begin{bmatrix*}[r]
1,728.29 & 1,466 & 1,567 & 1,577 & 1,495 & 1,523 & 1,562 & 1,625 & 1,668 & 1,614 & 1,476 & 1,153 \\
2,863.13 & 2,752 & 2,942 & 2,980 & 2,926 & 2,977 & 3,038 & 3,141 & 3,230 & 3,226 & 2,954 & 2,585
\end{bmatrix*}
\quad \wedge
\]
\[
q^* = 
\begin{bmatrix}
0.960488 & 0 & 0 & 0 & 0 & 0 & 0 & 0 & 0 & 0 & 0 & 0
\end{bmatrix}
\implies \| p - \tilde{X}^* q^* \|^2 = \num{2.58e-07}.
\]

\subsubsection{Basket Trading}
This example uses the index basket from Section 6.3 to conduct trading. The reference data set is the 2019 month end closing prices from the Yahoo website, as shown in Tables 21 and 22. A plot of this market data is shown in Figure 23. To satisfy the (strong) arbitrage conditions, an initial asset price vector $S_0 = \{ 28256, 3226, 9151, 10.84, 15.27 \}$ is selected. The portfolio $w^* = \{ 0.16, 5.13, -1.61, -179.35, -290.23 \}$ satisfies the (strong) arbitrage condition, for time 1 asset price vector $S_1$ following a uniform discrete distribution with the annual data set as its support. Converting to the nearest NA problem setting of $p = X q$, this support is used as the scenario matrix X and the initial asset price vector $S_0$ is used as the price vector $p$.

\begin{table}[!htb]
\begin{center}
\caption{Index Basket 2019 Market Data}
\begin{tabular}{ |r|r|r|r|r|r|r|r| }
 \hline
Date & 01/01 & 02/01 & 03/01 & 04/01 & 05/01 & 06/01 \\
 \hline
 DJI & 25,000 & 25,916 & 25,929 & 26,593 & 24,815 & 26,600 \\
 \hline
 GSPC & 2,704 & 2,785 & 2,834 & 2,946 & 2,752 & 2,942 \\
 \hline
 IXIC & 7,282 & 7,533 & 7,729 & 8,095 & 7,453 & 8,006 \\
 \hline
 USO & 11.35 & 11.95 & 12.50 & 13.29 & 11.10 & 12.04 \\
 \hline
 SGOL & 12.73 & 12.65 & 12.46 & 12.37 & 12.59 & 13.60 \\
 \hline
\end{tabular}
\end{center}
\end{table}

\begin{table}[!htb]
\begin{center}
\caption{Index Basket 2019 Market Data}
\begin{tabular}{ |r|r|r|r|r|r|r|r| }
 \hline
Date & 07/01 & 08/01 & 09/01 & 10/01 & 11/01 & 12/01 \\
 \hline
 DJI & 26,864 & 26,403 & 26,917 & 27,046 & 28,051 & 28,538 \\
 \hline
 GSPC & 2,980 & 2,926 & 2,977 & 3,038 & 3,141 & 3,231 \\
 \hline
 IXIC & 8,175 & 7,963 & 7,999 & 8,292 & 8,665 & 8,973 \\
 \hline
 USO & 12.04 & 11.46 & 11.34 & 11.30 & 11.62 & 12.81 \\
 \hline
 SGOL & 13.61 & 14.69 & 14.20 & 14.56 & 14.16 & 14.62 \\
 \hline
\end{tabular}
\end{center}
\end{table}

\begin{figure}[H]
	\centering
	\caption{Indices Reference Distribution}%
	\subfloat[Parallel Coords]{\scalebox{0.22}[0.2]{\includegraphics{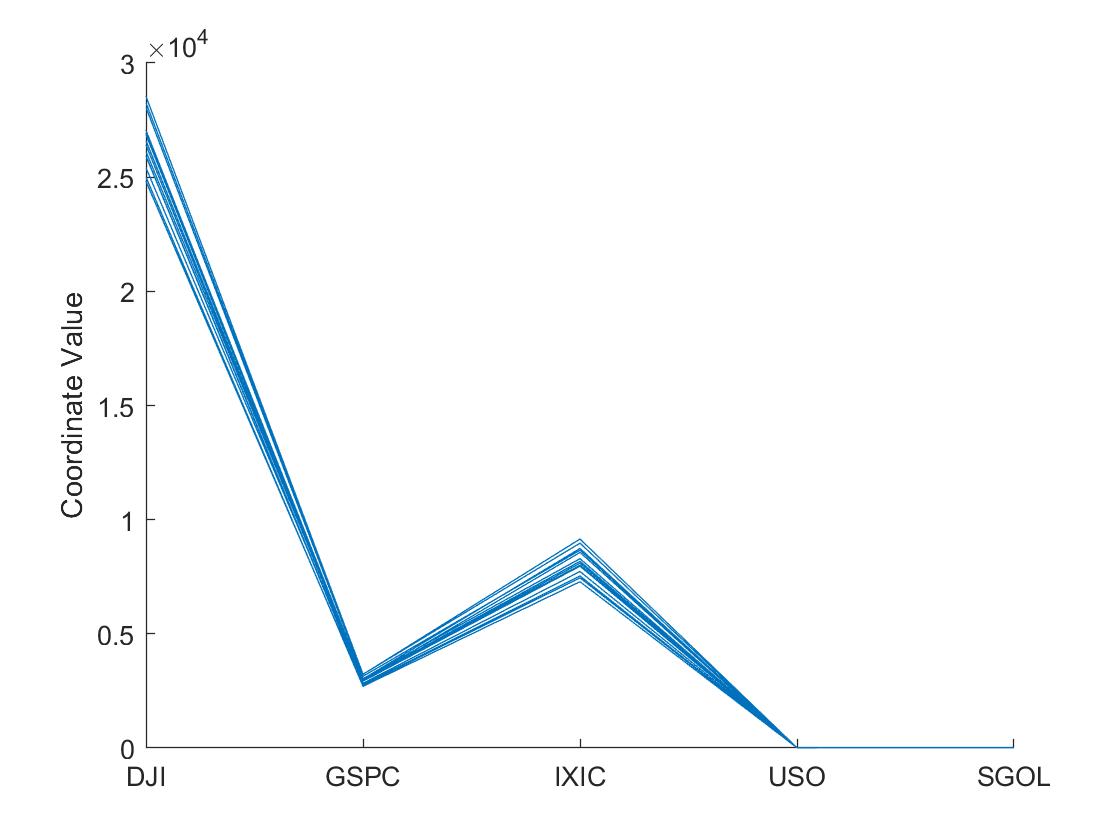}}}%
	\quad
	\subfloat[Quantiles]{\scalebox{0.22}[0.2]{\includegraphics{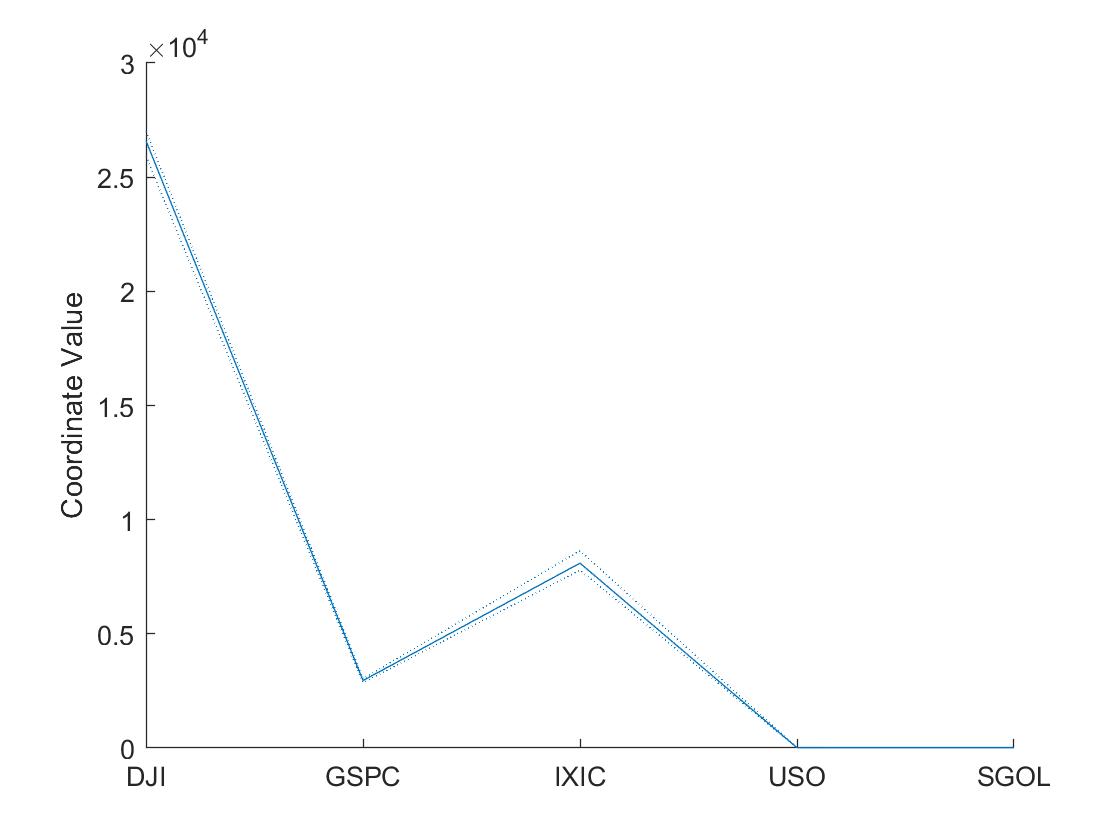}}}%
\end{figure}

Solving the penalty relaxation problem \ref{NstrongPR} using Neos / Knitro nonlinear programming (NLP) solver \citep{byrdintegrated} for a set of values for $\beta$ gives the results in Table 23. Using a subgradient method we find the solution to the tight relaxation problem \ref{NstrongPRT} to be $\delta^*_{nst} \approx 130.07$. The corresponding values for $\tilde{X}^*$ and $q^*$ are shown as well.

\begin{table}[!htb]
\begin{center}
\caption{Min Distance to Arbitrage-Free Measure}
\begin{tabular}{ |c|c|c|c|c|c|c|c|c|c|c|c| }
 \hline
$\beta$ & 1 & 2 & 4 & 8 & 16 & 32 & 64 & 128 & 256 & 512 & 1024 \\
 \hline
$\delta^*_{nsr}$ & 129.57 & 129.82 & 129.95 & 130.01 & 130.04 & 130.06 & 130.07 & 130.07 & 130.07 & 130.07 & 130.07 \\
 \hline
\end{tabular}
\end{center}
\end{table}

\[
p = 
\begin{bmatrix*}[r]
25,000 \\
2,704 \\
7,729 \\
12.50 \\
12.46
\end{bmatrix*}
\quad \wedge 
\]
\[
X = 
\begin{bmatrix*}[r]
25,000	& 25,916	& 25,929 & 26,593 & 24,815	& 26,600 & 26,864  & 264,03	& 26,917	& 27,046 & 28,051 & 28,538 \\
2,704	& 2,785	& 2,834  & 2,946   & 2,752	    & 2,942  & 2,980	 & 2,926	& 2,977	& 3,038  & 3,141	 & 3,230 \\
7,282	& 7,533	& 7,729  & 8,095	  & 7,453	    & 8,006  & 8,175	& 7,963    & 7,999  	& 8,292  & 8,665   &	8,973 \\
11.35	& 11.95	& 12.50  & 13.29	  & 11.10	    & 12.04 & 12.04  & 11.46   & 11.34    & 11.30  & 11.62	 & 12.81 \\
12.73	& 12.65	& 12.46 &	12.37	  & 12.59     & 13.60  &	13.61  & 14.69	& 14.20	& 14.56  &	 14.16  &	14.62
\end{bmatrix*}
\implies
\]

\[
\tilde{X}^* = 
\begin{bmatrix*}[r]
25,000	& 25,918.88	& 25,929 & 26,593 & 24,815	& 26,600 & 26,864  & 264,03	& 26,917	& 27,046 & 28,051 & 28,546.396 \\
2,704	& 2,743.19	& 2,834  & 2,946   & 2,752	    & 2,942  & 2,980	 & 2,926	& 2,977	& 3,038  & 3,141	 & 3,108.249 \\
7,282	& 7,538.30	& 7,729  & 8,095	  & 7,453	    & 8,006  & 8,175	& 7,963    & 7,999  	& 8,292  & 8,665   &	 8,988.435 \\
11.35	& 12.51	& 12.50  & 13.29	  & 11.10	    & 12.04 & 12.04  & 11.46   & 11.34    & 11.30  & 11.62	 & 14.429 \\
12.73	& 12.56	& 12.46 &	12.37	  & 12.59     & 13.60  &	13.61  & 14.69	& 14.20	    & 14.56  &	 14.16  &	14.351
\end{bmatrix*}
\]
\[
 \wedge \quad q^* = 
\begin{bmatrix}
0 & 0.229248 & 0 & 0 & 0 & 0 & 0 & 0 & 0 & 0 & 0 & 0.667620
\end{bmatrix}
\implies \| p - \tilde{X}^* q^* \|^2 = \num{4.78e-07}.
\]

\section{Conclusions and Further Work}

This work has developed theoretical results and investigated calculations of robust arbitrage-free markets under distributional uncertainty using Wasserstein distance as an ambiguity measure. The financial market overview and foundational notation and problem definitions were introduced in Section 1. Using recent duality results \citep {blanchetFirst}, the simpler dual formulation and its mixture of analytic and computational solutions were derived in Section 2. In Section 3 the robust arbitrage methodology was extended to encompass statistical arbitrage. In Section 4, some applications to robust option pricing and portfolio selection were presented. Section 5 gave formal proofs for the NP Hardness of the NA problem. In Section 6, we performed a computational study to calculate the critical radii (for the arbitrage conditions), optimal portfolios, and best (worst) case distributions for some concrete examples. The examples included a simple binomial tree, a pairs trading data set, and two trading baskets. The nearest NA problem was also explored to complete the study. Finally, we conclude with some commentary on directions for further research. \par

One direction for future research, as has been previously discussed in Section 1.4.2, would be to investigate robust arbitrage properties in a multi period \textit{continuous} time setting for a suitable class of admissible trading strategies. Recall that a more general version of the fundamental theorem of asset pricing applies there. Additional detail on this topic can be found in \citet{delbaen2006mathematics}. Another direction for future research, as mentioned in Section 2, would be to develop (and apply) a global solution strategy for the NLP problem formulations of Section 2.1.3. One possibility (as mentioned) is to construct an MINLP problem formulation, in programming languages such as GAMS, that is solvable to \textit{global} optimality using the Baron solver, for example. Perhaps a third direction for future research would be to investigate notions of robust (modern) portfolio theory applying and/or extending the framework developed thus far. \par

\section*{Data Availability Statement}
The raw and/or processed data required to reproduce the findings from this research can be obtained from the corresponding author, [D.S.], upon reasonable request.

\section*{Conflict of Interest Statement}
The authors declare they have no conflict of interest.

\section*{Funding Statement}
The authors received no specific funding for this work.

\clearpage
\bibliographystyle{apalike}
\bibliography{RobustArbFM}

\end{document}